\documentclass[a4paper,12pt,mathserif]{article} \usepackage{amsmath}
\usepackage{amsthm}
\usepackage{amsfonts}
\usepackage{bbm}
\usepackage{booktabs}
\usepackage{mathabx}
\usepackage{enumitem}
\usepackage{mathtools}
\usepackage{multirow}
\usepackage{diagbox}
\usepackage{tikz}
\usepackage{tikz-3dplot}
\usepackage{lineno}
\usepackage{qtree}
\usepackage{newfloat}
\usepackage{setspace}
\usepackage{wasysym}
\usepackage{subcaption}
\usepackage{thmtools,thm-restate}
\usepackage{epstopdf}
\usepackage{xr}
\usepackage[en-US]{datetime2}
\usepackage{geometry}
\usepackage[authoryear]{natbib}
\usepackage[max6]{authblk}
\usepackage{longtable}
\usepackage{cancel}

\usepackage{titlesec}

\titleformat{\subsection}[block]
  {\normalfont\large\scshape\centering}  {}                                       {0pt}                                    {}                                       
\usepackage{palatino}
\usepackage[sc]{mathpazo}
\usepackage{ragged2e} \newcommand{\floatnote}[1]{\vspace{.5\baselineskip}\begin{justify}\footnotesize \textbf{Note:} #1\end{justify}}

\setkeys{Gin}{width=1\textwidth}

\let\oldappendix\appendix\renewcommand{\appendix}{\oldappendix \renewcommand{\thesection}{A.\arabic{section}}}

\usepackage{amsopn}

\DeclareMathOperator*{\argmax}{arg\,max}

\newcommand{\indicator}{\mathbbm{1}}

 \newtheorem{pro}{Proposition}\newtheorem{dfn}{Definition}\newtheorem{thm}{Theorem}\newtheorem{lem}{Lemma}\newtheorem{cor}{Corollary}\newtheorem{ass}{Assumption}\newtheorem{rem}{Remark}\newtheorem{exa}{Example}

\newtheorem*{exa*}{Example}\newtheorem*{rem*}{Remark}\newtheorem*{hyp*}{Hypothesis}

\usepackage{float}
\newfloat{Algorithm}{H}{alg}[section]

\newcounter{program}

\newcounter{foc}

 \usepackage{xparse}

\newif\ifpaperabstract
\paperabstracttrue \DeclareDocumentEnvironment{myabstract}{m m}{
\par
\begingroup
\small
\itshape
\ifpaperabstract
\leftskip2em
\begin{onehalfspacing}
\else
\leftskip0em
\begin{singlespacing}
\fi
\rightskip\leftskip
\ifpaperabstract
\noindent \textbf{Abstract}
\fi
}
{
\ifpaperabstract
\\
\\
\noindent \textbf{JEL Codes:} #1 \\
\noindent \textbf{Keywords:} #2
\fi
\par
\normalsize
\ifpaperabstract
\end{onehalfspacing}
\else
\end{singlespacing}
\fi
\endgroup
}
 
\usepackage{layout}

\newcounter{hypothesis}

\makeatletter
\newcommand{\hyp@buildlabel}{\begingroup
  \count@=\value{hypothesis}\relax
  \advance\count@ by -1\relax
  \toks@={H_0}\loop\ifnum\count@>0
    \toks@=\expandafter{\the\toks@\noexpand^\prime}\advance\count@ by -1\relax
  \repeat
  \xdef\hyp@current{\the\toks@}\endgroup
}

\newenvironment{hypothesis}[1][]{\refstepcounter{hypothesis}\hyp@buildlabel
  \protected@edef\@currentlabel{\noexpand\ensuremath{\hyp@current}}\trivlist
  \item[\hskip\labelsep\bfseries Hypothesis \ensuremath{\mathbf{\hyp@current}}\if\relax\detokenize{#1}\relax\else\ \normalfont(#1)\fi.]\itshape\ignorespaces
}{\endtrivlist
}
\makeatother

\newcounter{rum}
\newenvironment{rum}{\refstepcounter{rum}\align}{\tag{R.\therum}\endalign}

\usepackage{amsmath}

\makeatletter
\g@addto@macro\normalsize{\setlength\abovedisplayskip{1mm}}
\g@addto@macro\normalsize{\setlength\belowdisplayskip{1mm}}
\makeatother

\setenumerate{label=(\roman*), itemsep=0mm, parsep=0mm}
\setitemize{label=(\roman*), noitemsep, nolistsep}
\sffamily
\usetikzlibrary{positioning,chains,fit,shapes,calc,arrows,backgrounds,patterns,matrix,tikzmark}
\DeclareFloatingEnvironment[fileext=lod]{diagram}

\xdefinecolor{orangeone}{RGB}{239,141,89}
\xdefinecolor{greenone}{RGB}{101,165,88}
\definecolor{Black}{RGB}{0,0,0}
\definecolor{Purple}{RGB}{128,0,128}

\definecolor{purple}{RGB}{147,51,234}
\definecolor{red}{RGB}{220,38,38}
\definecolor{blue}{RGB}{37,99,235}
\definecolor{grey}{RGB}{120,120,120}

\mathtoolsset{showonlyrefs}

\graphicspath{{./scarpfigures/}}

\usepackage[pdftex,
            pdfauthor={Stefan Hubner},
            pdftitle={Stefan Hubner: Testing Preference Stability},
            pdfsubject={It's complicated: A Nonparametric Test of Preference Stability between Singles and Couples},
            pdfkeywords={Collective Model, Matching, Random Utility, Stochastic Choice, Partial Identification},
            pdfproducer={pdflatex (texlive)},
            pdfcreator={pdflatex (texlive)}, hidelinks]{hyperref}

\DTMlangsetup{showdayofmonth=false}

\author[1]{Stefan Hubner}
\affil[1]{Department of Economics, University of Bristol\thanks{\href{http://api.hubner.info/api/dl/jmp}{\textit{Click}} to download the most recent version. The author (\href{mailto:stefan.hubner@bristol.ac.uk}{stefan.hubner@bristol.ac.uk}) would like to thank Arthur van Soest, Laurens Cherchye, Frederic Vermeulen, Stefan Hoderlein, Ian Crawford, Yuichi Kitamura, J\"org Stoye, John Quah, Martin Browning, Andrew Chesher, Monica Costa--Dias, Sami Stouli, and Pietro Spini for their helpful comments.}}

\makeatletter
\renewcommand*{\@fnsymbol}[1]{\ensuremath{\ifcase#1\or *\or \mathsection\or
\dagger\or \ddagger\or \mathparagraph\or \|\or **\or \dagger\dagger
\or \ddagger\ddagger \else\@ctrerr\fi}}
\makeatother

\newcounter{result}

\title{\vspace{-\baselineskip}\href{https://api.hubner.info/api/dl/jmp}{It's complicated: A Non--parametric Test of Preference Stability between Singles and Couples}}

\newcommand{\optionaltilde}[1]{#1} \newif\ifposter
\posterfalse
\makeatletter
\begin{document}

\onehalfspacing
\modulolinenumbers[2]
\maketitle
\vspace{-1\baselineskip}
\begin{myabstract}
{D12, D13, J12}
{Collective model, Preference Stability, Collective Axiom of Revealed Preference, Stochastic Choice, Random Utility, Matching}
This paper develops a method to use singles' data in a non-parametric revealed preference setting of collective household choice.
We use it to test the controversial assumption of preference stability between singles and couples, without data on intra-household allocation or marital transitions.
We show that, under the preference-stability hypothesis, consumption choices from an endogenously matched population admit a conditional random-utility representation over counterfactual pairings of couples and singles. 
Preference stability is testable as a feasibility restriction on the observed marginal choice distributions.
We reject the hypothesis using consumption data from the Dutch LISS, the Russian RLMS, and the Spanish ECPF panels.
\end{myabstract}

 \section{Introduction}
Measuring poverty levels, quantifying the effects of socio-economic policies on individuals, and understanding the mechanisms of individual decision making are pivotal challenges for economists and policymakers.  
Most of the relevant datasets, however, do not feature granular enough information to meet these challenges because a majority of individuals live in collective units, such as households or families.\footnote{
The collective model is the workhorse model in family economics with a long tradition dating back to \citet{Becker1965,Becker1981}, \citet{Gorman1976}, \citet{AppsRees1988}, \citet{Browning1994}, \citet{Browning1998}, \citet{Chiappori2006}, \citet{Chiappori2009}, \citet{Chiappori2012}.}
To open this black box, without observing information about resource sharing within the household, economists 
make two prevalent assumptions.\footnote{Cf. the seminal work of: \citet{BCL2013}, \citet{Lewbel2008}, \citet{LewbelLin2022}, \citet{LewbelPendakur2022,Lewbel2026} for identification of resource shares or equivalence scales based on single-person households, \citet{Mazzocco2013}, \citet{voena2015}, \citet{Gayle2016}, \citet{TheloudisEtAl2025}, \citet{LowMeghirPistaferriVoena2018} for identification in the context of inter-temporal models, and \citet{Chiappori2009a}, and \citet{Chiappori2020} for a survey on endogenous marriage market matching models.}
First, an individual's preferences do not depend on whether they are in a relationship or not.
Second, preference heterogeneity is largely described by two types: men and women. 

In this paper, we construct a test of the former and thoroughly relax the latter. 
Taking the information available in typical datasets as given, we will do so without observing any transitions between relationship states, i.e. marriage or divorce, and without observing more than aggregate household-level consumption choices. 
The test is fully non-parametric and allows for a heterogeneous population in order to avoid testing auxiliary restrictions.
Preference homogeneity is particularly restrictive in a model of collective decision making since it not only requires every individual to have the same preferences, but also assumes that any two individuals matched as a couple would arrive at the exact same sharing of resources.

In the presence of unobserved heterogeneity, stability requires that the individual preferences of partnered and unpartnered individuals are drawn from the same distribution.
The difficulty is that this distribution, as well as its realisations, is not only unobserved to the econometrician, it is also an equilibrium quantity that arises through matching.
Even if individual preferences remain unchanged upon entering or leaving a relationship, systematic differences between the single and partnered subpopulations can still arise through sorting into partnership.
Thus, we derive testable implications for the observed equilibrium marginal distributions of preference-induced demands of couples, single men, and single women.
Preference stability restricts how these marginal distributions can be jointly rationalised: there must be enough mass of each preference type among unpartnered individuals to match the preference composition of partnered individuals.

How can we test this when we only observe marginal distributions? 
We introduce \emph{configurations}: each configuration takes one couple and assigns to each spouse a counterfactual from the corresponding single population. It is \emph{preference-stable} when each spouse and their counterfactual have identical ordinal preferences.
We link the structural preference-stability restrictions on latent configurations to observed behaviour by mapping each configuration into household demand: the chosen bundles of couples, single men, and single women are the empirical objects through which the restriction is tested.
We show that, if the population has stable preferences, then the observed marginal demand distributions are compatible with some mixture over preference-stable configurations only.

To test the restriction on demand distributions whilst only observing one fixed, matched population, we introduce an auxiliary sampling device that generates hypothetical configurations through swapping individuals.
Formally, these swaps are permutations acting on the matching allocation and the corresponding household utilities.
Importantly, we show that the induced demands become conditionally i.i.d. under a mild \emph{anonymity} restriction on the matching mechanism and a weak dependence assumption about structural preferences: \emph{exchangeability}.
Based on this, we develop a test statistic and show that the large sample theory of \citet{KitamuraStoye2013}, the established standard for random utility models, can be applied.

We do not observe preferences or utilities directly, but rather the corresponding optimal demand in the form of continuous consumption bundles. 
Although we can recover preferences from continuous demand functions for a sufficient number of budgets, there would be an exploding number of configurations to consider, making any permutation test computationally challenging, even for small samples.
Thus, we propose to classify households into discrete types.
We define a \emph{single type} based on the equivalence relation induced by the generalised axiom of revealed preferences (GARP; \citealp{Afriat1967} and \citealp{Varian1982}), and, similarly, a \emph{couple type} based on the collective axiom of revealed preferences (CARP; \citealp{Cherchye2007,Cherchye2009,Cherchye2011}).\footnote{This is not restrictive, because any two household types are observationally equivalent if they are not distinguishable in terms of their preferences without an additional functional form restriction.}
Their inherent compatibility makes these axioms an effective modelling choice for our setup.
In order to identify revealed preference types which we combine into discrete configuration types, we make use of (short) panel data.

We apply our test to three popular datasets: the \emph{Dutch Longitudinal Internet Studies for the Social Sciences (LISS)}, the \emph{Russian Longitudinal Monitoring Survey (RLMS)}, and the \emph{Spanish Continuous Family Expenditure Survey (Encuesta Continua de Presupuestos Familiares, ECPF)} used by \citet{Cherchye2012}, \citet{Cherchye2011}, and \citet{Adams2014}, respectively, in the context of the collective model.
We consistently reject the hypothesis of preference stability across these datasets and different specifications.

The approach we develop in this paper can be contrasted with the literature on testing preference restrictions in a continuous setting, which is typically based on the Slutsky matrix and, thus, requires estimation of household demands and their derivatives. 
In their seminal work, \citet{Browning1998} construct a test of collective rationality based on a parametric almost ideal demand system with additive measurement errors. 
Similarly, \citet{brugler2016} estimates a parametric quadratic ideal demand system \citep{Banks1997} in a setting without preference heterogeneity and compares the parameter estimates for single men, single women and couples to draw conclusions about preference stability.
While almost ideal demand systems provide a flexible parametric form allowing for easy testing of parameters, both the potential for misspecification and the restrictions imposed on preferences to ensure additive separability of the errors, are problematic.
To allow for non-separability, and thus a larger class of preferences, \citet{Hubner2015} develops a collective random utility model and derives conditions for non-parametric identification of random utility and Pareto weights by showing global invertibility of demands, under the assumption of observed private consumption.  
Further, under the preference stability assumption, \citet{LewbelLin2022} show identification of a semi-parametric model with heterogeneous structural preferences and a general functional form assumption. \citet{botosarumurispendakur2023}, \citet{HsiehLewbelPendakur2024} incorporate explicit preference and sharing heterogeneity into collective demand systems, and \citet{ChiapporiMeghirOkuyama2025} estimate dynamic collective models allowing for unobserved preference heterogeneity and evolution across life-cycle stages.
While part of the literature has departed from the preference stability assumption in favour of functional form restrictions, such as \citet{DLP2013,DLP2021}, \citet{Lechene2019}, \citet{Calvi2020} who use preference similarity, or \citet{SokulluValente2022} who use a panel, combining singles data with couples data provides a strong form of identification, particularly in a non-parametric setting.

The use of singles data in the context of the revealed preference characterisation of the collective model \citep{Cherchye2007, Cherchye2009} is novel.
The advantage of a revealed preference based approach over the continuous approaches outlined above is the option of modelling unobserved preference heterogeneity without requiring global invertibility of demands.
This bypasses the need for ad-hoc functional form assumptions in favour of testable choice-based restrictions.
Stochastic revealed preference settings have been studied in the context of the unitary consumption model. 
\citet{Hoderlein2014} consider the weak axiom of revealed preference in the unitary model.
Observing the same population in different price regimes, as repeated cross-sections, they use copula bounds (Frechet-Hoeffding) on the probability that the population violates the weak axiom of revealed preferences.
\citet{KitamuraStoye2013,Deb2017} integrate this approach into the stochastic choice framework of \citet{McFadden1991} and \citet{McFadden2005} by partitioning budget sets into patches using the strong axiom of revealed preference. 

While, conceptually, our approach is very different, the final test statistic is closely related to theirs. 
We show that the large sample theory of \citet{KitamuraStoye2013} applies to our theory which extends random utility to also incorporate random matching through exchangeable configurations.
There is a range of recent contributions targeting the computational complexity of this class of problems, most prominently \citet{SmeuldersCherchyeDeRock2021}, \citet{AguiarKashaev2021}, \citet{KoidaShirai2024}, \citet{Turansick2025}.
We contribute to this literature by introducing a fast, parallel non-negative least squares algorithm which leverages the sparsity of the problem.\footnote{Haskell code is available as a GitHub repository and a simulation study to evaluate speed, finite sample size, and power of the test statistic can be found in Appendix \ref{sec:mc}.} 

We proceed as follows. Section \ref{sec:matching} introduces a theoretical device that allows us to define utilities under counterfactual assignments which we use to define configurations. Section \ref{sec:model} defines heterogeneous preferences as random variables, and their dependence in the population. It proceeds by defining latent configurations and develops the necessary theory relating them to observed continuous choices. Falsifiability of preference stability is discussed through the lens of the Slutsky matrix. Section \ref{sec:inference} then operationalises the theory by introducing the discrete characterisation of choices and configurations, based on revealed preference restrictions on observed demands. It develops a test statistic based on a reformulation of the restrictions as a semi-definite, quadratic programme.
Section \ref{sec:results} provides empirical results and discusses robustness and extensions, including endogeneity of total expenditures, public goods, an alternative characterisation of collective households, and a heterogeneity analysis. All proofs can be found in Appendix \ref{app:proofs}.

 \section{Factual and Counterfactual Assignment}\label{sec:matching}
The aim of this section is to put singles and couples into one common framework. This is needed because singles are observed on their own, while couples are observed only through joint household choices. Once both are described in the same way, we can compare an observed couple with the counterfactual couple formed by replacing both partners with singles. 

\subsection{A Four-Individual Population}

We begin with the smallest economy that contains the comparison at the heart of the paper.
It has only four individuals which we name: Apollo, Athena, Zeus, and Hera.
Let Apollo be \emph{the} single man and Athena be \emph{the} single woman.
Both are unitary households.
Let Zeus assume the role of \emph{the} married man and Hera the role of \emph{the} married woman. They also form a household together.

Ruling out the trivial case, in which they all share the same preferences, resulting in all three households being unitary, there are two remaining scenarios.
First, everyone has their own distinct preferences.
In this case we can think of Zeus \& Hera as a collective household.
Second, Apollo and Zeus as well as Athena and Hera share the same set of preferences, respectively, but the two sets could differ.
This distinction is what we study in this paper.

In a setting where each individual derives utility from consumption, the difficulty arises from what is observed in standard datasets.
For Apollo and Athena, we observe their respective household consumption, which, under standard conditions and sufficient price variation, allows us to directly recover their utilities up to an ordinal transformation \citep{Hurwicz1971}.
By contrast, we only observe Hera and Zeus as a couple, so we only know that their household's consumption maximised their collective utility.
However, both Hera's and Zeus' individual utilities remain unidentified \citep{Chiappori2006,Chiappori2009}.
Hence, in order to test preference stability, we cannot directly compare Zeus with Apollo or Hera with Athena.

Instead, we proceed indirectly. We start from the observed couple, Zeus and Hera, and use their choices across price regimes to verify that their behaviour is consistent with collective rationality.
Preference stability demands that moving from singlehood to partnership can only rescale the cardinality of utilities.
We therefore ask whether Zeus and Hera's observed choices would remain collectively rational if their preferences were replaced by the ones revealed by Apollo and Athena, respectively. 
If this breaks collective rationality, the four-tuple described by the two exchanges (Zeus, Hera, Apollo, Athena), which we call a \emph{configuration}, is not preference-stable, since at least one spouse's ordinal preferences must differ.

\subsection{A Countable Population}\label{subs:extended}
As we move to the general case, the economy consists of heterogeneous men and women, and each couple is a one-to-one match between any two individuals from opposite sides of the matching market. 
We represent the population by an assignment graph which we later permute to construct configurations. 
To achieve this, the assignment graph must not only connect partnered individuals but also singles. 
We therefore define an extended population with an assignment matrix in which each single is paired with a distinct dummy partner.\footnote{This is related in spirit to the dummy types used in matching models to absorb unmatched mass. Here, because the problem is formulated as an assignment at the individual level, each single requires a separate dummy counterpart.} 
This representation lets us describe preferences and choices for both observed and counterfactual assignments. Corollary \ref{pro:unitaryrep} later shows that any unitary utility admits such a representation.

Let the population of individuals on each side of the market be indexed by $ \mathbb{N} $ and partitioned into three countable subsequences denoted $ \mathbb{N}_2 $, $\mathbb{N}_1 $, $ \mathbb{N}_0 $.
They represent individuals currently living in a couple, single individuals, and dummies, respectively.\footnote{Formally: $ \mathbb{N}_k = \{ n \in \mathbb{N} : n \text{ mod } 3 = k \} $, rearranged to form blocks as in Figure \ref{fig:partition} (r.h.s.).}
Without loss of generality we arrange the data, leading with couples $ \mathbb{N}_2 $, followed by singles $ \mathbb{N}_1 $, and dummies $ \mathbb{N}_0 $. 
We represent matches as the matrix $ {M} = \{ {m}_{ij}\}_{i,j \in \mathbb{N}}  $ with entry $ {m}_{ij} = 1 $ if man $ i $ is matched with woman $ j $, and $ 0 $ otherwise.
Visualised in Figure \ref{fig:partition} (l.h.s.), couples  (\color{purple}purple\color{black}) form the leading block $ \mathbb{N}_2 \times \mathbb{N}_2 $.
Each single is matched with exactly one dummy individual from the other side: $ {\mathbb{N}_1 \times \mathbb{N}_0 } $ for the single men (\color{blue}blue\color{black}), and $ \mathbb{N}_0 \times \mathbb{N}_1 $ for single women (\color{red}red\color{black}). 

\begin{figure}[h!]
  \begin{center}

\caption{Block structure of $\mathbb{N} \times \mathbb{N}$ extended assignment matrix (l.h.s.) with finite sample counterparts (r.h.s.).}

\begin{minipage}[t]{0.52\textwidth}
  \vspace{0pt}
\resizebox{\textwidth}{!}{\begin{tikzpicture}[
    label/.style={font=\large\bfseries}
]

\definecolor{purple}{RGB}{147,51,234}
\definecolor{red}{RGB}{220,38,38}
\definecolor{blue}{RGB}{37,99,235}
\definecolor{grey}{RGB}{120,120,120}

\def\cellsize{3.5}
\def\spacing{0.1}

\foreach \i in {0,1,2} {
    \foreach \j in {0,1,2} {
        \coordinate (c\i\j) at ({\j*(\cellsize+\spacing)}, {-\i*(\cellsize+\spacing)});
    }
}

\draw[purple, ultra thick] (c00) rectangle +(\cellsize,\cellsize);
\node[align=center] at ($(c00)+(\cellsize/2,\cellsize/2)$) { Factual \\ couples\\ Zeus \& Hera};

\draw[purple, very thick, dashed] (c01) rectangle +(\cellsize,\cellsize);
\node[align=center] at ($(c01)+(\cellsize/2,\cellsize/2)$) { Counterfactual \\  couples  ($\tau^f$) \\ Zeus \& Athena };

\draw[blue, very thick, dashed] (c02) rectangle +(\cellsize,\cellsize);
\node[align=center] at ($(c02)+(\cellsize/2,\cellsize/2)$) { Counterfactual \\  single men  ($\tau^m$) \\ Zeus   };

\draw[purple, very thick, dashed] (c10) rectangle +(\cellsize,\cellsize);
\node[align=center] at ($(c10)+(\cellsize/2,\cellsize/2)$) { Counterfactual \\  couples  ($\tau^m$) \\ Apollo \& Hera };

\draw[purple, very thick, dashed] (c11) rectangle +(\cellsize,\cellsize);
\node[align=center] at ($(c11)+(\cellsize/2,\cellsize/2)$) { Counterfactual \\  couples  ($\tau^m, \tau^f $ ) \\ Apollo \& Athena };

\draw[blue, very thick] (c12) rectangle +(\cellsize,\cellsize);
\node[align=center] at ($(c12)+(\cellsize/2,\cellsize/2)$) { Factual \\  single men \\ Apollo };

\draw[red, very thick, dashed] (c20) rectangle +(\cellsize,\cellsize);
\node[align=center] at ($(c20)+(\cellsize/2,\cellsize/2)$) { Counterfactual \\  s. women ($\tau^f$) \\ Hera };

\draw[red, very thick] (c21) rectangle +(\cellsize,\cellsize);
\node[align=center] at ($(c21)+(\cellsize/2,\cellsize/2)$) { Factual \\  single women \\ Athena };

\draw[grey, very thick, dotted] (c22) rectangle +(\cellsize,\cellsize);
\node[align=center, text=grey] at ($(c22)+(\cellsize/2,\cellsize/2)$) { dummy \\  couples \\ };

\node[label, left=0.3cm of c00, anchor=east] at ($(c00)+(0,\cellsize/2)$) {$\mathbb{N}_2$};
\node[label, left=0.3cm of c10, anchor=east] at ($(c10)+(0,\cellsize/2)$) {$\mathbb{N}_1$};
\node[label, left=0.3cm of c20, anchor=east] at ($(c20)+(0,\cellsize/2)$) {$\mathbb{N}_0$};

\node[label, above=0.3cm of c00, anchor=south] at ($(c00)+(\cellsize/2,\cellsize)$) {$\mathbb{N}_2$};
\node[label, above=0.3cm of c01, anchor=south] at ($(c01)+(\cellsize/2,\cellsize)$) {$\mathbb{N}_1$};
\node[label, above=0.3cm of c02, anchor=south] at ($(c02)+(\cellsize/2,\cellsize)$) {$\mathbb{N}_0$};

\end{tikzpicture}

 }
\end{minipage}\hfill
\begin{minipage}[t]{0.47\textwidth}
  \vspace{20pt}
\resizebox{\textwidth}{!}{

\begin{tabular}{@{}cl|cc|cc|ccc@{}}
\toprule
& f & \multicolumn{2}{c}{\footnotesize{$\mathbb{N}_2$}} & \multicolumn{2}{c}{\footnotesize{$\mathbb{N}_1$}} & \multicolumn{3}{c}{\footnotesize{$\mathbb{N}_0$}} \\
m &  & {\scriptsize 2} & {\scriptsize 5} & {\scriptsize 1} & {\scriptsize 4} & {\scriptsize 0} & {\scriptsize 3} & {\scriptsize 6} \\
\midrule
\multirow{2}{*}{\footnotesize$\mathbb{N}_2$} 
  & \scriptsize{2}  & \tikzmark{tl}\color{red}{\xcancel{1}} &   &   &   & \color{blue}{1}  &   &   \\
  & \scriptsize{5}  &   & 1 &   &   &   &   &   \\
\cmidrule{2-9}
\multirow{3}{*}{\footnotesize$\mathbb{N}_1$} 
  & \scriptsize{1}  &   &   & \color{blue}{1}  &   & \color{red}{\xcancel{1}} &   &  \\
  & \scriptsize{4}  &   &   &   &   &   & {1} &  \\
  & \scriptsize{7}  &   &   &   & \tikzmark{br}  &   &   & 1 \\
\cmidrule{2-9}
\multirow{2}{*}{\footnotesize$\mathbb{N}_0$} 
  & \scriptsize{0}  & \color{blue}{1} &  & \color{red}{\xcancel{1}} &  &   &   &   \\
  & \scriptsize{3}  &   &   &   & 1 &   &   &   \\
\bottomrule
\end{tabular}
\begin{tikzpicture}[overlay,remember picture]
\draw[black, thick, dashed] ([shift={(-0.7ex,1.95ex)}]pic cs:tl) rectangle ([shift={(1.25ex,-0.75ex)}]pic cs:br);
\end{tikzpicture}
 }
\end{minipage}

\label{fig:partition}
\end{center}
\footnotesize \textbf{Note:} The figure depicts the block structure of the population representation. Men are represented by rows, women by columns. On the left hand side, each solid box represents a factual population block, dotted blocks are counterfactuals. The right hand side shows a population with two couples, three single men, and two single women before and after the transpositions $ \tau^f = \tau^m = (1\; 2) $. The removed links are \color{red}{red} \color{black} and the new ones \color{blue}{blue}\color{black}. Their respective unions with the unaffected black links make up the factual and counterfactual assignment, respectively (both are permutations). The dotted box shows the assignment matrix without dummies (not a permutation). Example \ref{exa:transposition} in Appendix \ref{app:examples} describes the four individual scenario from above within this framework.
\end{figure}

Due to the artificial matches between singles and dummies, each individual is matched \emph{exactly once}.
As a consequence, every observed assignment can be represented as a permutation $ \sigma $, a one-to-one map from $ \mathbb{N} $ to itself, and we can write assignments as $ m_{ij} = \delta_{j,\sigma(i)} $.\footnote{The function $\delta_{ij}$ is the Kronecker delta which takes the value one if $ i = j $ and zero otherwise.}
This allows us to write any counterfactual assignment as a function composition with any permutation of the original assignment (see e.g., \citealp[Chapter 3]{rotman1994}).
Here, the only counterfactual assignments with empirical content are exchanges of partnered individuals from $ \mathbb{N}_2 $ with singles from $ \mathbb{N}_1 $. 
\begin{dfn}[Transposition]
Let $ S_\infty $ be the set of permutations over $  \mathbb{N} $. 
A transposition $(i\ i') \in S_\infty$ is the permutation that swaps $i$ and $i'$ and fixes all other indices. Let $ \mathcal{T} $ be the set of any pairs of transpositions defining the exchanges between partnered individuals $i \in \mathbb{N}_2$ and singles $i' \in \mathbb{N}_1$.
\end{dfn}

This formalises the partner replacement from the thought experiment above as a \emph{joint action} on the extended assignment matrix: $ (\tau^m, \tau^f) \cdot M \equiv \{m_{\tau^m(i), \tau^f(j)}\}_{i,j\in\mathbb{N}} $ for transpositions $ (\tau^m, \tau^f) \in \mathcal{T} $.\footnote{An action in our setting is a map over the indices of a matrix $ \varsigma \cdot ij $ satisfying (i) identity: $ \text{id} \cdot ij = ij $ and (ii) compatibility: $ \sigma \cdot (\varsigma \cdot ij) = (\sigma \circ \varsigma) \cdot ij $  for $ \text{id},\sigma,\varsigma \in S_{\infty} $, see \citet[Chapter 9]{rotman1994}.}
This is a simultaneous \emph{row-column exchange} which leads to the new permutation $ \tau^f \circ \sigma \circ \tau^m $.\footnote{This is, because $ m_{ij}' = m_{\tau^m(i), \tau^f(j)} = \delta_{\tau^f(j), \sigma(\tau^m(i))} = \delta_{j, (\tau^f)^{-1}(\sigma(\tau^m(i)))} = m_{i, (\tau^f \circ \sigma \circ \tau^m)(i)} $.}
This can be read (right to left) as: for any man, $ \tau^m $ finds which position he occupies after the swap, $ \sigma $ then looks up who that position was originally matched to, a woman, to whom we apply $ \tau^f $. 

This row-column exchange is a relabelling of the edges in the matching graph. 
It serves two roles. First, it allows us to represent any counterfactual assignment structure. Each element represents a configuration. 
Second, it allows us to define preferences for both factual and counterfactual couples.
We require the following assumption about matching.
\begin{ass}[Anonymity]
\label{ass:matching} 
Let $ M $ be the outcome of a matching allocation mechanism based on some couple-specific variables $ Z = \{z_{ij}\}_{i,j \in \mathbb{N}} $, possibly unobserved to the econometrician. 
For any permutations $ \varsigma^m, \varsigma^f \in S_{\infty} $ we have
$ M( (\varsigma^m, \varsigma^f) \cdot Z) = (\varsigma^m, \varsigma^f) \cdot M(Z) $.
\end{ass}
This is a permutation-equivariance assumption and restricts matching so that it is not based on identities.
It states that if Zeus and Apollo exchanged characteristics, then Hera would have matched with Apollo instead of Zeus.
In other words, individuals only care about the characteristics of their partner, not their identity.
A violation would undo any forthcoming assumption that restricts dependence between individual's preferences. 
We show in Appendix \ref{lem:matching} that the solution to the finite assignment problem  $ M(z) = \argmax_{\sigma \in S_n} \sum_{i,j \in \mathbb{N}} \delta_{j\sigma(i)} \Phi(z_{ij}) $ \citep{Galichon2021} where matching is based on joint surplus $ \Phi(z) $, as in \citet{Becker1973}, satisfies this property.

 \section{Preference Stability in the Population}\label{sec:model}

In the previous section we introduced a language to define factual and counterfactual households purely from a matching perspective.
In this section, we introduce preferences for both factual and counterfactual households and formulate the economic content of the preference stability hypothesis.
We argue that, even if all individuals draw preferences from the same environment, equilibrium matching can sort them into single and partnered subpopulations which we call \emph{pools}.
Preference stability is therefore not a statement about equality household by household, but about the composition of preferences in these pools.
We show how this restriction on latent distributions can be translated into one based on observable demand distributions, through a random utility and matching representation.

\subsection{Model}\label{subs:primitives}
We now define our random utility primitives, based on an environment of unobserved preferences. We describe the class of utility models we consider and give a representation of household utility for both singles and couples.

\begin{ass}[Unobserved Heterogeneity]
  \label{ass:polish}
Let the environment of unobserved ordinal preferences be represented by a random element $ \omega \equiv (\omega^m, \omega^f) \in \Omega \equiv \mathbb{X}^\mathbb{N} \times \mathbb{X}^\mathbb{N} $ where $ \mathbb{X} $ is a Polish space.
Let a couple $ (i, j)$'s preferences be the element $ \omega_{ij} = (\omega^m_i, \omega^f_j) $.
For dummy individuals $ i, j \in \mathbb{N}_0 $, we normalise preferences to $ \omega^m_i = \omega^f_j = \omega^0 $.
\end{ass}

To accommodate any utility representation, we choose $ \mathbb{X} $ to be Polish, a space general enough for this task, while structured enough so that all relevant results for random variables carry over to this space, if endowed with the standard Borel sigma algebra $\mathcal{B}$.\footnote{Indeed, a Polish space is a \emph{separable}, \emph{complete}, \emph{metricable} space that induces the standard Borel sigma algebra, allowing standard conditioning statements \citep[Thm 5.3]{Kallenberg1997}, weak convergence of measures defined on it \citep[Thm 14.3 \& Thm 14.5]{Kallenberg1997}, and \citet{definetti1931} type representation theorems which we require below \citep{HewittSavage1955}. The product space $ \mathbb{X}^3 $ is, by definition, also a Polish space \citep[Lemma 1.2]{Kallenberg1997}} We let $ \mu $ be the corresponding probability measure.

The sample space $ \Omega $ allows for dependence across individuals.
We require this, because we only observe a sample from a \emph{fixed, endogenously matched} population.
Thus, we \emph{cannot} rely on independent sampling of $ \omega_{ij} $ from some common distribution. 
What matters for our purposes is a weaker form of independence. In particular, a symmetry assumption that requires that before matching, names carry no intrinsic economic information. Assumption \ref{ass:dependence} formalises this.\footnote{Indeed, since our test is based on choice frequencies, which we get from pushing the random preferences defined in Assumption \ref{ass:polish}, through the household's optimal choice rule, we \emph{need not} consider all possible events in $ \mathcal{B}(\Omega) $, it is sufficient to only consider $ \mathcal{I}_0 $ based on symmetric subsets of $ \Omega $.}

\begin{ass}[Exchangeability of the Preference Environment]
  \label{ass:dependence}
  Finite relabellings of real, non-dummy individuals $ \varsigma^m, \varsigma^f \in G_0 \equiv \{ \varsigma \in S_{\infty} : \varsigma(\mathbb{N}_0) = \mathbb{N}_0 \} $ preserve the joint distribution of household preferences.
  Formally, $(\omega^m, \omega^f ) \stackrel{d}{=} (\varsigma^m \cdot \omega^m, \varsigma^f \cdot \omega^f) $.\footnote{The symbol $ \stackrel{d}{=} $ refers to \emph{equality in distribution}: $ X \stackrel{d}{=} Y $ iff $ P(X \in A) = P(Y \in A) $ for all $ A $.} 
\end{ass}

In Assumption \ref{ass:polish} we impose a normalisation to remove economically irrelevant heterogeneity due to the arbitrary assignment of single individuals to dummy indices.
Consequently, we only look at permutations $ G_0 $ that leave those unaffected.
Assumption \ref{ass:dependence} tells us that being called Hera or Athena, respectively Zeus or Apollo, has no bearing on the realisation of preferences.
Thus, ex-ante, prior to matching, this immediately permits a \citet{definetti1931} interpretation of non-dummy individual preferences in which we may think of nature drawing a ``population-level'' distribution for the individual unobserved preferences $ \bar{\mu}_0 \equiv (\bar{\mu}_0^m, \bar{\mu}_0^f) $.\footnote{Measure $ \bar{\mu}_0 $, obtained by conditioning on $ {\mathcal I}_0 \equiv \sigma({\mathcal I}_0^m, {\mathcal I}_0^f) $, can distinguish only symmetric events.}
Consider the three distinct populations from the empirical section: Netherlands, Spain, and Russia.
What this says is that, for each of them, nature first draws the population-level composition of preferences, which we may think of as the institutional and cultural environment that leads to the formation of preferences. 
Then, conditional on these compositions, each individual draws i.i.d. preferences $ \omega^m_i $ and $ \omega^f_j $ from $ \bar{\mu}_0^m $ and $ \bar{\mu}_0^f $, respectively.
Unconditionally, preferences need not be independent, since all agents belong to the same realised market, making it a weaker requirement than unconditional independence.\footnote{One example where this may fail is, if the name entails information which is not otherwise accounted for, e.g. belonging to a certain local matching market. To mitigate this, in the empirical section, we also consider hierarchical specifications using observed demographics.} 

With this, we can now  relate the unobserved heterogeneity, defined as random variables in Assumptions \ref{ass:polish}-\ref{ass:dependence}, to structural, behaviour-relevant primitives. 

\begin{ass}[Random Collective Utility]\label{ass:collective}
  The following holds for all $ i, j \in \mathbb{N} $:
  \vspace{-.5\baselineskip}
  \begin{enumerate}
    \item Individual utilities \label{ass:ego} $ u^m_{i} \equiv u(\omega^m_i, \cdot ) $ and $ u^f_{j} \equiv u(\omega^f_j,  \cdot) $ are twice continuously-differentiable, strictly quasi-concave utility functions defined on $ \mathcal{X} \subseteq \mathbb{R}_+^L $, i.e. bundles of $ L $ continuous consumption goods.\footnote{For the exposition with continuous choice, we make the innocuous injectivity assumption that any two $ \omega_1^m \neq \omega_2^m $ induce different first order conditions (and similar for $ f$).} 
    \item \label{ass:efficiency} The aggregation rule $ (u^m_{i}, u^f_{j}) \mapsto \Phi_{ij} $ leads to Pareto-efficient outcomes.
    \item \label{ass:timehom} Individual utilities only depend on own-good consumption and are time-homogeneous.
  \end{enumerate}
\end{ass}

Assumption \ref{ass:collective}.\ref{ass:ego} lists the standard properties of a deterministic individual-specific utility function which guarantee a unique solution. 
Part \ref{ass:efficiency} defines the collective model of \citet{Chiappori1988,Chiappori1992}.
Efficient bargaining rules out non-cooperative, strategic behaviour of individuals towards their spouse.
Part \ref{ass:timehom} states that individuals are egoistic and only derive utility from their own consumption and not through externalities of their partner's consumption.\footnote{This nests the Beckerian caring model with altruistic preferences \citep{Becker1981}. A sufficient condition for this is weak separability of the form $ u_{i}(x^m,x^f) = G_{i}(g_i(x^m), x^f) $ for any two differentiable, increasing, real-valued functions $ G $ and $ g $.}
Without this assumption, we cannot differentiate between preference-driven consumption changes and the possibility of joint consumption of public goods (non-rival, non-excludable) as a couple. For example, consider individuals with stable preferences commuting to work by car. As a single they have to pay market prices for gasoline, but as a couple they can share the cost and consequently consume more other goods, which could lead us to believe that preferences have changed.\footnote{We relax this assumption in the empirical section by allowing for a parametric household production function that allows for consumption externalities.}
We further assume preferences to be time-homogeneous. 
Without this condition, any variation of choices between periods could be attributed to a change in preference over time rather than individuals facing a variation of prices in different periods.\footnote{Applied to long panel data, this assumption ceases to be innocuous. This can be mitigated by conditioning on time-dependent demographics which capture changes in how preferences are aggregated (distribution factors). The author would like to thank an anonymous referee for pointing this out.}
Together, they allow us to use survey data and not rely on an experimental setting in the empirical part of the paper.\footnote{See \citet{BlowBrowningCrawford2021, Adams2014} for a discussion of time consistency.}

We now give economic content to the factual and counterfactual household framework by showing how latent preference primitives generate common utility representations for singles and couples. 

\begin{lem}[Representation of Collective Utilities]\label{thm:collectiveRUM}
Under Assumption \ref{ass:collective}, for any individuals $ i,j \in \mathbb{N}_2 $, there exists a Pareto weight $ \lambda_{ij} \equiv \lambda(\omega^m_i, \omega^f_j, \cdot) $ that permits the following representation of household utility:
\begin{equation}
  \label{eq:collectiveutility}
  \Phi(\omega^m_i, \omega^f_j, \cdot) = \lambda(\omega^m_i, \omega^f_j, \cdot)  u^m(\omega^m_i, \cdot) + (1-\lambda(\omega^m_i, \omega^f_j, \cdot)) u^f(\omega^f_j, \cdot).
\end{equation}
\end{lem}

For any efficient aggregation of preferences to a collective unit, there exists a representation for the household utility that can be decomposed into a weighted combination of individual utilities with weights proportional to each member's bargaining power \citep{Chiappori2009}.
Note that Pareto weights depend on both individuals' preference types $(\omega^m_i, \omega^f_j)$.\footnote{In this exposition, we abstract from distribution factors that could shift bargaining power.}

In Corollary \ref{pro:unitaryrep} of Appendix \ref{app:proofs} we show that the household utility $\Phi_{ij} \equiv \Phi(\omega^m_i, \omega^f_j, \cdot)$ of a single-dummy match is an ordinally equivalent representation of the single's unitary utility $u_i$.
Thus, we may represent any household's utilities as an element of the array $ \{ \Phi_{ij} \}_{i,j \in \mathbb{N}} $, with the dummy partner's demand set to zero for singles.

\subsection{Preference Stability}\label{subs:pools}
In the previous subsection we have concluded that any factual or counterfactual household has a well-defined utility representation within the common framework.
We now tie this to the matching structure in equilibrium, which, under the alternative, induces population level differences due to sorting.
Consequently, this leads to different compositions of preferences between the subpopulations (pools) of single and partnered individuals.

While this paper is agnostic about the matching mechanism, we may think of people as matching to maximise joint surplus $ \Phi $.
The observed matching allocation is the result of individuals optimising their consumption utility by choosing partners.\footnote{The allocation is efficient, in the absence of blocking pairs \citet{Shapley1971}.}\textsuperscript{,}\footnote{One may shift ordinal utility representations by $ h(z_{i}, z_{j}) $ for observed $G$-exchangeable demographics $ (z_i, z_j) $, and an equivariant aggregator $_0 h $ to also account for matching based on observed characteristics, such as education and age, which we consider in Section \ref{sec:results}.}

\begin{lem}[Within Pool Exchangeability]
\label{lem:perm}
For the realised partition of partnered individuals $ \mathbb{N}_2 $ and singles $ \mathbb{N}_1 $ let permutations associated with within-pool relabellings be
\begin{equation}
G \equiv \left\{ \varsigma\in G_0: \varsigma(\mathbb N_2)=\mathbb N_2, \varsigma(\mathbb N_1)=\mathbb N_1 \right\}.
\end{equation}
Then, under Assumptions \ref{ass:matching}-\ref{ass:collective}, within every realised matching-induced pool: (i) the preference environment $ \omega $ remains exchangeable in the sense of Assumption \ref{ass:dependence}, but only with respect to $ G $, and (ii) the realised matching rule  $ \sigma_\omega \equiv M(\Phi(\omega)) $ is permutation equivariant.
\end{lem}

The matching equilibrium partitions the population into single and partnered individuals.
Lemma \ref{lem:perm}.(i) shows that ex-ante exchangeability is therefore no longer preserved across the whole population, but only within the realised pools of singles and couples.
Ex post, membership in the partnered pool $\mathbb{N}_2$ or the single pool $\mathbb{N}_1$ is informative, since the matching equilibrium may select different types into the two pools.
Being called Hera rather than Athena, or Zeus rather than Apollo, now matters: the names no longer only label individuals, but also identify whether their preferences are drawn from the realised partnered or single subpopulation.

In conjunction with part (ii), which establishes anonymity of the realised matching $\sigma_\omega $, we conclude that individuals belonging to factual households remain exchangeable within each block of Figure \ref{fig:partition}: the \color{purple}couples\color{black}, \color{blue}single men\color{black}, and \color{red}single women\color{black}.\footnote{Due to anonymity, once a population $ \omega $ is realised, we write $ \sigma $ without explicit dependence on $ \omega $.}
This allows us to represent the different preference distributions as conditionally i.i.d. and we may define the main population-level hypothesis as follows:

\begin{hypothesis}[Population Preference Stability]
\label{hyp:poollaw}
For each market side $ \kappa \in \{m,f\} $ and pool $ \ell\in\{1,2\} $, denote the realised within-pool distribution for $ A\in \mathcal{B}(\mathbb{X}) $ and $ i\in\mathbb N_\ell $ as:
\begin{equation}
\bar{\mu}_\ell^\kappa(A)
\equiv
\mu (\omega_i^\kappa\in A \mid {\mathcal I}_\ell^\kappa)(\omega),
\qquad
\end{equation}
The population is preference stable if and only if the distributions for singles and partnered individuals coincide: 
$$ \bar{\mu}_1^m=\bar{\mu}_2^m \qquad \text{and} \qquad \bar{\mu}_1^f=\bar{\mu}_2^f. $$
\end{hypothesis}

Starting from the ex-ante composition of preferences $\bar{\mu}_0^m$ and $\bar{\mu}_0^f$, matching sorts individuals into singles and partnered individuals, and within-pool distributions $\bar{\mu}_\ell^\kappa$ are realised.\footnote{The sigma-algebra $\mathcal I^\kappa_\ell$ collects events that are invariant under within-pool relabellings, as defined in Lemma \ref{lem:perm}. The object $\bar{\mu}^\kappa_\ell$ is a regular conditional distribution, that is, a kernel $\bar{\mu}^\kappa_\ell:\Omega\times\mathcal B(\mathbb X)\to[0,1]$. Equivalently, for each $A\in\mathcal B(\mathbb X)$, $\bar{\mu}^\kappa_\ell(\cdot,A)$ is $\mathcal I^\kappa_\ell$-measurable, while for each $\omega\in\Omega$, $\bar{\mu}^\kappa_\ell(\omega,\cdot)$ is a probability measure on $\mathbb X$ \citep[Theorem 5.3]{Kallenberg1997}.}
The null hypothesis restricts this post-matching composition.
It rules out selection into the single and partnered pools based on preferences $\omega_i^\kappa$ because both pools must contain the same realised distribution of individual preference types. 
This does not require matching to ignore preferences. 
Indeed, conditional on being partnered, preferences may still determine who is matched with whom, for example through positive assortative matching.\footnote{Suppose there are two types of preferences: type A and type B. A preference-stable population might consist of 50\% of type A and 50\% of type B in both the singles' and couples' pool. Within the couples' pool, all type A, respectively, type B men and women might be matched assortatively.}

The distributions defining population preference stability are unobserved.
We now develop the relevant individual-level foundations to translate this to an empirically falsifiable restriction. 
For this, we define a \emph{latent configuration} which bundles an observed couple with a single man and a single woman whose preferences serve as counterfactual replacements.

\begin{dfn}
  We call the quadruple  $ (i, j, i', j') \equiv (i, \sigma(i), \tau^m(i), \tau^f(\sigma(i))) $, where $ i \in \mathbb{N}_2 $, a latent household configuration which takes any factual household $ (i, j) $ and pairs it with the individual utilities of two singles, using the transpositions $ \tau^m $ and $ \tau^f $. 
   We call a configuration preference-stable if and only if $ \omega_i^m = \omega_{i'}^m $ and $ \omega_j^f = \omega_{j'}^f $.
\end{dfn}

With this, we can define the structural restriction implied by preference stability that gives our test empirical content.
Let $S_{ij}(p)$ be the Slutsky matrix associated with utility $\Phi_{ij}$  under normalised budgets $ w_{ij} = 1$. 
The Slutsky matrix tells us how demand reacts to prices, keeping utility constant.
For singles, this is purely a substitution effect and must, thus, be symmetric.
For couples, by \citet{Browning1998}, there is an additional exactly one-dimensional channel: a price change may alter the scalar resource sharing and thereby redistribute resources within the household.
That is, $ S_{ij}(p) - \bar{S}_{ij}(p) = u_{ij}(p)v_{ij}(p)^\top $ is rank one for some symmetric matrix $ \bar{S}_{ij}(p)$.\footnote{We may think of this as a factor structure: the factor $v_{ij}(p)$ measures how a price change affects sharing of resources. 
The loading $u_{ij}(p)$ is the direction of the demand changes as one Pound from $j$ is reallocated to $i$, holding prices fixed.} 
This restriction is testable with only aggregate demand data.
Preference stability sharpens it within a given configuration.

\begin{lem}[Preference-Stable Configuration]\label{lem:prefstable}
Under Assumptions \ref{ass:matching}-\ref{ass:collective}, for any preference-stable configuration $(i,j,i',j')$, preference stability requires that the symmetric part satisfies:
\begin{equation}
  \bar{S}_{ij}(p)\in
  \left\{
    \frac{1}{\eta}S^m_{i'0}\left(\frac{p}{\eta}\right)+\frac{1}{1-\eta} S^f_{0j'}\left(\frac{p}{1-\eta}\right)
    : \eta\in(0,1)
  \right\},
\end{equation}
 where $ S^m_{i'0} $ and $ S^f_{0j'} $ are the within-configuration Slutsky matrices for singles.
\end{lem}

This says that for each factual couple there exists a relative share of the household resources received by the man, denoted by $\eta_{ij}(p)\in(0,1)$, such that the symmetric component $\bar{S}_{ij}(p)$ can be replaced by the within-configuration singles $ i'$ and $ j' $ evaluated at $ \eta_{ij}(p) $ and $ 1- \eta_{ij}(p) $.
Preference stability then has a sharp implication: once the symmetric part has been fixed to the single Slutsky matrices, the couple may differ from their sum only through this one-dimensional redistribution channel.

\subsection{Random Utility and Matching Characterisation}\label{subs:population}
Having shown how preference stability restricts behaviour within a single latent configuration, we now show how these local restrictions help restrict the joint distribution to be compatible with preference stability at the \emph{observable} population level. 
Because we observe only one realised matched population, the relevant source of randomness cannot come from repeated markets.
Rather, we obtain it from a label-invariant randomisation over counterfactual assignments within that population. 
This yields a distribution over latent configurations and, through utility-maximisation, a random utility and matching representation of observed behaviour.

We start by describing the observable features of a typical dataset: a finite sample of household demands across budgets, common prices, household income, and possibly individual demographics:

\begin{ass}[Observed Behaviour]\label{ass:observed}
  Let $ x_{ijt} $ be the maximiser of $ \Phi_{ij} $ under budget $ B_{ijt} = \mathcal{B}(p_t, w_{ijt}) $, where $ \mathcal{B}(p, w) = \{ x \in \mathcal{X} : p^\top  x \leq w\} $.\footnote{As is common in the empirical literature of consumer demand we will work with normalised prices $ p_{t}/w_{ijt} $ leading to unit budgets. 
See e.g. \citet{Cornes1992} and also \citet{KitamuraStoye2013} for a discussion. In the empirical section we discuss endogeneity of total expenditure.}
  For household $ i = 1 \ldots n^c+n^m+n^f $, possibly a dummy, we observe $ (j, z_i, \{ x_{it}, w_{it}, p_t\}_{t = 1\ldots T} ) $ where $ z_{i} \equiv (z^m_{it}, z^f_{jt}) $ are observed characteristics, and aggregate demand is $ x_{ijt} = x^m_{ijt} + x^f_{ijt} $. 
\end{ass}

From this data we can identify factual household-level demand functions $ x_{ij} \equiv x(\omega^m_i, \omega^f_{j}, \cdot) $.
Representing configurations via pairs of transpositions $ (\tau^m, \tau^f) $:
\begin{equation}\label{eq:chi}
 \chi(\omega, \tau^m, \tau^f)_i \equiv  \left((\omega_i^m, \omega_{\sigma(i)}^f), (\omega^m_{\tau^m(i)}, \omega^0),  (\omega^0, \omega^f_{(\tau^f \circ \sigma)(i)})\right),
\end{equation}
we obtain the observable demand functions by sending them through the map: \begin{equation}\label{eq:psi}
  \Psi \equiv (\Psi^c, \Psi^m, \Psi^f) \; \text{where} \; \Psi^\kappa \equiv \text{argmax}\circ \Phi \circ \text{proj}^\kappa . 
\end{equation}
The projection $ \text{proj}^\kappa $ takes a configuration $\chi_i$ and extracts the respective couple's, single male's, and single female's preferences, $ \Phi $ builds household utility, which is then maximised to obtain demands.

Since, within a configuration, the single's demand functions serve as counterfactuals to the partnered individual's demand functions, we can check preference stability for a given configuration according to Lemma \ref{lem:prefstable}.
We collect all configurations that are consistent with preference stability into the set $ W_0 $.

To generate counterfactual assignments we randomise over transpositions. For this, we introduce the sampling device $ \rho $, a probability distribution defined on the space $ \mathcal{T}  $. We denote as $ \mu \otimes \rho $ the joint distribution of preferences and transpositions on $ \Omega \times \mathcal{T} $.
This leads to a distribution $ \nu $ of configurations induced by the preference environment and randomisation over transpositions.
\begin{dfn}\label{def:rum}
  Each random environment $ (\omega, \tau) \in \Omega \times \mathcal{T} $ defines the collection of demands $ (x^c, x^m, x^f) \equiv \Psi(\chi(\omega, \tau)) $ as functions from prices $ \mathcal{P} $ to bundles $ \mathcal{X} $. The distribution $ \pi $, measuring events $ A $ in the space of demand functions, is the push-forward of $ \nu $ under $ \Psi $: \begin{equation}\label{eq:rum1pre}
  \pi(A)=
  \int \mathbf{1}\{\Psi(\chi)\in A\}\nu(d\chi) = \int \mathbf{1}\{\Psi(\chi(\omega, \tau))\in A\}\,(\mu \otimes \rho)(d\omega, d\tau). \hspace{2em}
\end{equation}
We observe only the marginal distributions of $ \pi $ which we denote $ \pi^c $, $\pi^m$, and $ \pi^f $.
\end{dfn}

Through this formulation, we can define the respective marginal distributions of factual demands for couples, single men, and single women, as a mixture of configurations.
Each of the configurations we can check for preference stability. 
We show that this is a random utility model, by proving that configurations, and thus the induced demand functions, are sampled randomly.
In addition, we require that the population level preference-stability hypothesis \ref{hyp:poollaw} implies the existence of a distribution of configurations supported only on the preference-stable set $ W_0 $.
This leads us to the main result of the paper.

\begin{thm}[Random Utility and Matching Representation]
\label{pro:definetti}
Let $\mathcal{I}^\star$ collect all events on $\Omega \times \mathcal{T}$ that depend on the configuration sequence only up to relabelling of individuals within pools $ G $. Under Assumptions \ref{ass:matching}-\ref{ass:observed}, the sequence of configurations $(\chi_i)_{i \in \mathbb{N}_2}$ is conditionally i.i.d. given $\mathcal{I}^\star$ with distribution $\bar{\nu}$.
The conditional distribution of demand triples induced by sampling of configurations according to Definition \ref{def:rum} is the push-forward of $\bar{\nu}$ under $\Psi$:
\begin{rum}\label{eq:rum1}
    \bar{\pi}(A) = \int \mathbf{1}\{\Psi(\chi_i) \in A\} \, \bar{\nu}(d\chi_i).
\end{rum}
Further, under Hypothesis \ref{hyp:poollaw}, the marginal demand distributions $\bar{\pi}^c, \bar{\pi}^m, \bar{\pi}^f$ of $\bar{\pi}$ admit a preference-stable representation. That is, there exists $\nu^\star $ supported on preference-stable configurations $W_0$, such that pushing it forward through $\Psi^c, \Psi^m, \Psi^f$ yields the respective marginals.
\end{thm}

In the preliminary Lemma \ref{lem:proj_exchangeable} in Appendix \ref{app:proofOfprojection} we show that, if we sample configurations without introducing dependence on the unobserved preferences $ \omega $, e.g. by systematically over-sampling certain individuals or types, the within-pool exchangeability of preferences established in Lemma \ref{lem:perm} extends to the sequence of latent configurations.\footnote{In our implementation, we exhaustively enumerate configurations which satisfy this symmetry requirement, as would uniform sampling at random.}\textsuperscript{,}\footnote{In Figure \ref{fig:partition}, this extends the within-solid-block exchangeability to the dashed \emph{counterfactual blocks} generated by $(\tau^m,\tau^f) $, since each of them is a measurable projection of the exchangeable sequence $ \chi $.}
Consequently, this extends the \citet{definetti1931} representation of the preference environment to the configuration sequence $ \{\chi_i\}_{i \in \mathbb{N}_2}$ \citep{HewittSavage1955}.
Thus, every event we can learn from this economy which is not a consequence of arbitrary labels is contained in the sub-sigma-algebra ${\mathcal{I}}^\star $, and can be measured by a conditional distribution.\footnote{The distribution is random through its dependence on the state of the world $ (\omega, \tau) \sim \mu \otimes \rho $. Upon realisation of the whole environment of latent preferences and configurations, the distribution becomes deterministic. Any counterfactual $ (\omega',\tau') $ leads to a different distribution, on the same sigma-algebra.}

Theorem \ref{pro:definetti} tells us that, conditional on $ \mathcal I^\star $, each component of the configuration array has distribution $ \bar\nu $. 
Pushing this conditional distribution through \( \Psi \) gives the conditional choice distribution \( \bar\pi \).
Thus, we can treat demand functions as conditionally i.i.d. which permits the construction of non-parametric estimators for them. 

To rationalise the observed distributions $ \bar\pi $ there must exist a distribution $ \nu^\star $ that puts all mass on the subset of preference-stable configurations.
Theorem \ref{pro:definetti} further establishes that failure of rationalisability implies failure of the preference stability hypothesis \ref{hyp:poollaw}, making the hypothesis empirically falsifiable.\footnote{Corollary \ref{cor:blockmarschak} in the Appendix \ref{app:blockmarschak} goes beyond this and characterises preference stability, and, thus, the existence of $ \nu^\star $, via Block-Marschak inequalities.}

Having developed the random utility theory based on continuous demand functions, in the next section we operationalise it by deriving the discrete choice counterpart using revealed preference axioms.
 \section{Testing based on a Discrete Choice Characterisation}\label{sec:inference}
The configuration-level Slutsky restriction from the continuous characterisation in the previous section requires non-parametric estimation of demands and, thus, a large number of observed budget sets.
This section replaces the continuous characterisation by one that characterises choices on a small number of budgets.
We introduce distinct revealed preference types, replacing demand functions, which we combine to configuration types. 
Preference stability then becomes a support restriction on a finite-dimensional distribution over configuration types, which can be operationalised as a constrained optimisation problem based on observed choice frequencies and a deterministic matrix defining preference-stable configurations.

\subsection{Choice and Configuration Types}\label{subs:finite}
It is neither feasible nor necessary to consider the high-dimensional problem with continuous demands and permutations at the individual level.\footnote{For double transpositions $ (\tau^m,\tau^f) \in \mathcal{T} $, the cardinality of the space of permutations is of order $ \mathcal{O}(n^4) $, if the number of single individuals is proportional to $ n $.}
Instead, we now show that we can, equivalently, use a characterisation based on revealed preference types, defined in a way that knowledge of them fully determines rationality and efficiency.
Counterfactual assignments can then also be considered at the type level, reducing the dimensionality of the problem drastically.

\begin{dfn}[GARP]
  \label{dfn:garp}
A collection of choices $ (p_t,x_t)_{t=1}^{T} $ satisfies the generalised axiom of revealed preferences (GARP) if there exist binary relations $R$ and $\mathcal R$ that satisfy:
\begin{enumerate}
  \item\label{dfn:garpdirect} if $p_s^\top x_s \geq p_s^\top x_t $, then $x_s \ R \ x_t $,
\item if $ x_s \ R \ x_u $, $ x_u \ R \ x_v $, \ldots, $ x_z \ R \ x_t $ for some sequence $(u\ldots z)$, then $ x_s \ \mathcal R \ x_t $,
\item if $ x_s \ \mathcal R \ x_t $, then $p_t^\top x_t \leq p_t^\top x_s $.
\end{enumerate}
\end{dfn}

If an individual purchases bundle $x_s$ even though $x_t$ was affordable at the same prices, we say it was ``directly revealed preferred'' and write $x_s R x_t$. 
Further, by chaining together any (possibly empty) sequence of direct revelations of preferences, transitivity allows us to infer preference revelations of some bundles we cannot otherwise compare because we never observe budgets that allow us to directly distinguish them.
GARP demands that there is no $x_s$ which is revealed preferred to $x_t$ and yet, at the same time, $x_t$ is revealed preferred to $x_s$.
No cycles of mutual preference can occur.

\begin{exa*}

\begin{figure}[h!]
 \begin{subfigure}[b]{0.4\textwidth}
\begin{tikzpicture}[scale=0.7,>=stealth] \fill[pattern=north east lines,pattern color=pink]
        (0,4) -- (0,0) -- (3,0) -- cycle;
    \fill[pattern=north west lines,pattern color=teal]
        (0,2) -- (0,0) -- (6,0) -- cycle;
    \draw[->] (0,0) -- (0,4.5);
    \draw[->] (0,0) -- (6.5,0);
    \draw (0,2) -- (6,0);
    \draw (0,4) -- (3,0);
\end{tikzpicture}
\end{subfigure}
\begin{subfigure}[b]{0.2\textwidth}
\begin{tikzpicture}[scale=0.7,>=stealth]
    \fill[pattern=north east lines,pattern color=pink]
        (0,4) -- (0,0) -- (3,0) -- cycle;
    \draw (0,4) -- (0,0) -- (3,0) -- cycle;
    \draw (0, 2) -- (2, 4/3);
    \fill (1,5/3) circle (1.5pt) node[below] {$\mathbf{x_t'}$};
    \fill (2.5,4/6) circle (1.5pt) node[right] {$\mathbf{x_s'}$};
    \fill (1,8/3) circle (1.5pt) node[above right] {$\mathbf{x_s}$};
\end{tikzpicture}
\end{subfigure}
\begin{subfigure}[b]{0.2\textwidth}
\begin{tikzpicture}[scale=0.7,>=stealth]
    \fill[pattern=north west lines,pattern color=teal]
        (0,2) -- (0,0) -- (6,0) -- cycle;
    \draw (0,2) -- (0,0) -- (6,0) -- cycle;
    \draw (3,0) -- (2, 4/3);
    \fill (1,5/3) circle (1.5pt) node[above] {$\mathbf{x_t'}$};
    \fill (2.5,4/6) circle (1.5pt) node[left] {$\mathbf{x_s'}$};
      \fill (4,2/3) circle (1.5pt) node[above right] {$\mathbf{x_t}$};
\end{tikzpicture}
\end{subfigure}
  \caption{Illustration of individual choices under two different budgets.}
  \label{fig:budgets}
\end{figure}
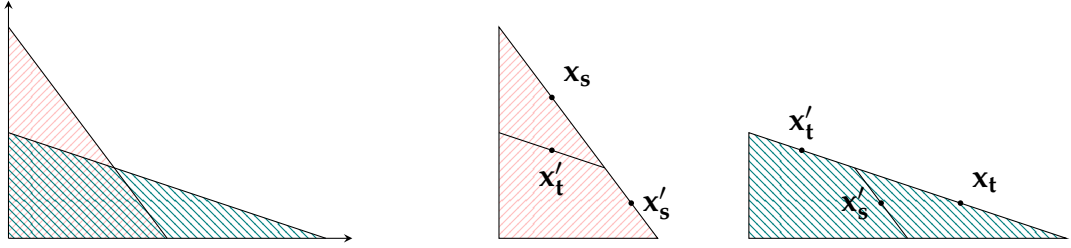
 Figure \ref{fig:budgets} depicts two intersecting budget lines and two arbitrary choices of the same individual when faced with either budget.\footnote{This is the simplest setting in which these axioms can produce behaviour inconsistent with utility optimisation. Because there are only two budgets, there are no cycles beyond violation of direct revealed preferences. We call this the weak axiom (or WARP).}
Assume $ x_s' $ (or any other point in its containing line-segment) is chosen if faced with the salmon budget $ B_s $. 
The individual has revealed that they prefer this choice over $ x_t' $ because the latter is also in the salmon budget ($ p_s x_s' \geq p_s x_t' $, i.e. $ x_t' \in B_s $). We conclude that $ x_s' R x_t' $.
We cannot infer anything about $ x_t $ because it is outside of the salmon budget. 
Because of the symmetry of the problem, we can also consider the teal budget $ B_t $.
Again, no claims can be made about $ x_s $ because it is not contained in it.
However, we might infer $ x_t R x_s' $ or $ x_t' R x_s' $ depending on which bundle was chosen ($x_s' \in B_t $).
In case of the latter, there is a cycle of mutual preference because from the salmon budget we concluded $ x_s' R x_t' $ and from the teal budget we concluded $ x_t' R x_s' $.
This is a violation of GARP, and ruled out by Assumption \ref{ass:collective}.\ref{ass:ego}. 
All other delegate bundles $ (x_s, x_t), (x_s', x_t), (x_s, x_t') $ are rational.
\end{exa*}

To characterise singles and couples as types we invoke two fundamental, well-established results from the revealed preference literature.
First, for singles, by \citet{Afriat1967} and \citet{Varian1982}, the existence of a utility function defined in Assumption \ref{ass:collective}.\ref{ass:ego} requires observed choices $ (x_t, p_t)_{t=1}^T  $ to satisfy GARP.
Second, for couples, by \citet{Cherchye2011}, under Assumptions \ref{ass:collective}.\ref{ass:ego} and \ref{ass:collective}.\ref{ass:efficiency}, there exist personalised continuous consumption bundles $ (\check{x}^m, \check{x}^f) $ such that $\check{x}^m + \check{x}^f = x $ and both $ (\check{x}^m_{t}, p_t)_{t=1}^T  $ and $ (\check{x}^f_{t}, p_t)_{t=1}^T  $ satisfy GARP.\footnote{Note that their characterisation also allows for public goods and consumption externalities.}

For our heterogeneous population, this means that given a realisation of the preference environment $ \omega = (\omega^m, \omega^f) $ the axioms must hold for every household (Assumption \ref{ass:collective}).
Preference heterogeneity allows the revealed preference relation $ R $ and its transitive closure $ \mathcal{R} $ (Definition \ref{dfn:garp}) to be different for any two individuals even if they face the same prices.
Thus, for $ \kappa \in \{m,f\} $, we write $ R^\kappa_{i} \equiv R(\omega^\kappa_i)$ and $ \mathcal{R}^\kappa_i \equiv \mathcal{R}(\omega_i^\kappa) $ and define $ x \ \mathcal{R}^\kappa_i \ x' $ if and only if $ (x,x') \in \mathcal{R}^\kappa_i \subseteq \mathcal{X} \times \mathcal{X} $.\footnote{$ R $ also depends on prices which we treat as fixed and the same for everyone by Assumption \ref{ass:observed}.}
For a fixed and finite number of budgets, the map $ \mathcal{R}: \omega \mapsto \mathcal{R}(\omega) $ is not injective even if utilities are.
This means that there are individuals $ \omega \neq \omega' $ whose utilities are not empirically distinguishable even if they pick different continuous bundles when faced with the same budget.
It is, thus, without loss for the test to treat their choice as equal.
Consequently, a finite number of budgets only induces a finite number of revealed preference types.\footnote{With choices on a dense set of budgets, we could recover preferences from observed choices \citep{mascolell1977, mascolell1978}. Since $ \omega $'s are ordinal preferences, the relation would then be one-to-one.}

\begin{dfn}[Individual Types]\label{dfn:garptype}
  For normalised budgets and common prices, each choice path
  $x=(x_1,\ldots,x_T)\in\mathcal X^T$ induces a revealed-preference relation
  $\mathcal R^\kappa\subseteq\mathcal X\times\mathcal X$.
  We call two choice paths equivalent if they induce the same relation.
  The resulting finite partition of $\mathcal X^T$ is denoted by
  $\bar{\mathcal X}^\kappa$, and each cell
  $\xi\in\bar{\mathcal X}^\kappa$ is called an individual revealed preference type.\footnote{Each cell is the intersection of half-spaces induced by direct revealed preference comparisons.}
\end{dfn}

Knowing an individual's revealed preference type answers all relevant \emph{revealed-preferred questions} for any two $ x, x' \in \mathcal{X} $ and, thus, describes the heterogeneous preferences of this individual, absent additional functional form restrictions.

For singles $ i \in \mathbb{N}_1 $ on either side of the matching market, $ \xi_{i} $ is directly observed from data.
For couples $ i \in \mathbb{N}_2 $, we think of household types as a \emph{latent} pair of individual types in $ \bar{\mathcal{X}}^m \times \bar{\mathcal{X}}^f $.
For normalised prices, write demand functions from the previous section as $ (x^m_i(\eta), x^f_{\sigma(i)}(\eta)) $ where $ \eta \in (0,1) $ is the endogenous relative share of endowment $ w $. 
Then exact knowledge of $ \eta_{i\sigma(i)} $ determines both individual's private consumption $ (x^m_i, x^f_{\sigma(i)}) $ and, thus, their revealed preference types $ R^m_i $ and $ R^f_{\sigma(i)} $ on the observed budgets.
By \citet{Cherchye2011}, under Assumption \ref{ass:collective}.\ref{ass:ego}-\ref{ass:efficiency}, 
$$\mathcal{X}_{i\sigma(i)} \equiv \left\{ \left( x^m_i(\eta), x^f_{\sigma(i)}(\eta) \right) : \eta \in {\mathcal E}_{i\sigma(i)} \right\} $$
is non-empty. This implies that the generating $ \mathcal E_{i\sigma(i)} \subseteq (0,1) $ must also be non-empty.
In general, the set is not a singleton, and any two couples' discrete revealed-preference-types generated by this set, are observationally equivalent since they share the same set of feasible quantities $ \mathcal{X}_{i\sigma(i)} $. 

Unfortunately, without imposing restrictions beyond Assumption \ref{ass:collective}, there is no unique way to partition this type space further, to accommodate sub-types based on each member's revealed preference type (which is identified within a stable configuration).
Thus, to discretise the space of configurations we have to make a choice.
Two possibilities have been established in the literature.
Either we pre-test the data for the existence of feasible quantities using the mixed integer approach in \citet{Cherchye2009,Cherchye2011}, discard all couples for which $ \mathcal{X}_{i\sigma(i)} $ is empty, and characterise couple's types only via the binary relation on aggregate choices. 
Alternatively, we resort to a collection of necessary conditions based on \emph{hypothesised (revealed) preference relations}, listed in Definition \ref{thm:necessary} below.
We choose the latter for the remainder of the paper, but also report results from both implementations, which we discuss in Section \ref{sec:results}.

\begin{dfn}[CARP]\label{thm:necessary} If a collection of choices $ (p_t, x_t)_{t=1}^T $ satisfies the Collective Axiom of Revealed Preferences then there exist binary relations $ R, H^m, \mathcal{H}^m, H^f, \mathcal{H}^f \subseteq \mathcal{X} \times \mathcal{X} $, s.t.
\begin{enumerate}
  \item\label{axm:i} if $ x_s \ R \ x_t $, then $x_s \ H^m\  x_t $ or $ x_s \ H^f \ x_t $,
 \item \label{axm:ii}if $ x_s \ H^\kappa \ x_{u} $,  $ x_{u} \ H^\kappa \ x_{v} $, $\ldots$, $ x_{z} \ H^\kappa \ x_t $ then $ x_s \ \mathcal{H}^\kappa \ x_t $ for $ \kappa \in \left\{ m, f \right\} $,
\item \label{axm:iii} if $ x_s \ R \ x_t $ and $ x_t \ \mathcal{H}^\kappa \ x_s $, then $ x_s \ H^{\kappa'} \ x_t $ for 
  $ \kappa \neq \kappa' $,
\item\label{axm:iv} if $ x_s \ R \ \left(x_{t_1} + x_{t_2} \right) $ and  $ x_{t_1} \ \mathcal{H}^\kappa \ x_s $ then $ x_{s} \ H^{\kappa'} \ x_{t_2} $ for $ \kappa \neq \kappa' $,
\item\label{axm:v} if $ x_{s_1} \ \mathcal{H}^m \ x_{t} $ and $ x_{s_2} \ \mathcal{H}^f \ x_{t} $ then  $ \neg \left( x_t \ R \ \left(x_{s_1} + x_{s_2}\right) \right) $,
\item \label{axm:vi}if $ x_s \ \mathcal{H}^m \ x_t $ and $ x_s \ \mathcal{H}^f \ x_t$, then $ \neg \left( x_t \ R \  x_s  \right) $
\end{enumerate}
where $ x_s \ R \ x_t $ whenever $ p_s^\top x_s \geq p_s^\top x_t $ and $ \mathcal{H}^\kappa $ is the transitive closure of $ H^\kappa$.\footnote{Note that $ R $ is a binary relation but does not correspond to an actual preference relation, since household consumption is the result of aggregation of individual preferences.}
\end{dfn}

\citealp{Cherchye2007} show that, under Assumptions  \ref{ass:collective}.\ref{ass:ego} and  \ref{ass:collective}.\ref{ass:efficiency}, the collective axiom (CARP) holds.
Since this characterisation does not use individualised quantities, we have the additional requirement of items \ref{axm:iv} and \ref{axm:v} which rule out the situation where individuals have different preferences over bundles but as a household they consume an inferior bundle when they could have afforded both.
This is clearly a violation of efficiency.
Importantly, each of the restrictions divides $ \mathcal{X}^T $ into two well-defined half-spaces.
\begin{dfn}[Collective Types]\label{dfn:carptype}
  For normalised budgets and common prices, the finite collection of
  CARP-relevant inequalities induces a finite partition
  $\bar{\mathcal X}^c$ of aggregate household choice paths
  $x^c=(x_1^c,\ldots,x_T^c)\in\mathcal X^T$.
  Two household choice paths are in the same cell if they have the same
  revealed-preference pattern for all aggregate and double-sum comparisons
  appearing in Definition \ref{thm:necessary}.
  We call each cell $\xi^c \in\bar{\mathcal X}^c$ a collective revealed preference type.
\end{dfn}
These restrictions are fine enough to classify couples not only by the collective revealed preference type $\xi^c \in\bar{\mathcal X}^c$ induced by their observed aggregate choices $(p_t,x_{i,\sigma(i),t})_{t=1}^T$, but also by the counterfactual revealed preference types assigned to their two members in a configuration.
Hence, they allow us to replace the hypothesised relations by the respective singles' actual revealed-preference relations within a given preference-stable configuration.
This strengthens the requirement of collective rationality of the observed couple, i.e. existence of a feasible resource share $ \eta $, to the configuration retaining rationality after the exchange with single preferences.
By Lemma \ref{lem:prefstable}, the additional restrictions tighten the feasible set of resource shares, thus reducing the number of preference-stable configurations.\footnote{Appendix \ref{app:examples} discusses the relationship between $ \eta $ and the random utility representation.}

With our definition of discrete types, many realisations of individual and collective types are equivalent.
Since transpositions that swap two individuals of the same type have no empirical content, we only have to sample matches based on revealed preference types rather than individual assignments.
Thus we can discretise a configuration $ \chi_i $ defined by $ (\tau^m, \tau^f) $ by a configuration type:
\begin{equation}\label{eq:conftype}
  \theta_{i} \equiv \left(\xi_{i\sigma(i)}^c, \xi_{\tau^m(i)}^m, \xi^f_{\tau^f(\sigma(i))}\right) \in \Theta \equiv \mathcal{\bar{X}}^c \times \mathcal{\bar{X}}^m \times \mathcal{\bar{X}}^f,
\end{equation} 
where we denote the subset of preference-stable type configurations by $ \Theta_0 \subset \Theta $.

The test statistic is derived in the next section. We finish this section with an example of a minimal economy that has power to detect failure of preference stability and a discussion of the dimension of the discrete type space.

\begin{exa}\label{exa:types}
  Let us revisit a typical configuration of a household $ i \in \mathbb{N}_2 $ with hypothetical partners randomised $ (\tau^m, \tau^f) \sim \rho $ by returning to our example of Figure \ref{fig:budgets}. We extend it by an third, umber budget $ u $.\footnote{The collective model is testable only when at least three goods and three budgets are available.}
Let $ x^m = x_{\tau^m(i)} $ be the single man (Apollo) consuming $ x^m_s $, $ x^m_t $, and $ x^m_u $.
Further let $ x^f = x_{\tau^f(\sigma(i))} $ be the single woman (Athena) consuming $ x^f_s $, $ x^f_t $, and $ x^f_u $.
The original couple (Zeus \& Hera) is jointly consuming $ x^c_s $, $ x^c_t $, $ x^c_u $ where $ x^c = x_{i,\sigma(i)} $.
As discussed in Lemma \ref{lem:prefstable}, we normalise budgets to one: $ p_{v} x_{v} = 1 $ for all $ v \in \{ s,t,u \} $. 
Suppose the configuration satisfies the following inequalities, sufficient to characterise it in terms of revealed preference types, which can be checked against the conditions in Definition \ref{thm:necessary}. 
{
\arraycolsep=2.75pt
\begin{equation}
\begin{array}{ccc|cc|cc}
  p_t x^c_s \geq 1,& p_u x^c_s \geq 1,& p_s (x^c_t+x^c_u) \leq 1 & p_t x_s^m \leq 1,& p_u x_s^m \geq 1& p_t x^f_s \geq 1, & p_u x^f_s \leq 1 \\
	 p_s x^c_t \geq 1,& p_u x^c_t \geq 1,& p_t (x^c_s+x^c_u) \geq 1& p_s x_t^m \geq 1,& p_u x_t^m \leq 1& p_s x^f_t \geq 1,& p_u x^f_t \geq 1 \\
   p_s x^c_u \leq 1,& p_t x^c_u \geq 1,& p_u (x^c_s+x^c_t) \leq 1 & p_s x_u^m \geq 1,& p_t x_u^m \geq 1 & p_s x^f_u \geq 1,& p_t x^f_u \geq 1 
 \end{array} \text{.}
\end{equation}
}

Each of these inequalities defines a half-space in $ \mathcal{X}  $. Households are characterised by which half contains their continuous choice, a consequence of their preferences.
In this particular configuration, each type of household is rational, the couple is CARP-consistent, the man's revelations $x^m_t R x^m_s$ and $x^m_u R x^m_t$ satisfy GARP, and the woman's $x^f_u R x^f_s$ satisfies GARP.
However, the configuration is not preference-stable: the man's $x^m_t R x^m_s$ and the woman's $x^f_u R x^f_s$ together imply, by item \ref{axm:v} of Definition 6, that the couple cannot have $x^c_s\, R (x^c_t + x^c_u)$.

By counting the distinct sign patterns of this finite set of revealed-preference inequalities induced by intersecting the  $ T $ budgets, we get the cardinality of unitary household types $ |\bar{\mathcal{X}}^s| = 2^{T(T-1)} $ for $ s \in \left\{ m, f \right\} $. For couples, we have to evaluate inequalities for double-sums according to Definition \ref{thm:necessary} \ref{axm:iv} \& \ref{axm:v}, which appear in the third column, adding another $T\binom{T-1}{2}$ comparisons. 
In our case with $ T = 3 $, we have $ |\bar{\mathcal{X}}^c| = 2^6\cdot 2^3 = 512 $ collective revealed preference types.
In total, we thus have $ 512\cdot64\cdot64 = 2^{21} = 2,097,152 $ configurations.
Evaluating them computationally, $ 475,136 $ are consistent with the collective axiom based on the necessary conditions from Definition \ref{thm:necessary} using only aggregate household consumption data.\footnote{This leaves us with about $ 22.7 \% $ collectively rational types. 
From this, we should not necessarily conclude a restrictive nature of the collective model since for a given range of budget planes only a subset of the total choice set would actually be feasible (e.g. have positive demands).}
Imposing preference stability, this further reduces to $ 2,996 $ preference stable configurations.
\end{exa}

\subsection{Testing Preference Stability}
The finite-type characterisation of the random utility model \eqref{eq:rum1} can be analysed within the stochastic choice setting of \citet{McFadden1991}, and \citet{McFadden2005}.
Rationalisability of the model, defined in equation \eqref{eq:rumdiscrete} below, asks whether the distribution of observed revealed preference types can be rationalised by a population of deterministic preference-stable configuration types $ \theta \in \Theta_0 $.
In Theorem \ref{pro:definetti}, we showed that failure of rationalisability falsifies preference stability.\footnote{By Lemma \ref{lem:prefstable} and Definition \ref{dfn:carptype}, $\theta(W_0) \subseteq \Theta_0$. The inclusion may be strict, since the discrete characterisation is necessary but not sufficient for the underlying preference-stability restriction. The test based on $\Theta_0$ therefore has correct size under \ref{hyp:poollaw} but is conservative. We discuss this in Proposition \ref{lem:equivalences}.}
For the statistical test, we must account for the sampling uncertainty entering through the estimation of the, now discrete, conditional revealed preference type distribution $ \bar{\pi} $.

For a configuration $ \chi $, defined in \eqref{eq:chi} we defined demand functions through the optimal choice rule \eqref{eq:psi}, which we now explicitly let be dependent on prices through the budget constraints and write $ \Psi_p $. Each of the resulting margins $ (x^c, x^m, x^f) $ is compatible with unitary, respectively, collective utility maximisation.
Let $ p \equiv (p_t)_{t = 1}^T $ collect all prices and let $ x(p) \equiv (x^{c}(p), x^{m}(p), x^{f}(p)) $ be the optimal demands of a given configuration where $ x^{\kappa} \equiv (x_{t}^{\kappa}(p_t))_{t=1}^T $.
We then apply the discretisation map $ \Delta : \mathcal{X}^{3\cdot T} \rightarrow \Theta $ to the optimal demands $ x(p) $ of a configuration $ \chi $ which maps to the unique equivalence class (Definitions \ref{dfn:garptype} and \ref{dfn:carptype}) containing them. 
This determines the configuration type $ \theta(\chi, p) = \Delta(\Psi_p(\chi))$ as a function of (observed) prices.

Taking prices as given, $ \theta $ defined in equation \eqref{eq:conftype}, is a deterministic function of the configuration $ \chi $.
Hence we can define the observed distribution of discrete choices of households of type $ \kappa $ as the push-forward of the distribution of configurations $ \nu^\star $: \begin{rum}
  \label{eq:rumdiscrete}
  \bar{\pi}\left(\xi^\kappa_{i} = \xi \right) = \sum\limits_{\theta \in \Theta_0} \mathbf{1}\left\{ \text{proj}^{\kappa}(\theta) = \xi \right\}\nu^{\Delta}(\theta) = \int_{\chi \in W_0} \hspace{-0.75em} \mathbf{1}\left\{ \theta(\chi) \in (\text{proj}^\kappa)^{-1}(\xi) \right\} \nu^\star(d\chi).\hspace{1.5em}
\end{rum}
This is an empirically tractable version of the random utility and matching model \eqref{eq:rum1}. 
The first equation of \eqref{eq:rumdiscrete} defines a linear program.
The data identifies only the marginal distributions of the observable revealed preference types appearing on the left-hand side of \eqref{eq:rumdiscrete}. 
The joint distribution over configurations is latent, but GARP, CARP, and preference stability restrict its support to the admissible set $\Theta_0\subset\Theta$. 

By Theorem \ref{pro:definetti}, Hypothesis \ref{hyp:poollaw} implies the existence of a distribution $\nu^\star$ on $W_0$ rationalising $\bar\pi$. 
Through the discretisation map $\theta(\chi)$, this in turn implies that the discrete marginals are rationalisable by $\nu^\Delta = \nu^\star \circ \theta^{-1}$ on $\Theta_0$. 
Proposition \ref{lem:equivalences} characterises this discrete rationalisability and provides the basis for the test statistic.

\begin{pro}\label{lem:equivalences}
  Under Assumptions \ref{ass:matching}-\ref{ass:observed}, the following statements are equivalent:
\begin{enumerate}
\item \label{equ:qualitative} The marginal choice distributions $\bar\pi^c, \bar\pi^m, \bar\pi^f$ are rationalisable by a distribution supported on the set of preference-stable configuration types $\Theta_0$ according to \eqref{eq:rumdiscrete}.
\item \label{equ:system} There exists $\nu^\Delta $ on the $ |\Theta_0|$-dimensional unit simplex such that $ A \nu^\Delta = \bar{\pi} $, where the columns of $ A $ represent all preference-stable type configurations $ \theta \in \Theta_0 $. 
\item \label{equ:quadratic} For $ \underline{\nu} = 0 $, the projection residual satisfies $ \mathcal{J}_n(\bar{\pi}, \underline{\nu}) \equiv n \min_{\gamma \in \left\{ A\nu^\Delta | \nu^\Delta \geq \underline{\nu} \right\}} (\bar{\pi} - \gamma)^T \Omega (\bar{\pi} - \gamma) = 0 $ where $ \Omega $ is a positive definite square weighting matrix.
\item \label{equ:nnls} The vector $ \nu^\Delta $ is a fixed point under the operation 
\begin{equation}\label{eq:landweberstep}
\Gamma_{\bar{\pi},\underline{\nu}} : s \mapsto \text{max}(0, s - \text{diag}(H\iota)^{-1}(Hs + f(\bar{\pi}, \underline{\nu})))
\end{equation}
where $ H = A^\top \Omega A $ and $ f(\bar{\pi}, \underline{\nu}) = - A^\top \Omega (\bar{\pi} - A \underline{\nu}) $.
\end{enumerate}
\end{pro}
We construct the matrix $ A $ in $ A \nu^\Delta = \bar{\pi} $, defined in Proposition \ref{lem:equivalences}.\ref{equ:system}, based on deterministic configuration-types, i.e. with a typical column representing a preference-stable configuration, which vertically concatenates one-hot encodings of a male single type ($ \xi^m $), a female single type ($\xi^f$), and a couple type ($\xi^c$) each of them individually rational.
Consequently, the matrix consists of $ \sum_{\kappa \in \left\{c,f,m\right\}} |\bar{\mathcal{X}}^{\kappa}| $ rows and $ |\Theta_0| $ columns, where $ |\bar{\mathcal{X}}^{\kappa}|  $ is the number of different choices a household of a given kind can make.
We then split $ A $ into $ 3 $ blocks of respective row-length $ |\bar{\mathcal{X}}^c| $, $ |\bar{\mathcal{X}}^f| $ and $ |\bar{\mathcal{X}}^m| $ and denote by $ A_{\kappa,\boldsymbol{\cdot},\boldsymbol{\cdot}} $ each block of $ A $.
If household configuration $ \theta \in \Theta_0 $ (columns, indexed by $ l $) yields type $ \xi^{\kappa}_{j} $ for $ \kappa \in \left\{c,f,m\right\} $ then $ A_{\kappa,j,l} = 1 $ and zero otherwise.

Because there are many preference-stable configurations compared to the number of individual types, the matrix $ A $ does not have full column-rank.
Thus $ \nu $ is not point-identified. Following \citet{KitamuraStoye2013} we exploit Proposition \ref{lem:equivalences}.\ref{equ:quadratic} as the computational formulation for the condition \ref{equ:system}, in which we obtain  $ \gamma $ by projecting choice probabilities $ \bar{\pi} $ onto the linear cone enforcing the preference-stability constraints $ \left\{ A\nu : \nu \geq \underline{\nu} \right\} $ and define the test statistic as the corresponding projection residual.
The case $\underline{\nu}=0$ gives the population rationalisability condition, while inference below uses tightened lower bounds.

\subsection{Inference}\label{subs:inference}
The vector of choice probabilities $ \bar{\pi} $ is subject to sampling uncertainty.
To obtain the sample statistic $ \mathcal{J}_n(\widehat{\bar{\pi}}_n, \underline{\nu}) $, we require a consistent estimator $ \widehat{\bar{\pi}}_n $ of $ \bar{\pi} $.
To obtain critical values for the random quantity $ \mathcal{J}_n(\widehat{\bar{\pi}}_n, \underline{\nu}) $, we need a consistent approximation of the asymptotic distribution $ \sqrt{n}(\widehat{\bar{\pi}}_n - \bar{\pi}) $.
In Theorem \ref{pro:definetti}, we established the de Finetti representation, as a consequence of exchangeability of $ \chi $ and uniform sampling of transpositions.
This result immediately carries over to household types, due to the measurability of the map $ \Delta $ from configurations to configuration types $ \theta $.
Hence, conditional on the permutation-invariant sigma-algebra ${\mathcal{I}^\star} $, observed revealed preference types are i.i.d. with probabilities $\bar{\pi}^\kappa=(\bar{\pi}^\kappa_1,\dots,\bar{\pi}^\kappa_{|\bar{\mathcal{X}}^\kappa|})$ where $  \bar{\pi}_{j}^\kappa \equiv \bar{\pi}(\xi_i^\kappa=\xi_j^\kappa) $.

Consequently, we can obtain a consistent estimator $ \widehat{\bar{\pi}}_n $ for $ \bar{\pi} $, by taking sample analogues of the discrete choice probabilities. 
Partitioning $ \bar{\pi} $ the same way as a column $ A_{\kappa} $, we estimate the sample proportions of a given type by $ \widehat{\bar{\pi}}^{\kappa}_{n,j} = \frac{1}{n_{\kappa}}\sum_{i=1}^{n_{\kappa}} \indicator \{ \xi^{\kappa}_{i} = \xi^{\kappa}_j \} $ where $ \xi^{\kappa}_i $ is the revealed preference type of household $ i = 1 \ldots n^{\kappa} $.

To obtain the critical values for inference, we may use a non-parametric bootstrap. 
We construct a bootstrap sample $ \widehat{\bar{\pi}}_n^b $ for  $ b = 1 \ldots B $.
Because of many binding constraints, inference requires a tuning parameter $\underline{\nu} = \tau_n \iota$ with $\tau_n \to 0$ as $ n \to \infty $.\footnote{We need this, because otherwise many parameters lie on the boundary of the parameter space. Without it, the bootstrap would not be valid \citep{Andrews2000}. $ \tau_n = |\Theta_0|^{-1}\sqrt{{\log \underline{n}}/{\underline{n}}} $ is a tightening parameter that shifts out the cone from the origin  where $ \underline{n} $ is the minimum number of available observations among  $ n^c, n^m, n^f $ and $ \iota $ is the vector of ones with dimension aligning with $ \nu $. We set bootstrap repetitions to $ B = 500 $, and tune the tightening parameter $ \tau_n $ based on our simulation study. }
Let $ \widehat{\gamma}_{n,\tau_n} $ be the minimiser of \ref{equ:quadratic} under the tightened cone constraint.
For each bootstrap draw $ b $, we compute the centred choice probabilities $ \widehat{\bar{\pi}}^b_{n,\tau_n} = \widehat{\bar{\pi}}_n^b - \widehat{\bar{\pi}}_n + \widehat{\gamma}_{n,\tau_n} $ and evaluate the test statistic $ \mathcal{J}_n(\widehat{\bar{\pi}}^b_{n,\tau_n}, \iota\tau_n) $ to obtain their empirical distribution $ \widehat{F}_{n,B,\mathcal{J}_n} $.
We now establish that the corresponding critical value yields an asymptotically valid test.

\begin{thm}\label{thm:inference}
  Under Assumptions \ref{ass:matching}-\ref{ass:observed} we have for $ \alpha \in (0, \frac{1}{2}) $, vanishing tuning parameter $ \tau_n \to 0 $,  $ \tau_n \sqrt{n} \rightarrow \infty $, and non-vanishing subpopulations $ \frac{n_\kappa}{n} \to c_\kappa > 0 $  for $ \kappa \in \{ c, m, f \} $:
\begin{equation}\label{eq:test}
 \liminf_{n\to\infty} \inf\limits_{\bar{\pi} \in\left\{ A\nu : \nu \geq 0 \right\} } \mathbf{P} \left( \mathcal{J}_n(\widehat{\bar{\pi}}_n,\tau_n \iota) \leq \widehat{F}^{-1}_{n,B,\mathcal{J}_n}(1-\alpha) \right) = 1 - \alpha  \text{.}
\end{equation}
\end{thm}
\begin{proof}
We have to check the conditions of \citet{KitamuraStoye2013}.
Their Assumption 4.1 requires that $ n_\kappa/n $ does not converge to zero. Here, the number of singles and couples must grow at the same rate so that no subpopulation vanishes.
For the analogue of their Assumption 4.2, random sampling of each observed distribution $ \bar{\pi}^\kappa $, we refer to Theorem \ref{pro:definetti} and measurability of $ \theta $, which, indeed, induces conditional i.i.d. marginals through the measurable function $ \Psi $. We conclude that the within-pool bootstrap consistently approximates the conditional law of $\sqrt n(\widehat{\bar{\pi}}_n-\bar{\pi})$. 
\end{proof}

Computing the test statistic requires repeated solution of a high-dimensional constrained quadratic problem.
Rather than relying on generic sequential quadratic programming routines used for solving inequality constrained problems\footnote{This algorithm is used for \texttt{lsqnonneg} (Matlab) and \texttt{optimize.nnls} (SciPy).}, we rewrite the problem as non-negative least squares, exploit the sparsity of $ A $, and implement a coordinate-wise projection method \citep{Franc2005, Johansson2006}.
Equation \eqref{eq:landweberstep} in Proposition \ref{lem:equivalences}.\ref{equ:nnls} defines the step and shows convergence.

Finally, in our simulation study, we find that the test has power to detect an ''irrational`` population of close to one with 500 observations per household composition if only 15\% of the population is not preference-stable. 
By doubling the sample size, the required proportion drops to 5\%. 
In addition, we discuss worst cases by considering ''similar configurations`` and show correct size under different worst-case samples.

 \section{Empirical Analysis}\label{sec:results}
In this section, we apply the test to three household panels that differ in data quality and measurement detail. 
Across all three datasets, the evidence points against stable preferences. Testing varying specifications, helps us understand the effects of price variation, sample size, and the assumptions of the model. 
We conclude by examining how these findings hold up against different extensions and robustness checks.

\subsection{Data}\label{subs:data}
For the test we consider households consisting of singles or couples.
We exclude households with children or other cohabiting groups of individuals who are not in a romantic relationship.
We consider a minimal setting with three periods and three goods, where we have $ 64 $ types of singles and $ 512 $ types of couples, resulting in $ 2,996 $ collectively rational preference-stable configuration types (see Example \ref{exa:types}).
Two of the panels we study are longer than necessary.
For transparency, we report results for different combinations of years. After dropping incomplete cases, we order the year triplets by the resulting sample size.
Due to attrition in panels, this pick out consecutive years.
We face the trade-off between sample size and price variation.\footnote{A discussion about the effectiveness of revealed preference methods with respect to price variation can be found in \citet{Crawford2011}.} 

First, we apply the test to the time use and consumption module \citep{Cherchye2012} from the Dutch LISS (Longitudinal Internet Studies for the Social Sciences) panel.
The panel is collected by CentERdata and consists of 5000 households and 8000 individuals, drawn from the population register of Statistics Netherlands. 
The survey is internet-based where households are provided with the necessary hardware to participate in the study.
Prices are obtained from the Dutch CPI for different consumption categories published by Eurostat (normalized to $ 100 $ for the year $2005$).
We select the private consumption categories: clothing, food \& beverages and recreation.

Second, we consider phase two of the Russian Longitudinal Monitoring Survey (RLMS), collected in form of personal interviews by the Carolina Population Center (University of North Carolina) and available for the years 1994 -- 2014.
Due to the amount of zeros observed for many private consumption expenditure categories, we focus on different categories of food. The survey distinguishes between 57 different food consumption categories, which we aggregate to dairy, bread and meat.
These three categories account for more than half of the food consumption, which itself takes a large proportion of total expenditure.\footnote{We make use of a weak separability assumption that is standard in the empirical demand estimation literature which allows us to be able to consider a subset of goods for estimation. Later, we relax this by allowing for endogeneity of expenditure on the selected goods.}
Price data is obtained from the Federal State Statistics Service (GKS) and available for the years: 2000, 2005, 2010, 2011 -- 2015.

Third, we use data from the Spanish Continuous Family Expenditure Survey (ECPF), collected by the Spanish statistics office (INE) on a quarterly basis for the period 1985 -- 2005.
The survey is designed in a way that participants are part of the sample for at most eight consecutive periods or two years.
There was a discontinuity in the design of the study in 1997, where the focus was shifted away from detailed consumption expenditure categories.
The ECPF was replaced by the Encuesta de Presupuestos Familiares (EPF) in 2006, where the collection frequency was extended to yearly with participation lifespan of two years being maintained.
Requiring a panel of at least three periods we, therefore, use data from the original ECPF from 1985 to 1996.
We select the same goods as in the LISS panel: clothing, food consumed outside of the household, and consumption of non-durables.
Price data is also published by INE.
Descriptive statistics can be found in Tables \ref{tab:descriptive.onebig.priv} and \ref{tab:descriptive.onebig.pub} in Appendix \ref{sec:descriptives}.

\subsection{Results}\label{subs:results}

Table \ref{tab:private.nomip.nohet} presents the baseline results in the form of p-values for different combinations of periods.
Rejection, indicated by a low p-value, corresponds to a violation of the stable-preference hypothesis.

{\singlespacing \begin{table}[ht!]
\caption{Results for Private Goods with Exogenous Prices}
\centering
\singlespacing
\small
\renewcommand{\arraystretch}{0.9}
\resizebox{.9\textwidth}{!}{
\begin{tabular}{cccccc}
\multicolumn{6}{c}{\textsc{Longitudinal Internet studies for the Social Sciences (LISS)}} \\
\toprule
Years  &  $ N_{\text{couples}}^{\text{total}} $  &  $ N_{\text{couples}}^{\text{rational}} $  &  $ N_{\text{singles}}^{\text{total}} $  &  $ N_{\text{singles}}^{\text{rational}} $  &  p-value \\
\midrule
2009   2010   2012  &  605  &  598  &  463  &  380  &  0.000 \\
[-7pt]\\

\multicolumn{6}{c}{\textsc{Russian Longitudinal Monitoring Survey (RLMS)}} \\
\toprule
Years  &  $ N_{\text{couples}}^{\text{total}} $  &  $ N_{\text{couples}}^{\text{rational}} $  &  $ N_{\text{singles}}^{\text{total}} $  &  $ N_{\text{singles}}^{\text{rational}} $  &  p-value \\
\midrule
2012   2013   2014  &  322  &  319  &  300  &  295  &  0.138 \\
2011   2013   2014  &  309  &  308  &  275  &  272  &  0.112 \\
2011   2012   2014  &  310  &  310  &  279  &  276  &  0.176 \\
2011   2012   2013  &  328  &  328  &  308  &  305  &  0.128 \\
2010   2013   2014  &  255  &  253  &  215  &  211  &  0.082 \\
2010   2012   2014  &  252  &  251  &  218  &  215  &  0.042 \\
2010   2012   2013  &  264  &  263  &  239  &  234  &  0.002 \\
2010   2011   2013  &  264  &  264  &  235  &  234  &  0.156 \\
2010   2011   2012  &  294  &  294  &  264  &  263  &  0.100 \\
2005   2011   2012  &  256  &  256  &  212  &  208  &  0.082 \\
[-7pt]\\

\multicolumn{6}{c}{\textsc{Spanish Continuous Family Expenditure Survey (ECPF)}} \\
\toprule
Years  &  $ N_{\text{couples}}^{\text{total}} $  &  $ N_{\text{couples}}^{\text{rational}} $  &  $ N_{\text{singles}}^{\text{total}} $  &  $ N_{\text{singles}}^{\text{rational}} $  &  p-value \\
\midrule
1994.3   1994.1   1994.2  &  106  &  104  &  5  &  3  &  0.076 \\
1993.4   1994.1   1994.2  &  108  &  97  &  8  &  5  &  0.056 \\
1992.2   1992.3   1992.1  &  93  &  89  &  15  &  14  &  0.004 \\
1990.4   1991.1   1991.2  &  95  &  91  &  14  &  12  &  0.004 \\
1989.3   1989.1   1989.2  &  107  &  103  &  4  &  3  &  0.024 \\
1988.2   1988.3   1988.4  &  96  &  93  &  6  &  3  &  0.498 \\
1987.1   1987.2   1987.3  &  124  &  120  &  8  &  6  &  0.128 \\
1986.4   1987.1   1987.2  &  154  &  149  &  9  &  5  &  0.440 \\
1986.3   1986.4   1987.1  &  129  &  121  &  7  &  6  &  0.006 \\
1986.3   1986.4   1986.2  &  125  &  124  &  9  &  9  &  0.050 \\
[-7pt]\\

\bottomrule
\end{tabular}}
\label{tab:private.nomip.nohet}
\vspace{-\baselineskip}
\floatnote{Number of total couples, rational couples according to aggregate CARP, total singles and rational singles according to GARP, for different combinations of periods. Sampling for the ECPF is quarterly, for which we use the \texttt{year.quarter} notation.}
\vspace{-\baselineskip}
\end{table}
 }

There is strong evidence to reject the stable-preference hypothesis for the LISS panel, for the RLMS and ECPF there are some combinations of periods for which there is not enough evidence to arrive at this conclusion.
In these non-rejection cases, we either have three consecutive years in which we are faced with limited power of revealed preference axioms due to the lack of price variation, or a particularly small sample size due to the wider span of considered years in combination with attrition.
This all points towards the trade-off discussed above.
For the RLMS, the food-bundle specification is, perhaps, more prone to habit formation. The test rejects there only at the 10\% level.
The sample for the ECPF is very small, particularly for single households.
Abstracting from the inferior statistical properties of the test in small samples (we still have numerical convergence), the strong rejection of the hypothesis may reflect a finite-sample support issue, i.e. it is harder to rationalise the choice distributions when we observe zero probability for some single types.

{ \singlespacing \begin{table}[ht!]
\caption{Results Conditional on Demographics for Private Goods, Exogenous Prices}
\centering
\singlespacing
\small
\renewcommand{\arraystretch}{0.9}
\resizebox{\textwidth}{!}{
\begin{tabular}{cccccccc}
\multicolumn{8}{c}{\textsc{Longitudinal Internet studies for the Social Sciences (LISS)}} \\
\toprule
Years  &  College  &  Age  &  $ N_{\text{couples}}^{\text{total}} $  &  $ N_{\text{couples}}^{\text{rational}} $  &  $ N_{\text{singles}}^{\text{total}} $  &  $ N_{\text{singles}}^{\text{rational}} $  &  p-value \\
\midrule
2009   2010   2012  &  1  &  2  &  92  &  90  &  79  &  71  &  0.138 \\
2009   2010   2012  &  1  &  1  &  163  &  162  &  93  &  75  &  0.000 \\
2009   2010   2012  &  1  &  0  &  65  &  65  &  64  &  56  &  0.142 \\
2009   2010   2012  &  0  &  2  &  137  &  134  &  121  &  97  &  0.000 \\
2009   2010   2012  &  0  &  1  &  114  &  113  &  85  &  65  &  0.002 \\
2009   2010   2012  &  0  &  0  &  34  &  34  &  21  &  16  &  0.004 \\
[-7pt]\\

\multicolumn{8}{c}{\textsc{Russian Longitudinal Monitoring Survey (RLMS)}} \\
\toprule
Years  &    &  Age  &  $ N_{\text{couples}}^{\text{total}} $  &  $ N_{\text{couples}}^{\text{rational}} $  &  $ N_{\text{singles}}^{\text{total}} $  &  $ N_{\text{singles}}^{\text{rational}} $  &  p-value \\
\midrule
2010   2011   2012  &    &  2  &  77  &  77  &  119  &  114  &  0.010 \\
2010   2011   2012  &    &  1  &  161  &  160  &  105  &  105  &  0.122 \\
2010   2011   2012  &    &  0  &  56  &  56  &  40  &  39  &  0.026 \\
2011   2012   2013  &    &  2  &  77  &  77  &  136  &  130  &  0.232 \\
2011   2012   2013  &    &  1  &  193  &  190  &  127  &  127  &  0.026 \\
2011   2012   2013  &    &  0  &  55  &  54  &  42  &  42  &  0.020 \\
2012   2013   2014  &    &  2  &  75  &  75  &  133  &  128  &  0.000 \\
2012   2013   2014  &    &  1  &  196  &  191  &  131  &  128  &  0.002 \\
2012   2013   2014  &    &  0  &  49  &  49  &  34  &  33  &  0.024 \\
[-7pt]\\

\bottomrule
\end{tabular}}
\label{tab:private.nomip.het}
\vspace{-\baselineskip}
\floatnote{Number of total couples, rational couples according to aggregate CARP, total singles and rational singles according to GARP, for consecutive periods and demographics.}
\vspace{-\baselineskip}
\end{table}
 }

Next, Table \ref{tab:private.nomip.het} reports results of the test which takes into account matching on education and age, by conditioning on these observed demographics.
We limit ourselves to consecutive years, in which these demographics are likely to be stable over time.\footnote{In large enough samples, this could easily be carried out with time-dependent demographics, however such an approach would suffer from the curse of dimensionality, and we already have to manage a relatively small sample size for the RLMS.}
While we observe these characteristics for both spouses, we only look at assortatively matched couples, i.e. couples in which individuals fall into the same education and age category.\footnote{For non-assortative matching, we would have to change the randomisation procedure over type transpositions to one that retains the observed matching pattern. For example, for a couple with a low educated man, and a high educated woman we would only consider swaps with single men and women within the same respective demographic category. The theory goes through.}
Assortative matching accounts for almost all of the couples in the sample.
We define education as a binary variable indicating whether the individual has completed higher education (college) or not, and age as a categorical variable with three classes: $ \{0: \text{age} < 40, \; 1: \; 40 \leq \text{age} \leq 60, \; 2: \; \text{age} > 60 \} $.
We repeat the analysis, for all combinations of education and age, only for the LISS panel and RLMS. Rejections go through across all non-college subpopulations whereas only the middle-aged subsample of college-educated couples rejects.

\subsection{Extensions and Robustness}\label{sec:extensions}

In this section we discuss how we can weaken the no-consumption-externalities assumption, deal with endogeneity of budgets, and discuss a different characterisation of the collective axiom.
We elaborate on the corresponding empirical results and refer the interested reader to the tables in Appendix \ref{sec:robustness}. 

\paragraph{Public Goods}
To account for arbitrary consumption externalities, we augment the random utility and matching model with a \citet{Barten1964} linear consumption technology, famously adapted to the collective model by \citet{BCL2013}.
Then household $ (i, \sigma(i)) $ maximises
\begin{align}\label{eq:collectiveBarten}
  \max_{x^f, x^m } &\left\{\lambda_{i\sigma(i)}(p_t) u^m_{i}(x^m) + (1-\lambda_{i\sigma(i)}(p_t)) u^f_{i}(x^f) \right\} \\
  \text{s.t.} \;  & \; D ( \optionaltilde{x}^f + \optionaltilde{x}^m ) \in B_{i\sigma(i)t} = \left\{ \optionaltilde{x} \; | \; p^\top_t \optionaltilde{x} \leq w_{i\sigma(i)t} \right\} 
\end{align}
where $ D $ is a production technology matrix.
It is commonly assumed to be diagonal, restricting complementarities between consumption externalities. 
Its elements range from $ 0.5 $ for an entirely public good for which both individuals pay half of the prices, to $ 1.0 $ for an entirely private good for which individuals pay market prices.
The production technology matrix $ D $ is not identified without restrictions on heterogeneity or the functional form of utilities.
Thus we will calibrate $ D $ from estimates of \citet{Cherchye2017}, a study conducted using the LISS panel.\footnote{A promising approach is adopted by \citet{Gauthier2025} who uses the assignable consumption in the LISS panel, to extend the collective axiom to incorporate a household production function. One could, theoretically, use these extended axioms to allow for more general forms of the production technology, in line with the non-parametric nature of the test. We leave this for future research.}

For the empirical specification, we select the aggregate goods \emph{housing, transport, and energy}, with corresponding Barten scales: $ \text{diag}(D) = (0.683, 0.692, 0.748)$,
which, arguably, represent goods subject to consumption externalities taking values about half way on the spectrum from public to private.
They are equally available in the LISS and the RLMS.
For the ECPF we use a hybrid specification using \emph{clothing, transportation, and petrol}, with Barten scales: $\text{diag}(D) = (1.00, 0.683, 0.748) $.
We obtain house price indices (HPI) from the same sources as the respective CPIs.

Table \ref{tab:public.nomip.nohet} (no demographics) and Table \ref{tab:public.nomip.het} (by education and age) in Appendix \ref{sec:robustness} report the results. 
Despite the much larger dataset, which appears due to fewer boundary cases, the evidence is not as clear as for the private goods.
This could be due to the additional homogeneity restriction imposed by the Barten technology, which is assumed to be the same for all households.
Despite this, we still reject the stable preference hypothesis for most datasets for a 10\% significance level.

\paragraph{Endogeneity of Total Expenditure}
Total expenditure may be endogenous.
In particular, we may think of it as determined by household income $ y_{i\sigma(i)} $ and an unobserved taste shifter $ \zeta_{i\sigma(i)} = \zeta(\omega^m_i, \omega^f_{\sigma(i)}) $, a one-dimensional summary of preferences capturing the household's propensity to allocate resources toward the goods we study:
\begin{equation}
w_{i\sigma(i)} = g(y_{i\sigma(i)}, \zeta_{i\sigma(i)}).
\end{equation}
The relationship is unconstrained, other than $g$ being strictly increasing in its second argument. 
Income may itself depend on $(\omega^m_i, \omega^f_{\sigma(i)})$ through channels such as labour supply and human capital. 
We follow the standard assumption that $ \zeta_{i\sigma(i)} $ is orthogonal to those channels. 
Following \citet{Imbens2009}, we exploit the monotonicity of $g$ to define the control function $v_{i\sigma(i)}$ as the rank of total expenditure, given income, through the conditional CDF: \begin{equation}
v_{i\sigma(i)} = F_{W \mid Y}\left(w_{i\sigma(i)} \mid y_{i\sigma(i)}\right) = F_{\zeta \mid Y}\left(\zeta_{i\sigma(i)} \mid y_{i\sigma(i)}\right).
\end{equation}
Conditioning on this control function absorbs the endogenous component of total expenditure.
We obtain estimates using the empirical conditional distribution function:
$ \hat{v}_{i\sigma(i)} \;=\; \widehat{F}_{W \mid Y} (w_{i\sigma(i)} \mid y_{i\sigma(i)})$.
Since this quantity is continuous, we implement conditioning by a kernel.
In particular, for a grid $ v_0 \in \{0.05, 0.15, \ldots, 0.95 \} $ and bandwidth $ h = \frac{1}{20}$, we report the test statistic and corresponding p-values for the sub-samples $ \{(i, \sigma(i)) : \hat{v}_{i\sigma(i)} \in [v_0 - h,\, v_0 + h]\} $.

The results are shown in Table \ref{tab:ec.nomip.nohet} of Appendix \ref{sec:robustness}.
Conditioning on small cells substantially reduces the effective sample size and removes much of the variation used by the test, so these results should be interpreted as conservative.
For the LISS panel, we find that we still reject for more than half of these sub-samples in the private good case, but in only about a quarter of the cases for public goods.

\paragraph{Conditioning on Observable Resource Shares}
Whenever resource sharing is observed, as in the LISS panel, we may split the sample into brackets of similar individual expenditure on private goods and match partnered individuals only with singles in the corresponding expenditure bracket.
In such a setting, singles are then used as benchmarks for partnered individuals only if they have similar private expenditure levels.
This is best interpreted as a robustness exercise.
It checks whether the rejection is driven by comparing households at very different individual budget levels.

Table \ref{tab:private.nomip.rs} in Appendix \ref{sec:robustness} reports the results for the LISS panel. Based on private expenditure terciles we classify their private expenditure categories on the selected goods as low, mid, and high. We only look at equal-splitting couples which account for most of the sample. While they would be economically interesting cases, we do not report off-diagonals, in which there is unequal splitting between the spouses, due to the small sample size and the arising curse of dimensionality.
The results reveal that the test rejects for the mid- and high-expenditure subsamples but not for the low-expenditure subsample. 
The non-rejection at the low end is consistent with low-spending households allocating expenditure on these goods toward necessities, where choices are largely determined by budget rather than taste.

\paragraph{Mixed Integer Programming Approach}
In this paper, we used restrictions from Definition \ref{thm:necessary} to define collectively rational household types.
One might use a stronger characterisation based on both spouses satisfying individual GARP.
Such a characterisation relies on recovering feasible quantities $ \check{x}_{i\sigma(i)}^m $ and $ \check{x}_{i\sigma(i)}^f $ for each household.
\citet{Cherchye2011} provide a mixed integer programming procedure to recover these individualised quantities.
We also implement their procedure.
It allows us to check whether personalised quantities exist and, thus, if a given household can be rationalised. 
Beyond that, it can characterise a couples' revealed preference type solely on the aggregate-choice GARP partition, with no further structure available to interact with singles' types in a configuration.
Thus, all the bite of the restrictions from the collective model is used in the pre-testing, and cannot be exploited further in conjunction with the single's preferences, in the same way as the baseline characterisation.
This will result in a larger set of preference-stable configurations, and thus, easier rationalisability of the observed choice distributions.
We might, thus, expect less power of the test to detect violations of the stable preference hypothesis.
The results, reported in Table \ref{tab:mixedinteger} in Appendix \ref{sec:robustness}, confirm this conjecture for private goods but the test does remarkably well for public goods.

 \section{Conclusion}
This paper asks whether the preferences individuals reveal as singles can also explain their behaviour in couples. 
The comparison is not immediate because singles are observed as unitary households, whereas couples are observed only through joint choices.
With preference heterogeneity, stability is not a household-by-household restriction but a population restriction: matching may affect who becomes single or partnered, and who is matched with whom, but under the null it must not change the distribution of underlying preferences across the single and partnered pools.

Although we observe only the separate demand distributions of couples, single men, and single women, the null restricts the latent structure that can rationalise them jointly. 
We place these objects in a common framework in which each observed couple is compared with counterfactual single men and women.
Such configurations are admissible only if collective rationality is preserved after replacing the couple's latent individual preferences by the preferences revealed by those singles. 
Preference stability requires the observed demand distributions to admit a rationalisation using only admissible configurations. 
If no such rationalisation exists, the preferences revealed by singles cannot rationalise the behaviour of couples.

We apply the test to the Dutch LISS, the Russian RLMS, and the Spanish ECPF. 
In the baseline specification with private goods and exogenous expenditure, all three datasets provide evidence against preference stability. 
The rejection largely remains after conditioning on observed demographics, where sufficient observations are available, while specifications allowing public goods or endogenous expenditure produce more mixed results.
 \bibliography{submission}
\appendix
\section{Proofs}\label{app:proofs}
\subsection{Proof of Lemma \ref{thm:collectiveRUM}}\label{app:collectiveRUM}
We start with a more general model in which cardinal utility might depend on the partner. For this we let 
$u_{ij} = g^c_{ij}(u_i) = a_{ij} u_i + b_{ij} $, where $a_{ij} > 0$.
From e.g. \citet{AppsRees1997}, the men's problem can be written as a unitary problem with the female partner's reservation utility as a constraint:
\begin{equation}
  (x_{it}^{m*}, x_{it}^{f*}) = \arg\max_{(x^m,x^f)\in\mathcal{X}^2} u^m_{ij}(x^m) 
\quad\text{s.t.}\quad 
p(x^m+x^f) \le w_{ij},
\quad 
u^f_{ij}(x^f) \ge u^f_{ij}(x^{f*}_{ij}).
\end{equation}
Letting $ u_{0,ij}^* = u(\omega^f_j,x^{f*}_{ij}) $ we write the Lagrangian
\begin{equation}
\mathcal{L}(x^m,x^f,\mu_{ij},\rho_{ij};\omega) 
= g^c_{ij}(u(\omega^m_i,x^m)) 
+ \mu_{ij}(w_{ij} - p(x^m+x^f)) 
+ \rho_{ij}(g^c_{ij}(u(\omega^f_j,x^f)) - g^c_{ij}(u_{0,ij}^*) ).
\end{equation}
For optimal consumption $ (x^{m*}_{ij}, x^{f*}_{ij}) $, the Lagrange multipliers $( \mu_{ij}, \rho_{ij})$ satisfy
\begin{equation}
  \nabla_u g^{c}_{ij}(u(\omega^m_i, x^{m*}_{ij}))\nabla_{x^m} u(\omega^m_i, x^{m*}_{ij}) = \mu_{ij} p,
\end{equation}
\begin{equation}
  \rho_{ij}\nabla_u g^{c}_{ij}(u(\omega^f_j, x^{f*}_{ij}))\nabla_{x^f} u(\omega^f_j, x^{f*}_{ij}) = \mu_{ij} p.
\end{equation}
Equating the right-hand sides:
\begin{equation}
  \label{eq:rhoij}
  \nabla_ug^{c}_{ij}(u(\omega^m_i, x^{m*}_{ij}))\nabla_{x^m} u(\omega^m_i, x^{m*}_{ij}) 
= \rho_{ij}\nabla_u g^{c}_{ij}(u(\omega^f_j, x^{f*}_{ij}))\nabla_{x^f} u(\omega^f_j, x^{f*}_{ij}).
\end{equation}
Since $\nabla_u g^{c}_{ij}(z) = a_{ij}$ is constant, and we can simplify equation \eqref{eq:rhoij} to
\begin{equation}
  \nabla_{x^m} u(\omega^m_i, x^{m*}_{ij}) = \rho_{ij}\nabla_{x^f} u(\omega^f_j, x^{f*}_{ij}).
\end{equation}
Thus, we can solve for $\rho_{ij}$ which does not depend on $a_{ij}$ or $b_{ij}$, and thus $ g^c_{ij} $. 
Defining the Pareto weight $ \lambda_{ij} = \lambda(\omega_i^m, \omega^f_j) = \frac{1}{1 + \rho_{ij}} $ such that $ \rho_{ij} = \frac{1 - \lambda_{ij}}{\lambda_{ij}} $, the $ \lambda_{ij} $-re-weighted version of the total derivative of the Lagrangian becomes 
\begin{equation}
  \lambda_{ij} \nabla_{x^m} u(\omega^m_i,x^m) dx^m
  + (1-\lambda_{ij}) \nabla_{x^f} u(\omega^f_j, x^f) dx^f
  + \kappa_{ij} p (dx^m+dx^f)  = 0 
\end{equation}
where $ \kappa_{ij} \equiv \lambda_{ij} \mu_{ij} $, which coincides with the problem in equation \eqref{eq:collectiveutility}.\hfill\qed

\subsection{Corollary \ref{pro:unitaryrep}}\label{app:proofOfunitaryrep}
\begin{cor}[Representation of Unitary Utilities]
\label{pro:unitaryrep}
For any individual $i$, on either side of the market and any dummy partner $j\in\mathbb{N}_0$, consider the dyad $(i,j)$ with the dummy partner's private bundle fixed at $0\in\mathcal{X}$. 
Under Assumptions \ref{ass:dependence}, \ref{ass:collective}, there exist strictly increasing, affine functions $h_{i},h_j $ such that
\begin{equation}\label{eq:unitary}
  \Phi_{i0}(x^m,0,p)=h_{i}(u_i(x^m), p) \qquad \text{and} \qquad \Phi_{0j}(0,x^f,p) = h_j(u_j(x^f), p). 
\end{equation}
\end{cor}
\begin{proof}
Fix $i$, prices $ p $ and income $ w $. Let $x^m \in \mathcal{X}$ be individual $i$'s private consumption. 
Since $x$ is private to $i$ the $j$-side terms in household utilities do not depend on $x$ and we can set $ x^f = 0 $:
\begin{equation}
\Phi_{ij}(x^m, x^f) = \lambda(\omega^m_i, \omega^f_j)  u(\omega^m_i, x^m ) + C(\omega^m_i, \omega^f_j, x^f),
\end{equation}
For any $ x_0^m, x^m_1,x^f_0 \in \mathcal{X} $, we have $ u_i(x^m_1; \omega^m_i) - u_i(x^m_0; \omega^m_i) > 0 $ if and only if 
$$
\Phi_{ij}(x^m_1,x^f_0) - \Phi_{ij}(x^m_0,x^f_0) = \lambda(\omega^m_i, \omega^f_j) [ u_i(x_1^m) - u_i(x_0^m) ] > 0,
$$
because $\lambda(\omega^m_i, \omega^f_j) > 0$ and $C_{ij}$ cancels. 
Both $ \Phi_{ij} $ and $ u_i $ are differentiable in $ x^m $. 
Setting $ x^m_1 = x^m_0 + h $ and dividing both sides by $ h > 0 $ and taking limits yields:
$$ \frac{d}{dx^m} \Phi(x^m, x^f;\omega^m_i,\omega^f_j) = \lambda(\omega^m_i, \omega^f_j)  \frac{d}{dx^m} u(x^m; \omega^m_i). $$
Integrating both sides over $ x^m $ and setting $ \omega^f_j = \omega^0 $, since $ j \in \mathbb{N}_0 $ proofs the claim. 
\end{proof}

\subsection{Proof of Lemma \ref{lem:perm}}\label{app:proofOfperm}
Let $ \mathbb{N}(\omega)=(\mathbb N_2(\omega),\mathbb N_1(\omega), \mathbb N_0)  $ be a matching allocation for a given population induced by the assignment matrix $ M $. The realised partition for our population is denoted as $\mathbb{N} $. 
By Assumption \ref{ass:dependence}, the primitive environment is exchangeable ex-ante with respect to the group $ G_0\times G_0 $ (every real individual's preferences is not tied to their identity). By permutation equivariance of the matching rule by Assumption \ref{ass:matching}, relabelling the primitive environment only relabels the matching outcome and hence the induced partition. In particular,
$$ \mathbb{N}((\varsigma^m,\varsigma^f)\cdot\omega) = (\varsigma^m,\varsigma^f)\cdot \mathbb{N}(\omega).
$$

Now, once the realised partition $\mathbb{N} $ is fixed, only those relabellings that preserve $\mathbb{N} $ remain admissible. These are exactly the elements of the stabiliser:
$$
G = \left\{ \varsigma\in G_0: \varsigma(\mathbb N_\ell)=\mathbb N_\ell,\ \ell=0,1,2 \right\}.
$$
The stabiliser is the subgroup fixing the realised partition, whose orbit is the set of all relabelled partitions \citep[p. 56]{rotman1994}.
If $ (\varsigma^m,\varsigma^f)\in G\times G $, then the event $\mathbb{N}(\omega)=\mathbb{N}$ is unchanged by relabelling, so the primitive exchangeability from Assumption \ref{ass:dependence} carries over within the realised pools. This gives part (i).

By Assumption \ref{ass:matching}, the matching rule $M$ is permutation equivariant, and the surplus map $\Phi$ inherits the same relabelling from the primitive environment. Hence the composition $ \sigma_\omega \equiv M(\Phi(\omega))$ is permutation equivariant under $G_0\times G_0$. After fixing the partition $\mathbb{N} $, this again restricts to the subgroup $G\times G$, proving part (ii).

No larger symmetry is available in general. 
Once the realised partition is fixed, individual identities remain irrelevant within a given pool, but not across pools. 
Any relabelling outside $ G $ moves at least one individual from single to partnered or vice versa, and therefore changes the selection into household type. 
It maps \(\{\mathbb N(\omega)=\mathbb N\}\) to a different conditioning event. 
Thus, after the realised partition is fixed, the surviving symmetry group is precisely the within-pool relabelling group, namely the stabiliser $ G\times G $. \hfill \qed

\subsection{Proof of Lemma \ref{lem:prefstable}}\label{app:proofOfprefstable_slutsky}

By \citet{Chiappori2009} we can split up the collective problem into two stages. In the first stage households agree on the male resource share $w^m(p,w)$. In the second stage, they solve an individual standard consumption problem with endowment $ w^m $ and $ w^f $, respectively. Denoting the corresponding solutios as $ x^m $ and $ x^f $, we can write aggregate demand as:
$$ x(p,w) = x^m(p,\,w^m(p,w)) + x^f(p,\,w-w^m(p,w)). $$
Following \citet{Browning1998}, we write the pseudo-Slutsky matrix for the household as:
$$
  \bar S(p,w) \equiv \frac{\partial x}{\partial p^\top}(p,w)
             + \frac{\partial x}{\partial w}(p,w)\,x(p,w)^\top.
$$

\noindent Differentiating demands we obtain
$$
  \frac{\partial x}{\partial p^\top}
  = \frac{\partial x^m}{\partial p^\top}
  + \frac{\partial x^f}{\partial p^\top}
  + (\frac{\partial x^m}{\partial w^m}
          - \frac{\partial x^f}{\partial(w-w^m)})\frac{\partial w^m}{\partial p^\top}
$$
and
$$
  \frac{\partial x}{\partial w}
  = \frac{\partial x^m}{\partial w^m}\,\frac{\partial w^m}{\partial w}
  + \frac{\partial x^f}{\partial(w-w^m)}\,
    (1-\frac{\partial w^m}{\partial w}).
$$
Hence
\begin{align*}
  \bar S
  &= \frac{\partial x^m}{\partial p^\top}
   + \frac{\partial x^f}{\partial p^\top}
   + (\frac{\partial x^m}{\partial w^m}
           - \frac{\partial x^f}{\partial(w-w^m)})\frac{\partial w^m}{\partial p^\top} \\
  &\quad + [
      \frac{\partial x^m}{\partial w^m}\,\frac{\partial w^m}{\partial w}
    + \frac{\partial x^f}{\partial(w-w^m)}\,
      (1-\frac{\partial w^m}{\partial w})
    ](x^m+x^f)^\top .
\end{align*}
Expanding the last term and rearranging yields the individual Slutsky matrices:
$$ \bar S^m(p,w^m) \equiv \frac{\partial x^m(p,w^m)}{\partial p^\top} + \frac{\partial x^m(p,w^m)}{\partial w^m}\,x^m(p,w^m)^\top $$
and
$$ \bar S^f(p,w-w^m) \equiv \frac{\partial x^f(p,w-w^m)}{\partial p^\top} + \frac{\partial x^f(p,w-w^m)}{\partial(w-w^m)}\,x^f(p,w-w^m)^\top. $$
\noindent Consequently, the household pseudo-Slutsky matrix can be written as:
$$ \bar S(p,w) = \bar S^m(p,w^m) + \bar S^f(p,w-w^m) + U(p,w)V(p,w)^\top, \label{eq:decomposition} $$
where
$$ U(p,w) \equiv \frac{\partial x^m}{\partial w^m} - \frac{\partial x^f}{\partial(w-w^m)} $$
and
$$ V(p,w)^\top \equiv \frac{\partial w^m}{\partial p^\top} + \frac{\partial w^m}{\partial w}(x^f)^\top - (1-\frac{\partial w^m}{\partial w})(x^m)^\top. $$

\noindent Now define the relative male resource share $ \eta(p,w) \equiv \frac{w^m(p,w)}{w}\in(0,1)$.
By Lemma \ref{thm:collectiveRUM}, the Pareto weight is homogeneous of degree zero in $(p,w)$. Hence a proportional rescaling $(p,w)\mapsto(tp,tw)$ leaves the real budget set and the Pareto weight and the corresponding resource allocation unchanged. Under Walras' law, the male budget share is exhausted, so $w^m(p,w)=p^\top x^m(p,w)$. It follows immediately that $w^m(tp,tw)=t w^m(p,w)$. Thus $w^m$ is homogeneous of degree one, and $\eta(p,w)=w^m(p,w)/w$ is homogeneous of degree zero.
Writing budget-normalised prices as $q=p/w$, we may write $\eta(p,w)=\eta(q,1)$, and denote this value by $\eta(q)$.

By homogeneity of Marshallian demands $ x^m $ and $ x^f $, the individual Slutsky matrices are homogeneous of degree $ -1 $.\footnote{Take $x(\alpha p, \alpha w) = x(p,w)$ to be homogeneous of degree zero. Differentiating yields, by the chain rule, $\frac{\partial x}{\partial p^\top}(\alpha p, \alpha w) = \frac{1}{\alpha}\frac{\partial x}{\partial p^\top}(p,w)$, and similarly for $\partial x/\partial w$.}
For each $ \kappa \in\{m,f\} $, define the unit-budget individual Slutsky
matrix by  $ S^\kappa(r)\equiv \bar S^\kappa(r,1) $, where $ r $ denotes the individual unit-budget price vector. Then, for any
individual expenditure $y>0$,
$$
  \bar S^\kappa(p,y)
  =
  \frac{1}{y}\,
  \bar S^\kappa\!\left(\frac{p}{y},1\right)
  =
  \frac{1}{y}\,
  S^\kappa\!\left(\frac{p}{y}\right).
$$
With $q=p/w$, $w^m=w\eta(q)$, and $w^f=w-w^m=w(1-\eta(q))$, it follows that
$$
  \bar S^m(p,w^m)
  =
  \frac{1}{w\eta(q)}\,
  S^m\!\left(\frac{q}{\eta(q)}\right),
  \qquad
  \bar S^f(p,w^f)
  =
  \frac{1}{w(1-\eta(q))}\,
  S^f\!\left(\frac{q}{1-\eta(q)}\right).
$$

The income-effect vector $ U(p,w) $ is also homogeneous of degree $ -1 $ making  $wU(p,w) $ homogeneous of degree zero. 
Similarly, since $ w^m(p,w) $ is homogeneous of degree one, its derivatives with respect to  $ p $ and $ w $ are homogeneous of degree zero. Together with homogeneity of individual demands, this implies that  $V(p,w)$ is homogeneous of degree zero.
Thus both $ wU(p,w)$ and $ V(p,w) $ depend on $ (p,w) $ only through
$ q=p/w $. We therefore define $ u(q) \equiv wU(p,w) $ and $ v(q) \equiv V(p,w) $ such that:
$$
  w\,U(p,w)V(p,w)^\top = u(q)v(q)^\top .
$$
Multiplying equation \eqref{eq:decomposition} by $w$ therefore yields the unit-budget representation
$$
  S(q)\equiv w \bar S(p,w)
  =
  \frac{1}{\eta(q)}\,
  S^m\left(\frac{q}{\eta(q)}\right)
  +
  \frac{1}{1-\eta(q)}\,
  S^f\left(\frac{q}{1-\eta(q)}\right)
  +
  u(q)v(q)^\top.
$$

\noindent Now take a preference-stable configuration $(i,j,i',j')$. By definition, $ \omega_i^m=\omega_{i'}^m $ and $ \omega_j^f=\omega_{j'}^f $.
The primitive $\omega_i^m$ determines the individual utility $u_i^m$, and $\omega_j^f$
determines $u_j^f$. Under collective rationality, these same primitives jointly determine
the Pareto weight $\lambda_{ij}$ and hence the induced share $\eta_{ij}$. Therefore the
factual individuals $i$ and $j$ and the counterfactual singles $i'$ and $j'$ share the
same individual demand systems, evaluated at the same normalised individual budgets. Hence
$$
  S_i^m\left(\frac{p}{\eta_{ij}(p)}\right)
  =
  S_{i'0}\left(\frac{p}{\eta_{ij}(p)}\right),
  \qquad
  S_j^f\left(\frac{p}{1-\eta_{ij}(p)}\right)
  =
  S_{0j'}\left(\frac{p}{1-\eta_{ij}(p)}\right).
$$
Writing everything on the unit budget therefore gives
\begin{align*}
  S_{ij}(p)
  &=
  \frac{1}{\eta_{ij}(p)}\,
  S_{i'0}\left(\frac{p}{\eta_{ij}(p)}\right)
  +
  \frac{1}{1-\eta_{ij}(p)}\,
  S_{0j'}\left(\frac{p}{1-\eta_{ij}(p)}\right)
  +
  u_{ij}(p)v_{ij}(p)^\top \\
  &= \bar S_{i'0}(p,\eta_{ij}(p)) + \bar S_{0j'}(p,1-\eta_{ij}(p)) + u_{ij}(p)v_{ij}(p)^\top. 
\end{align*}
which is the final representation.\hfill \qed

\subsection{Lemma \ref{lem:proj_exchangeable}}\label{app:proofOfprojection}
\begin{lem}\label{lem:proj_exchangeable}
Let Assumptions \ref{ass:matching}-\ref{ass:collective} hold. If configurations are sampled $ (\tau^m, \tau^f) \sim \rho $, independent of $ \omega \sim \mu $, and label-invariant, the sequence of configurations, defined in equation \eqref{eq:chi}, is within-pool-exchangeable: $ \{ \chi_i \}_{i \in \mathbb{N}_2} \stackrel{d}{=} \{ \chi_{\varsigma(i)} \}_{i \in \mathbb{N}_2} $ for $ \varsigma \in G $.
\end{lem}
\begin{proof}
By Lemma \ref{lem:perm}, 
\begin{equation}\label{eq:equicomp}
\sigma_{\omega'} \equiv (M \circ \Phi)(\omega') = (M \circ \Phi)(\varsigma \cdot \omega) = \varsigma \cdot (M \circ \Phi) (\omega) = \varsigma \cdot \sigma_\omega = (\varsigma^f)^{-1} \circ \sigma_\omega \circ \varsigma^m,
\end{equation}
where the last equality follows for fixed $\omega$ from:  $m_{ij}' = m_{\varsigma^m(i), \varsigma^f(\sigma_\omega(i))} = \delta_{\varsigma^f(j), \sigma_\omega(\varsigma^m(i))} = \delta_{j, (\varsigma^f)^{-1}(\sigma_\omega(\varsigma^m(i)))} = m_{i, ((\varsigma^f)^{-1} \circ \sigma_\omega \circ \varsigma^m)(j)} $.
We start with the couples pool.
We show that the matched-pair sequence $\{(\omega_i^m,\omega^f_{\sigma_\omega(i)})\}_{i\in\mathbb N_2}$ is $G$-exchangeable.
For $ i\in \mathbb{N}_2$ the matched-pair sequence for the couple's pool is $ (\omega)_i \equiv (\omega^m_i, \omega^f_{\sigma_\omega(i)}) $ which we compare to the one under $\omega'$:
\begin{equation}
\omega'_i = (\omega^{m'}_i, \omega^{f'}_{\sigma_{\omega'}(i)}) = (\omega^m_{\varsigma^m(i)}, \omega^f_{\varsigma^f(\sigma_{\omega'}(i))}).
\end{equation}
Using equivariance from Lemma \ref{lem:perm}, by associativity and identity of composition,
\begin{equation}
  (\varsigma^f \circ \sigma_{\omega'})(i) = (\varsigma^f \circ ((\varsigma^f)^{-1} \circ \sigma_\omega \circ \varsigma^m)) (i) = \sigma_\omega(\varsigma^m(i)),
\end{equation}
from which we get
\begin{equation}
\omega'_i = (\omega^m_{\varsigma^m(i)}, \omega^f_{\sigma_\omega(\varsigma^m(i))}) = \omega_{\varsigma^m(i)} = \varsigma^m \cdot \omega \; a.s.
\end{equation}
Since $\omega' \stackrel{d}{=} \omega$, applying the measurable map $E$ yields $\omega' \stackrel{d}{=} \omega$, we conclude
\begin{equation}
\omega \stackrel{d}{=} \varsigma^m \cdot \omega.
\end{equation}

For the individual terms, for fixed $\tau^m,\tau^f\in G$, $\{\omega^m_{\tau^m(i)}\}_{i\in\mathbb N_2}$ and $\{\omega^f_{\tau^f(\sigma_\omega(i))}\}_{i\in\mathbb N_2}$ are $G$-exchangeable.
For men, we have a fixed $ \tau^m  $, and action $ \varsigma^m \in G $.
Thus, by Assumption \ref{ass:dependence}, the permutation $ \tau^m \circ \varsigma^m \circ \tau^m \in G $, and we have $ \omega^m \stackrel{d}{=}  (\tau^m \circ \varsigma^m \circ \tau^m)  \cdot \omega^m $.
Reindexing both sides,\footnote{For composition of actions $ g, h \in G $, we have $ (g\cdot(h\cdot \omega^m))_i =(h\cdot \omega^m)_{g(i)} =\omega^m_{h(g(i))} =((h\circ g)\cdot \omega^m)_i $.} 
$$ \tau^m \cdot \omega^m \stackrel{d}{=} ( (\tau^m \circ \varsigma^m \circ \tau^m ) \circ \tau^m ) \cdot \omega^m  = (\tau^m \circ \varsigma^m) \cdot \omega^m = \varsigma^m \cdot (\tau^m \cdot \omega^m ) . $$
For women we also fix $ \tau^f $, and note that $G$-exchangeability of the matched-pair sequence $\omega_i $ implies the same for each coordinate.
Taking the second coordinate projection we have $ \sigma_\omega \cdot \omega^f \stackrel{d}{=} (\sigma_\omega \circ \varsigma^m) \cdot \omega^f $.
Applying $ \tau^f $ to  $ \omega^f $ on both sides, we get 
$$ (\tau^f \circ \sigma_\omega ) \cdot \omega^f \stackrel{d}{=} (\tau^f \circ (\sigma_\omega \circ  \varsigma^m)) \cdot \omega^f = \varsigma^m \cdot ((\tau^f \circ \sigma_\omega) \cdot \omega^f). $$

Finally, to complete the proof and show exchangability of configurations, we not that we have label-invariance of $ (\tau^m, \tau^f) \in \mathcal{T} $. Putting it together, for any $\varsigma=(\varsigma^m,\varsigma^f)\in G$, set $\omega'=\varsigma\cdot\omega$ and $(\tau^{m'},\tau^{f'})=(\varsigma^m)^{-1}\circ\tau^m \circ\varsigma^m, (\varsigma^f)^{-1}\circ\tau^f\circ\varsigma^f$. 
Using $\sigma_{\omega'}=(\varsigma^f)^{-1}\circ\sigma_\omega\circ\varsigma^m$, we get for all $i\in\mathbb{N}_2$ 
$$
\chi(\omega',\tau^{m'},\tau^{f'})_i=\chi(\omega,\tau^m,\tau^f)_{\varsigma^m(i)} \quad\text{a.s.}
$$
Since $\omega'\stackrel{d}{=}\omega$, $(\tau^{m'},\tau^{f'})\stackrel{d}{=}(\tau^m,\tau^f)$, and $(\tau^m,\tau^f)\perp \omega$, it follows that $\{\chi_i\}_{i\in\mathbb{N}_2}\stackrel{d}{=}\{\chi_{\varsigma(i)}\}_{i\in\mathbb{N}_2}$, i.e. we have $G$-exchangeability. We denote the corresponding permutation-invariant sigma-algebra as $ \mathcal{I}^\star $. 
\end{proof}

\subsection{Proof of Theorem \ref{pro:definetti}}\label{app:proofOfdefinetti}
We work on the probability space $ (\Omega\times\mathcal{T},\mathcal{I}\otimes\mathcal{B}(\mathcal{T}),\mu\otimes\rho). $
Since $\Omega$ is Polish by Assumption \ref{ass:polish} and $\mathcal{T}$ is Polish so is $\Omega\times\mathcal{T}$ and $ \mathcal{B}(\Omega\times\mathcal{T}) = \mathcal{I}\otimes\mathcal{B}(\mathcal{T}) $.
Hence, by \citet[Theorem 5.3]{Kallenberg1997}, there exists a regular conditional distribution of the primitive pair $(\omega,\tau)$ given $\mathcal{I}^\star$ defined in Lemma \ref{lem:proj_exchangeable}. Denote it by
$$
  \bar Q:(\Omega\times\mathcal{T})\times(\mathcal{I}\otimes\mathcal{B}(\mathcal{T}))\to[0,1].
$$
Thus, by \citet[Theorem 5.3]{Kallenberg1997} for each $D\in\mathcal{I}\otimes\mathcal{B}(\mathcal{T})$, the map
 $ (\omega,\tau)\mapsto\bar Q((\omega,\tau),D) $ is $\mathcal{I}^\star$-measurable, and for each $(\omega,\tau)$, the map $ D\mapsto\bar Q((\omega,\tau),D) $ is a probability measure. 
 By the same reference, for all $B\in\mathcal{I}^\star$ and all $D\in\mathcal{I}\otimes\mathcal{B}(\mathcal{T})$,
$$
  \int_B \bar Q((\omega,\tau),D)\,(\mu\otimes\rho)(d\omega,d\tau)
  = (\mu\otimes\rho)(B\cap D).
$$

The conditional distribution of a given configuration is obtained by pushing this conditional distribution through $\chi_i$. Since we can select $ \omega^0 $ arbitrarily, the space containing $ \chi_i $ is isomorphic to $ \mathbb{X}^4 $. Thus, for all $B\in\mathcal{B}(\mathbb{X}^4)$, define
$$
  \bar\nu((\omega,\tau),B)
  \equiv (\mu\otimes\rho)(\chi_i\in B\mid\mathcal{I}^\star)(\omega,\tau)
  = \int \mathbf{1}\{\chi_i(\omega',\tau')\in B\}\,\bar Q((\omega,\tau),d\omega',d\tau').
$$
This is the directing conditional distribution defined in the statement of the theorem.

By Lemma \ref{lem:proj_exchangeable}, the sequence $(\chi_i)_{i\in\mathbb{N}_2}$ is exchangeable under within-pool relabellings. Since the configuration space isomorphic to $\mathbb{X}^4$ is Polish, the Hewitt-Savage extension of de Finetti's theorem applies. Therefore, conditional on $\mathcal{I}^\star$, the sequence is i.i.d.\ with directing measure $\bar\nu$. That is, for all measurable $B_1,\ldots,B_n\subseteq\mathbb{X}^4$,
 $$ (\mu\otimes\rho)(\chi_1\in B_1,\ldots,\chi_n\in B_n\mid\mathcal{I}^\star)(\omega,\tau)
  = \prod_{k=1}^n \bar\nu((\omega,\tau),B_k).$$
This proves the conditional i.i.d.\ part of the theorem.

The conditional demand distribution is the image of this directing measure under the demand map. Since $ \Psi^\kappa = \operatorname{argmax}\circ\,\Phi\circ\operatorname{proj}^\kappa $ for  $\kappa\in\{c,m,f\} $,
each coordinate of $\Psi$ first extracts the respective household's preferences, then forms the household utility representation, and finally maps it into utility-maximising demand. Hence, for every event $A$ in the joint space of demand triples,
$$
  \bar\pi((\omega,\tau),A) = \int \mathbf{1}\{\Psi(\chi)\in A\}\,\bar\nu((\omega,\tau),d\chi).
  $$
This is the conditional random utility and matching representation.

Next we show that Hypothesis \ref{hyp:poollaw} implies existence of a mixture of configurations that rationalises demand distributions. 
Let $\bar\mu^c$ denote the conditional distribution of the factual matched-couple primitives $ (\omega_i^m,\omega_{\sigma(i)}^f) $ for $ i\in\mathbb{N}_2$, with marginals $\bar\mu_2^m$ and $\bar\mu_2^f$. By Hypothesis \ref{hyp:poollaw} we have $ \bar\mu_2^m=\bar\mu_1^m $ and $ \bar\mu_2^f=\bar\mu_1^f $.
Thus the marginal distribution of partnered men and women are also the marginal distributions of single men and women.
Thus, we can construct a probability measure $\nu^\star$ on $\mathbb{X}^4$ as follows. Draw $ (\omega^m,\omega^f)\sim\bar\mu^c $ and set $ (\omega^{m\prime},\omega^{f\prime})=(\omega^m,\omega^f) $ so that $\nu^\star$ is the distribution of  $ (\omega^m,\omega^f,\omega^m,\omega^f) $, which, by construction, is supported on $W_0$. 
Its couple projection has the factual matched-couple distribution $\bar\mu^c$, and its counterfactual single projections have the correct single-side marginal distributions by the equalities above. Therefore, pushing $\nu^\star$ forward through $\Psi^c,\Psi^m,\Psi^f$ yields $\bar\pi^c,\bar\pi^m,\bar\pi^f$, respectively.

Hence the marginal conditional demand distributions admit a preference-stable representation $ \nu^\star $ as a consequence of \ref{hyp:poollaw}.\footnote{Note that this is not required to coincide with the sampled directing law $\bar\nu$.}\hfill\qed

\subsection{Corollary \ref{cor:blockmarschak}}\label{app:blockmarschak}

\begin{cor}\label{cor:blockmarschak}
Let $\bar\pi^c,\bar\pi^m,\bar\pi^f$ be the marginal conditional demand distributions of the main theorem, let $\Psi=(\Psi^c,\Psi^m,\Psi^f)$ be the demand map, and let $W_0$ denote the preference-stable subset of the configuration space. Set
$$
  S \equiv \Psi(W_0),
  \qquad
  \Gamma(\bar\pi)
  \equiv \left\{\gamma:\gamma\circ(\operatorname{proj}^\kappa)^{-1}=\bar\pi^\kappa,\ \kappa\in\{c,m,f\}\right\}.
$$
Then the following are equivalent:
\begin{enumerate}
  \item[(i)] \emph{Preference-stable representation.} There exists a probability measure $\nu^\star$ on $\mathbb{X}^4$ with $\nu^\star(W_0)=1$ and $\nu^\star\circ(\Psi^\kappa)^{-1}=\bar\pi^\kappa$ for $\kappa\in\{c,m,f\}$.
  \item[(ii)] \emph{Multi-marginal feasibility.} There exists $\gamma\in\Gamma(\bar\pi)$ with $\gamma(S)=1$ such that
  $$ 
    \inf_{\gamma\in\Gamma(\bar\pi)}\int\mathbf{1}_{S^c}(x)\,d\gamma(x)=0.
 $$ 
  where $ S^c $ is the complement of $ S $, i.e. the set of demands incompatible with pref.-stability.
  \item[(iii)] \emph{Block-Marschak inequalities.} For every bounded measurable triple $(\varphi^c,\varphi^m,\varphi^f)$ with
 $$ 
    \varphi^c(x^c)+\varphi^m(x^m)+\varphi^f(x^f)\leq 0
    \qquad\text{for all }(x^c,x^m,x^f)\in S,
 $$ 
  we have
 $$ 
    \int\varphi^c\,d\bar\pi^c+\int\varphi^m\,d\bar\pi^m+\int\varphi^f\,d\bar\pi^f\leq 0.
 $$ 
\end{enumerate}
\end{cor}

\begin{proof}
We prove (i)$\Rightarrow$(ii), (ii)$\Rightarrow$(i), and (ii)$\Leftrightarrow$(iii).

\noindent\emph{(i)$\Rightarrow$(ii).} Suppose $\nu^\star$ satisfies (i). Define
$$
  \gamma^\star \equiv \nu^\star\circ\Psi^{-1}.
$$
Since $\Psi=(\Psi^c,\Psi^m,\Psi^f)$, the $\kappa$-marginal of $\gamma^\star$ is $\nu^\star\circ(\Psi^\kappa)^{-1}=\bar\pi^\kappa$, so $\gamma^\star\in\Gamma(\bar\pi)$. Moreover, $\gamma^\star(S)=\nu^\star(\Psi^{-1}(S))\geq\nu^\star(W_0)=1$. Hence
$$
  \inf_{\gamma\in\Gamma(\bar\pi)}\int\mathbf{1}_{S^c}(x)\,d\gamma(x)=0.
$$

\noindent\emph{(ii)$\Rightarrow$(i).} Suppose the multi-marginal problem has a feasible solution $\gamma\in\Gamma(\bar\pi)$ with $\gamma(S)=1$. Since $\mathbb{X}^4 $ is Polish, by \citet[Theorem 12.13]{kechris1995} there exists a measurable selection
  $ s:S\to W_0 $ such that $ \Psi(s(x))= x $  for $x\in S$. Define  $\nu \equiv \gamma\circ s^{-1} $, then $\nu(W_0)=\gamma(s^{-1}(W_0))\geq\gamma(S)=1$. Since $\Psi^\kappa\circ s=\operatorname{proj}^\kappa$, 
$$
  \nu\circ(\Psi^\kappa)^{-1}
  =\gamma\circ s^{-1}\circ(\Psi^\kappa)^{-1}
  =\gamma\circ(\operatorname{proj}^\kappa)^{-1}
  =\bar\pi^\kappa,
$$
for $\kappa\in\{c,m,f\}$ we have that $\nu$ is a preference-stable representation, establishing (i).

\smallskip
\noindent\emph{(ii)$\Leftrightarrow$(iii).} By Kantorovich duality for the multi-marginal transport problem with cost $\mathbf{1}_{S^c}$ (\citealp[Theorem 5.10]{villani2009}), feasibility is equivalent to the dual inequalities of (iii). These are the \citet{Block1960} inequalities for the preference-stable random utility representation. In the finite-support implementation, they reduce to the corresponding finite Block-Marschak polynomial restrictions.
\end{proof}

\subsection{Proof of Proposition \ref{lem:equivalences}}\label{app:proofOfequivalences}

\emph{\ref{equ:qualitative} $\Rightarrow$ \ref{equ:system}:} 
Suppose the marginals $\bar\pi^\kappa$ admit a rationalisation by some $\nu^\Delta$ supported on $\Theta_0$. 
Stack the marginals into the vector $\bar\pi = (\bar\pi^c, \bar\pi^m, \bar\pi^f)$ and define $\nu^\Delta $ per component $\nu^\Delta_l = \nu^\Delta(\theta_l)$ for $\theta_l \in \Theta_0$. 
By construction of $A$, the column $A_{\cdot, l}$ is the indicator vector of the revealed preference types projected from $\theta_l$ via $\text{proj}^\kappa$ for $\kappa \in \{c, m, f\}$. 
Hence $(A\nu^\Delta)_{\kappa, j} = \sum_l \mathbf{1}\{\text{proj}^\kappa(\theta_l) = \xi^\kappa_j\}\,\nu^\Delta(\theta_l) = \bar\pi^\kappa_j$ by the first 
equality of \eqref{eq:rumdiscrete}, so $A\nu^\Delta = \bar\pi$.

\emph{\ref{equ:system}$\Rightarrow$ \ref{equ:qualitative}:} Conversely, given $\nu^\Delta $ with $A\nu^\Delta = \bar\pi$, define $\nu^\Delta$ on $\Theta_0$ by $\nu^\Delta(\theta_l) = \nu^\Delta_l$. The same computation shows that the marginals of $\nu^\Delta$ under $\text{proj}^\kappa$ coincide with $\bar\pi^\kappa$, so $\nu^\Delta$ rationalises the observed marginals.

The equivalence between \ref{equ:system} and \ref{equ:quadratic} is shown in \citet{McFadden1991,McFadden2005}. 
Statement \ref{equ:quadratic} referenced therein, differs from \ref{equ:quadratic} in that it additionally requires $ \iota^\top \nu^\Delta = 1 $.
We now show that this is implied.
It is easy to see that by construction of $ A $ for any solution of the quadratic problem we have $ \gamma = \pi $ and since $ 3 = \iota^\top \pi = \iota^\top A \nu^\Delta = 3 \iota^\top \nu^\Delta $ by construction, we get $ \iota^\top \nu^\Delta = 1 $. 
Thus constraint $ \nu^\Delta \geq 0 $ in is sufficient for $ \gamma $ to be on the probability simplex.

It will be useful to write this problem with a \emph{tightened} cone constraint indexed by $ \underline{\nu} $.
Let $ L $ be a lower diagonal matrix from the Cholesky decomposition $ \Omega = L L^\top $. Then we can rewrite the quadratic form \ref{equ:quadratic} as
\begin{equation}
\min\limits_{\gamma \in \left\{ A\nu^\Delta | \nu^\Delta \geq \underline{\nu} \right\}} (\pi - \gamma)^\top L L^\top (\pi - \gamma) \text{.}
\end{equation}
Using $ \gamma = A \nu^\Delta $ and introducing a slack variable $ s \geq 0 $ such that we can write $ \nu^\Delta = \underline{\nu} + s $ we obtain
\begin{equation}
\min\limits_{\nu^\Delta = \underline{\nu} + s, s\geq 0 } (\pi - A(\underline{\nu}+s))^\top L L^\top (\pi - A(\underline{\nu}+s)) \text{.}
\end{equation}
This does not depend on $ \nu^\Delta $ but only on $ s $ and we can write it in the quadratic form
\begin{equation}
\min\limits_{s\geq 0} \left\{ \frac{1}{2}s^\top A^\top \Omega A s - s^\top A^\top \Omega (\pi - A\underline{\nu}) \right\} \text{.}
\end{equation}
Letting $ H = A^\top\Omega A $ and $ f(\pi, \underline{\nu}) = -A^\top\Omega(\pi-A\underline{\nu}) $ we get a canonical form of a non-negative least squares problem, with gradient for iteration $ \tau \geq 0 $ defined as $ \mu_{\tau} = H s_{\tau} + f(\pi, \underline{\nu}) $.
\citet{Johansson2006} show that component-wise projection $ s_{\tau+1,j} = \max(0, s_{\tau,j} - \mu_{\tau,j} d_j) $ where $ d = \text{diag}(H \iota)^{-1} $ and $ j = 1, \ldots, |\Theta_0| $ referring to the $j^{\text{th}}$ component of $ s $ will find the solution of the problem. \hfill\qed

\begin{rem}\label{rmk:conservatism}
Proposition \ref{lem:equivalences} characterises rationalisability at the discrete level. Its relationship to the structural hypothesis \ref{hyp:poollaw} is characterised as follows:
$$ W_0 \subseteq \theta^{-1}(\Theta_0^{\text{Slutsky}}) \subseteq \theta^{-1}(\Theta_0) $$
where $W_0$ is the set of configurations with $\omega^m_i = \omega^m_{i'}$ and $\omega^f_j = \omega^f_{j'}$ (true preference stability), $\Theta_0^{\text{Slutsky}}$ collects discrete configuration types whose continuous representatives satisfy the Slutsky restriction of Lemma \ref{lem:prefstable}, and $\Theta_0$ collects those satisfying the discrete CARP/GARP characterisation of Definition \ref{thm:necessary}. Both inclusions reflect that each subsequent characterisation is necessary but not sufficient for the previous one.

Consequently, existence of $\nu^\star$ on $W_0$ implies existence of a rationalising $\nu^\Delta$ on $\Theta_0$, but not the converse. The test based on $\Theta_0$ therefore has correct size under \ref{hyp:poollaw}: rejection rules out rationalisability on the most permissive set, and hence rules out \ref{hyp:poollaw}. Power is conservative and the test may fail to detect violations of \ref{hyp:poollaw} that lie in the gap between $W_0$ and $\theta^{-1}(\Theta_0)$.
\end{rem}

\section{Auxiliary Proofs}\label{sec:auxiliaryproofs}
\subsection{Permutation equivariance of assignment problem}\label{lem:matching}

Let $ \Phi_{ij}\equiv\Phi(z_i,z_j) $ and the set of permutation matrices $ M \equiv \{ \mu\in\{0,1\}^{\mathbb N\times\mathbb N}:\ \sum_j \mu_{ij}=1,\ \sum_i \mu_{ij}=1 \} $.
Define the assignment problem as
\begin{equation}
M(\Phi)\ \in\ \arg\max_{\mu\in\mathcal M}\ \sum_{i,j} \mu_{ij}\Phi_{ij}.
\end{equation}
Now fix $(\varsigma^m,\varsigma^f)\in G\times G$ and define the relabelled surplus as $ \Phi'_{ij}\equiv\Phi_{\varsigma^m(i),\,\varsigma^f(j)} $, and the assignment matrix $ \mu'_{ij}\equiv\mu_{\varsigma^m(i),\,\varsigma^f(j)} $.
$ M $ is permutation-equivariant if $ \mu\in M(\Phi) $ implies $ \mu' \in M(\Phi') $.
First, for feasibility, if $\mu\in\mathcal M$, then for every $i$,
\begin{equation}
\sum_j \mu'_{ij} =\sum_j \mu_{\varsigma^m(i),\,\varsigma^f(j)} =\sum_{j'} \mu_{\varsigma^m(i),\,j'}=1,
\end{equation}
since $ \varsigma^f \in G \subset S_{\infty} $ is a bijection. The same holds for $ \varsigma^m \in G $, and hence $\mu'\in\mathcal M$.
Second, for the objective, we note that $ (\varsigma^m, \varsigma^f) $ is a bijection on $ \mathbb{N}^2 $, with inverse $ ((\varsigma^m)^{-1}, (\varsigma^f)^{-1}) $. Hence, by reindexing the sum, we show that the relabelled assignment problem has the same objective value as the original one:
\begin{equation}\label{eq:obj_equiv}
\sum_{i,j} \mu'_{ij} \Phi'_{ij} =\sum_{i,j} \mu_{\varsigma^m(i),\,\varsigma^f(j)}\, \Phi_{\varsigma^m(i),\,\varsigma^f(j)} =\sum_{i',j'} \mu_{i'j'}\Phi_{i'j'}.
\end{equation}
Finally we must show that any other $ \nu' \in \mathcal{M} $  is inferior to $ \mu' $ in the relabelled problem. 
Using \eqref{eq:obj_equiv} for the first and last equality, we have
\begin{equation}
\sum_{i,j} \mu'_{ij} \Phi'_{ij} = \sum_{i',j'} \mu_{i'j'} \Phi_{i'j'} \geq \sum_{i',j'} \nu_{i'j'} \Phi_{i'j'} = \sum_{i,j} \nu'_{ij} \Phi'_{ij}
\end{equation}
where the middle inequality follows from $\mu \in M(\Phi)$.
We move from $ \nu' $ back to $ \nu  $ using the same inverse as defined on $ \mu $. \hfill\qed

\subsection{Examples}\label{app:examples}
\begin{exa} \label{exa:transposition}
  For our example in Figure \ref{fig:partition} let Zeus \& Hera be the couple $(2, 2)$, Athena the single $ (0, 1) $, and Apollo the single $ (1, 0) $. The transposition $ \tau^f = (1\; 2) $ generates the hypothetical couple Zeus \& Athena $ (2, 1) $ and the counterfactual single woman Hera $ (0, 2) $. Similarly, the transposition $ \tau^m = (1\; 2) $ creates the counterfactual couple Apollo \& Hera $ (1, 2) $ and the counterfactual single Zeus $ (2, 0) $. Applying them both through $ \tau' = (1\; 2) \circ \sigma \circ (1\; 2) $ leads to the counterfactual couple Apollo \& Athena $ (1, 1) $ and the counterfactual singles Zeus $(2, 0)$ and Hera $(0, 2)$.
Indeed, we have for Zeus: $\tau^f(\sigma(\tau^m(2))) = \tau^f(\sigma(2)) = \tau^f(2) = 1 $ and for Apollo: $\tau^f(\sigma(\tau^m(1))) = \tau^f(\sigma(1)) = \tau^f(0) = 0 $. The matching graph is unchanged but the original couples edge is now relabelled as Apollo \& Athena, the single man edge as Zeus, and the single woman edge as Hera.
The triple $ (i, \tau^m, \tau^f) = (2, (1\; 2), (1\; 2)) $ represents the corresponding configuration.
\end{exa}

\begin{exa}\label{exa:nondegenerate}
To see how \eqref{eq:rumdiscrete} relates to the existence of resource shares $ \rho $, we look at a specific numerical example of a stylised reduced-form dictatorship collective model, in which the household type is determined by the preferences of person with the higher bargaining power.
Refer to Figure \ref{fig:budgets}. An individual can be of type $ \xi \in \{ \xi_{st}, \xi_{s't}, \xi_{st'} \} $ representing tuples of the line-segments of the respective delegate consumption bundles $ (x_s, x_t), (x_s', x_t), (x_s, x_t') $.
Denote the three distributions $ \pi^m, \pi^f, \pi^c $ supported on this choice space.
Now, assume we knew how the latent configurations were allocated.
As stated above, for this we do not need the whole assignment matrix but only the implied matches (couplings) of latent male and female revealed preference types.
Let the matching matrix induced by the mixture distribution of primitive types $ \omega^m $ and $ \omega^f $, be denoted by $ \mu_{mf} = P(m = \xi_m, f = \xi_f) $ and take $$ \mu = \left( \begin{matrix} \mu_{st,st} & \mu_{st, s't} & \mu_{st, st'} \\ \mu_{s't,st} & \mu_{s't, s't} & \mu_{s't, st'} \\ \mu_{st',st} & \mu_{st', s't} & \mu_{st', st'}\end{matrix} \right)  = \left( \begin{matrix} 0 & 0.2 & 0.4 \\ 0.2 & 0.1 & 0 \\ 0 & 0.1 & 0 \end{matrix} \right), $$
which implies the row margins $ \pi^m = (\pi^m_{st}, \pi^m_{s't}, \pi^m_{st'}) = (0.6, 0.3, 0.1) $ and column margins $ \pi^f = (\pi^f_{st}, \pi^f_{s't}, \pi^f_{st'}) = (0.2, 0.4, 0.4) $.
Now, assume that we observe couple's choices $ \pi^c = (\pi^c_{st}, \pi^c_{s't}, \pi^c_{st'}) = (0.2, 0.7, 0.1) $. 
Using, the masses of our couplings matrix $ \mu $ we can ask the question whether there exists $ \rho $'s which induces choices consistent with this observed distribution.
Take, for example $ (\xi_{st}, \xi_{s't}) $. 
Using the extreme cases where every household consisting of individuals of the respective types had bargaining power approaching $ 0 $ or $ 1 $, this type of couple could induce $ \xi_{st} $ with probability masses $ \pi_{st} \in (0, 0.2) $ and $ \xi_{s't} $ with $ \pi_{s't} \in (0, 0.2) $.
Proceeding in the same fashion for all other types, we can construct bounds $ \pi^c_{st} \in (0.0, 0.8) $, $ \pi^c_{s't} \in (0.1, 0.6) $, and $ \pi^c_{st'} = (0, 0.5) $.\footnote{There are also joint restrictions which can be obtained by solving a linear program.}
Clearly, the observed $ \pi^c $ is outside these bounds.
Thus, there is no $ \rho $ which rationalises the observed type distribution, which provides evidence against stable preferences. 
In our test, we do not assume knowledge of the coupling matrix $ \mu $, and rationalisability becomes an existence statement over over all possible couplings consistent with the observed marginals.
\end{exa}
 \section{Simulations}\label{sec:mc}
In this section, we investigate the properties of our proposed test in a simulation setting. In particular, we are interested in
how much power it has to detect a violation of the stable preference assumption and whether or not it has a correct proportion of false positives. 
Since specifying a parametric continuous demand system requires at least five goods to impose the SNR(S-1) condition on the Slutsky matrix and distinguish the collective model
from the unitary model, we will not sample continuous demands as functions of prices and individual budget constraints, but rather draw our sample directly from the discrete choice space.\footnote{A revealed preference based setting allows us to test the restrictions of the model with only three goods \citep{Cherchye2007}, whereas \citet{Browning1998} need five goods.} 
This should be interpreted as a continuous uniform distribution of choices on different budget planes, where the relative prices are such that the partitions of the budget planes
are of equal size. 
Recall that we test this against the set of households which are consistent with the necessary conditions of the collective axioms based on aggregate consumption but not 
consistent when single data and the stable preference assumption is added. This set is denoted by $ \Theta_1 $ and we have $ \Theta_{\text{collective}} = \Theta_0 \cup \Theta_1 $. 
If we reject the null hypothesis that both the collective axiom and the stable preference assumption holds, by excluding all irrational matches $ \Theta \setminus \Theta_{\text{collective}} $,
we must conclude that the stable preference assumption does not hold.
To control the proportion of households for whom this is the case (our data generating process) we introduce the parameter $ p $ which specifies the probability that a particular choice is both collectively rational and satisfies the stable preference assumption $ p \equiv P(\theta \in \Theta_0) $.\footnote{This rationality parameter is similar as for example $ \lambda $ in \citet{Hoderlein2011} which specifies the population's deviation from Slutsky symmetry.} 
By only considering collectively rational choices in our simulations we thus have $ 1 - p = P(\theta \not\in \Theta_0) = P(\theta \in \Theta_1) $ by construction.
Simulation lets us trivially treat $ P = \mu \otimes \rho $ as a joint measure over the type space, rather than a directing measure from a de Finetti representation of configurations.

Our simulation setting is as follows. We consider $ S = 100 $ samples of size $ \underline{n} \in \left\{ 500, 1000, 2000 \right\} $ where $ \underline{n} = n_f = n_m = n_c $
such that $ n = 3 \underline{n} $ in a minimal setting with $ T = 3 $ periods which we construct by drawing $ \lfloor \underline{n}p \rfloor $ indices from the space of 
collectively rational matches $\mathfrak{X}^0$ for which the stable preference assumption holds and $ \lceil \underline{n}(1-p) \rceil $ indices from the space 
of collectively rational types $ \mathfrak{X}^{1} $ which does not satisfy the assumption.  
Based on a sample of matches, we then calculate the choice probabilities $ \widehat{\pi} $ accordingly. 
For estimation, we only use the marginal distribution of choices of each sample of household compositions and
draw $ B = 100 $ samples from the respective empirical distributions (i.e. with replacement) to calculate $
\pi^b_{\tau_n} $ and estimate the empirical distribution of the test statistic $ \mathcal{J}^{\tau_n}_{n,b} $. 
These simulations are repeated for $ p \in \left\{ 0.75, 0.85, 0.9, 0.95, 0.975, 0.99, 1.00 \right\} $. 

\begin{figure}[t]
	\caption{Power function for $ n = 1,500 $ (l.h.s) and $ n = 3,000 $ (r.h.s.)}
	{\centering
	\begin{minipage}{\textwidth}
	\includegraphics[width=0.48\textwidth]{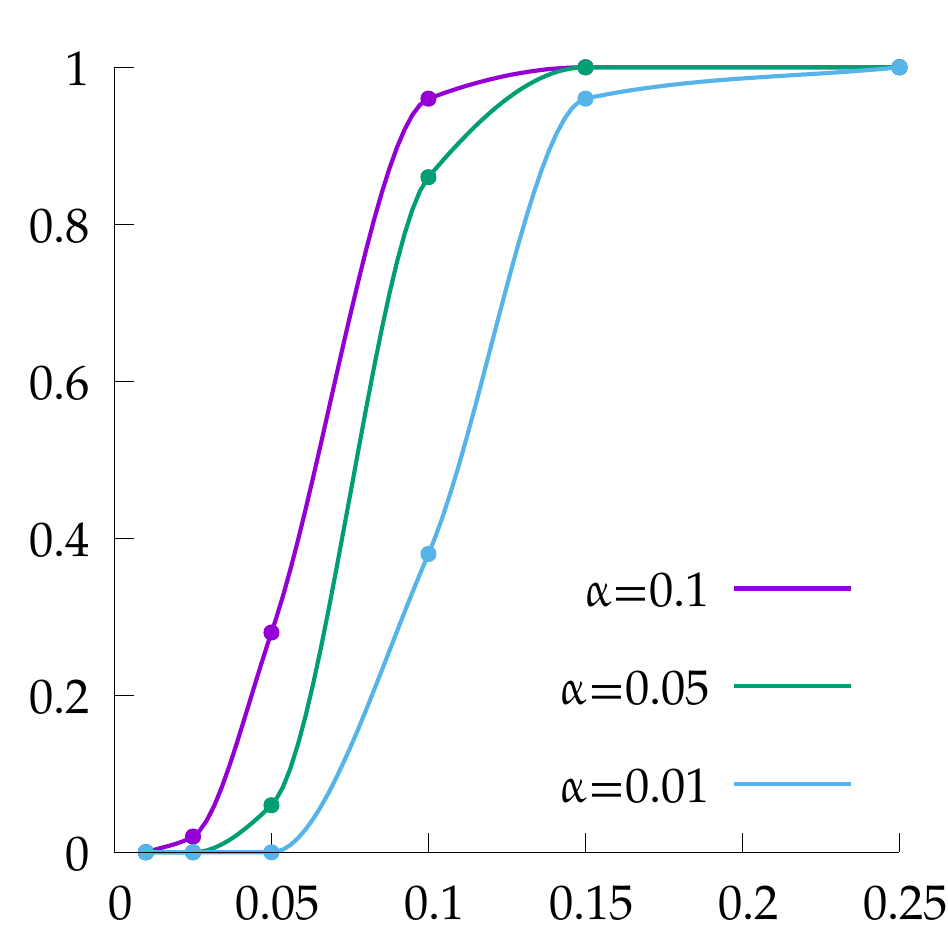} \hfill \includegraphics[width=0.48\textwidth]{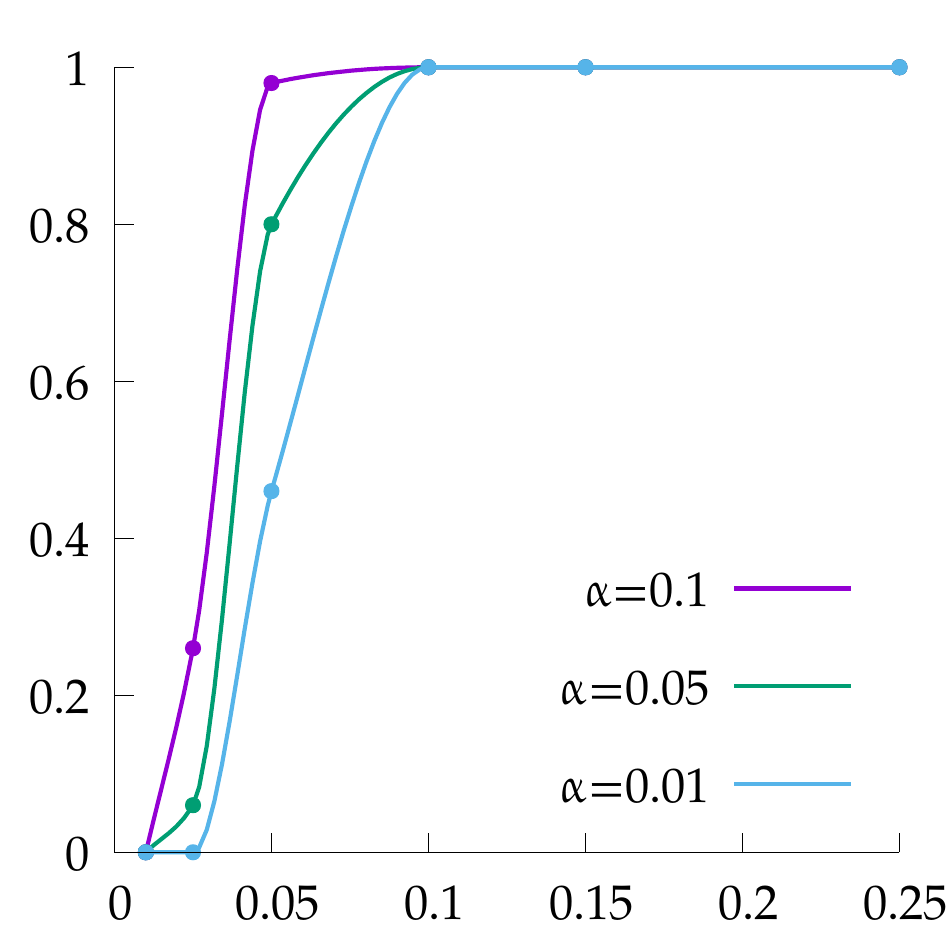}
	\end{minipage}}
	\label{fig:power}
\end{figure}

Figure \ref{fig:power} shows the power of our test against the non-stable preference alternative as a function of $ p $, with
sample-size $ \underline{n} = 500 $ for the left-hand side graph, and $ \underline{n} = 1000 $ for the right-hand side graph, respectively. 
We use monotone cubic splines to interpolate between the actual simulation results, which are marked as solid dots.
To be more precise, the respective functions refer to sample rejection frequencies using the rejection 
rule $ J \mapsto \indicator\left\{J > \widehat{F_{\mathcal{J}_{n}}^{-1}}(1-\alpha) \right\} $ for $\alpha \in \left\{ 0.01, 0.05,
0.10 \right\} $.
In addition to this, we also observe that as $ \underline{n} $ increases the power of our test improves and is able to correctly reject the hypothesis of a 
collectively rational population already at small proportions $ p $.

The intercepts of these functions should be interpreted as the proportion of false positives (type I errors) since they correspond to the case 
where everyone is rational. One might expect that for a correctly sized test the empirical rejection frequencies should tend to $ \alpha $. 
However, given our partial identification procedure we have a composite null hypothesis, \textit{i.e.} the probability of a type I error should be at most $ \alpha $
as defined in equation \eqref{eq:test}. 
To see this note that every vector of "true" choice frequencies denoted by $ \pi_0 $ lying in the interior of the cone will have projection residuals
of length zero. Bootstrapping out of $ \widehat{\pi} $ which tends to $ \pi_0 $ using the usual regularity properties could then lead to a confidence interval
which is always entirely in the interior of the cone and we would never wrongly reject the null hypothesis. This also implies that in such a case 
our bootstrap distribution is degenerate and has mass one at point zero. 

In our Monte Carlo setting and the case where $ p = 1.0 $, we randomly select types from the type-space $ \mathfrak{X}^0 $, satisfying collective rationality. 
Thus the "true" parameter vector $ \nu_0 $ is assumed to have a uniform distribution over the probability simplex and the worst-case, namely to get a $ \nu $ such that $ \pi_0 = A \nu $ is on the boundary of the cone with respect to any of its dimensions, occurs with measure zero. 

Thus, in order to evaluate whether the size of our test is correct under the test's minimax strategy, we have to construct a worst case.
For this, note that the test is constructed in a way that considers hypothetical types by taking combinations of possible household choice behaviour per price regime
over a range of price regimes. To fix notation, we will call two collectively rational matches \emph{similar} if there is at least one element
in the product space spanned by these two matches which is an element of the space of collectively rational matches that do not satisfy the stable preference hypothesis.
We will then construct worst cases by specifying a distribution over $ n_0 $ such similar matches. To make sure that our $ \pi_0 $ is on the boundary
of the cone in all dimensions, \textit{i.e.} on the cusp, we shift the cone by manually controlling the tightening parameter  $ \tau_n $ according to this distribution. 
Figure \ref{fig:size} shows simulation results for two such worst case scenarios with $ 5 $ similar matches and $ 2 $ similar matches, respectively. 

\begin{figure}[t]
	\caption{Type I error for $ n_0 = 5 $ (l.h.s) and $ n_0 = 2 $ (r.h.s.) worst-case matches}
	{\centering
	\begin{minipage}{16.06cm}
	\includegraphics[width=0.48\textwidth]{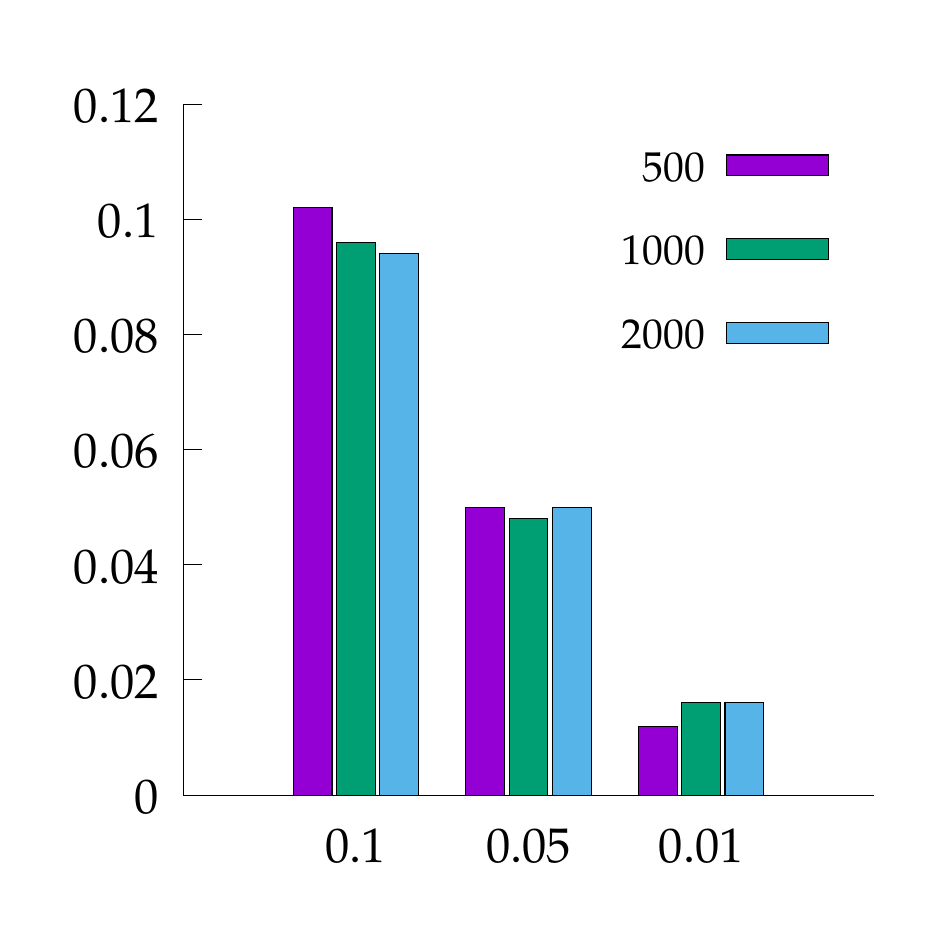} \hfill \includegraphics[width=0.48\textwidth]{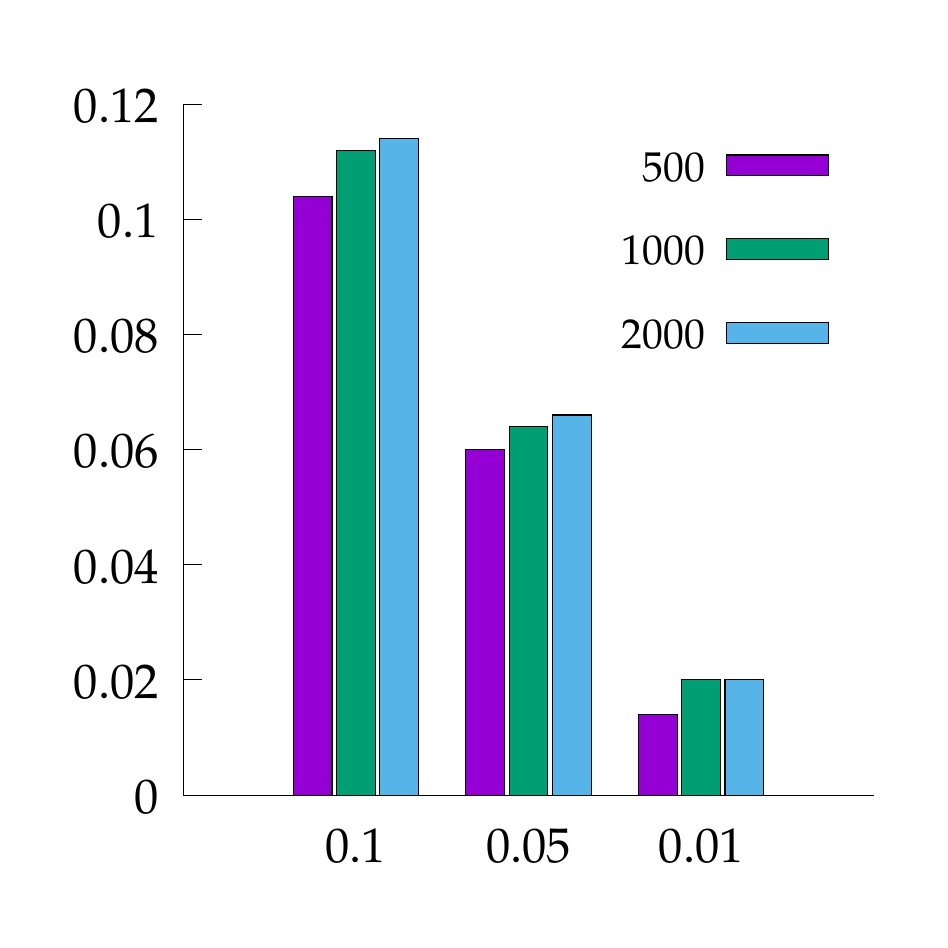}
	\end{minipage}}
	\label{fig:size}
\end{figure}

The size results do not seem to deteriorate much with the number of worst case matches included in the sample.
Since the properties of the test are based on an asymptotic argument, we should see the empirical frequency of false positives tending to the respective $ \alpha $ which define the rejection rules and are plotted 
on the $ x $-axis.  The results are what one would expect, with all sample sizes being reasonably accurate. 
Since in a well-behaved test, false-positives are by definition rather rare events, in order to minimize simulation uncertainty, we increased the 
number of Monte Carlo repetitions to $ S = 500 $.
which greatly increased computational complexity due to the high dimensionality of the testing problem. 

 \section{Further Results}\label{sec:robustness}
{\singlespacing \begin{table}[ht!]
\caption{Results Conditional on Private Expenditure Class, Private Goods, Exogenous Prices}
\centering
\singlespacing
\small
\renewcommand{\arraystretch}{0.9}
\resizebox{\textwidth}{!}{
\begin{tabular}{ccccccc}
\multicolumn{7}{c}{\textsc{Longitudinal Internet studies for the Social Sciences (LISS)}} \\
\toprule
Years  &  Private Expenditure  &  $ N_{\text{couples}}^{\text{total}} $  &  $ N_{\text{couples}}^{\text{rational}} $  &  $ N_{\text{singles}}^{\text{total}} $  &  $ N_{\text{singles}}^{\text{rational}} $  &  p-value \\
\midrule
2009   2010   2012  &  high  &  90  &  90  &  127  &  115  &  0.000 \\
2009   2010   2012  &  mid  &  100  &  99  &  84  &  69  &  0.028 \\
2009   2010   2012  &  low  &  110  &  107  &  89  &  55  &  0.386 \\
[-7pt]\\

\bottomrule
\end{tabular}}
\label{tab:private.nomip.rs}
\vspace{-\baselineskip}
\floatnote{Number of total couples, rational couples according to aggregate CARP, total singles and rational singles according to GARP, for different combinations of periods and private expenditure class (low, mid, high).}
\vspace{-\baselineskip}
\end{table}
 }

{\singlespacing \begin{table}[ht!]
\caption{Results for Public Goods with Exogenous Prices}
\centering
\singlespacing
\small
\renewcommand{\arraystretch}{0.9}
\resizebox{\textwidth}{!}{
\begin{tabular}{cccccc}
\multicolumn{6}{c}{\textsc{Longitudinal Internet studies for the Social Sciences (LISS)}} \\
\toprule
Years  &  $ N_{\text{couples}}^{\text{total}} $  &  $ N_{\text{couples}}^{\text{rational}} $  &  $ N_{\text{singles}}^{\text{total}} $  &  $ N_{\text{singles}}^{\text{rational}} $  &  p-value \\
\midrule
2009   2010   2012  &  1259  &  1217  &  390  &  320  &  0.074 \\
[-7pt]\\

\multicolumn{6}{c}{\textsc{Russian Longitudinal Monitoring Survey (RLMS)}} \\
\toprule
Years  &  $ N_{\text{couples}}^{\text{total}} $  &  $ N_{\text{couples}}^{\text{rational}} $  &  $ N_{\text{singles}}^{\text{total}} $  &  $ N_{\text{singles}}^{\text{rational}} $  &  p-value \\
\midrule
2012   2013   2014  &  312  &  300  &  291  &  276  &  0.008 \\
2011   2013   2014  &  303  &  298  &  270  &  253  &  0.012 \\
2011   2012   2014  &  300  &  297  &  270  &  253  &  0.026 \\
2011   2012   2013  &  317  &  307  &  298  &  275  &  0.024 \\
2010   2013   2014  &  246  &  238  &  207  &  192  &  0.000 \\
2010   2012   2014  &  246  &  239  &  212  &  194  &  0.022 \\
2010   2012   2013  &  258  &  252  &  234  &  212  &  0.052 \\
2010   2011   2013  &  259  &  248  &  231  &  207  &  0.060 \\
2010   2011   2012  &  288  &  281  &  259  &  237  &  0.054 \\
2005   2011   2012  &  254  &  239  &  210  &  182  &  0.012 \\
[-7pt]\\

\multicolumn{6}{c}{\textsc{Spanish Continuous Family Expenditure Survey (ECPF)}} \\
\toprule
Years  &  $ N_{\text{couples}}^{\text{total}} $  &  $ N_{\text{couples}}^{\text{rational}} $  &  $ N_{\text{singles}}^{\text{total}} $  &  $ N_{\text{singles}}^{\text{rational}} $  &  p-value \\
\midrule
1994.3   1994.1   1994.2  &  106  &  105  &  5  &  3  &  0.206 \\
1993.4   1994.1   1994.2  &  108  &  105  &  8  &  4  &  0.052 \\
1992.2   1992.3   1992.1  &  93  &  90  &  15  &  14  &  0.030 \\
1990.4   1991.1   1991.2  &  95  &  94  &  14  &  12  &  0.012 \\
1989.3   1989.1   1989.2  &  107  &  104  &  4  &  3  &  0.026 \\
1988.2   1988.3   1988.4  &  96  &  94  &  6  &  3  &  0.398 \\
1987.1   1987.2   1987.3  &  124  &  121  &  8  &  6  &  0.140 \\
1986.4   1987.1   1987.2  &  154  &  152  &  9  &  5  &  0.406 \\
1986.3   1986.4   1987.1  &  129  &  127  &  7  &  6  &  0.022 \\
1986.3   1986.4   1986.2  &  125  &  120  &  9  &  9  &  0.058 \\
[-7pt]\\

\bottomrule
\end{tabular}}
\label{tab:public.nomip.nohet}
\vspace{-\baselineskip}
\floatnote{Number of total couples, rational couples according to aggregate CARP, total singles and rational singles according to GARP, for different combinations of periods. Sampling for the ECPF is quarterly, for which we use the \texttt{year.quarter} notation.}
\vspace{-\baselineskip}
\end{table}
 }

{\singlespacing \begin{table}[ht!]
\caption{Results Conditional on Demographics for Public Goods, Exogenous Prices}
\centering
\singlespacing
\small
\renewcommand{\arraystretch}{0.9}
\resizebox{\textwidth}{!}{
\begin{tabular}{cccccccc}
\multicolumn{8}{c}{\textsc{Longitudinal Internet studies for the Social Sciences (LISS)}} \\
\toprule
Years  &  College  &  Age  &  $ N_{\text{couples}}^{\text{total}} $  &  $ N_{\text{couples}}^{\text{rational}} $  &  $ N_{\text{singles}}^{\text{total}} $  &  $ N_{\text{singles}}^{\text{rational}} $  &  p-value \\
\midrule
2009   2010   2012  &  1  &  2  &  167  &  157  &  72  &  64  &  0.070 \\
2009   2010   2012  &  1  &  1  &  373  &  363  &  83  &  69  &  0.174 \\
2009   2010   2012  &  1  &  0  &  183  &  180  &  50  &  32  &  0.900 \\
2009   2010   2012  &  0  &  2  &  239  &  232  &  105  &  92  &  0.050 \\
2009   2010   2012  &  0  &  1  &  237  &  230  &  62  &  50  &  0.086 \\
2009   2010   2012  &  0  &  0  &  60  &  55  &  18  &  13  &  0.456 \\
[-7pt]\\

\multicolumn{8}{c}{\textsc{Russian Longitudinal Monitoring Survey (RLMS)}} \\
\toprule
Years  &    &  Age  &  $ N_{\text{couples}}^{\text{total}} $  &  $ N_{\text{couples}}^{\text{rational}} $  &  $ N_{\text{singles}}^{\text{total}} $  &  $ N_{\text{singles}}^{\text{rational}} $  &  p-value \\
\midrule
2012   2013   2014  &    &  2  &  70  &  68  &  131  &  122  &  0.024 \\
2012   2013   2014  &    &  1  &  188  &  180  &  123  &  114  &  0.014 \\
2012   2013   2014  &    &  0  &  49  &  44  &  33  &  31  &  0.220 \\
2011   2013   2014  &    &  2  &  66  &  65  &  116  &  108  &  0.002 \\
2011   2013   2014  &    &  1  &  184  &  180  &  117  &  109  &  0.050 \\
2011   2013   2014  &    &  0  &  50  &  50  &  34  &  29  &  0.006 \\
2011   2012   2014  &    &  2  &  65  &  65  &  117  &  108  &  0.056 \\
2011   2012   2014  &    &  1  &  184  &  178  &  117  &  110  &  0.030 \\
2011   2012   2014  &    &  0  &  49  &  49  &  34  &  30  &  0.048 \\
2011   2012   2013  &    &  2  &  73  &  71  &  128  &  117  &  0.070 \\
2011   2012   2013  &    &  1  &  183  &  174  &  124  &  114  &  0.036 \\
2011   2012   2013  &    &  0  &  55  &  52  &  41  &  37  &  0.004 \\
2010   2013   2014  &    &  2  &  50  &  48  &  84  &  76  &  0.072 \\
2010   2013   2014  &    &  1  &  145  &  140  &  87  &  81  &  0.192 \\
2010   2013   2014  &    &  0  &  48  &  48  &  33  &  30  &  0.064 \\
2010   2012   2014  &    &  2  &  49  &  46  &  86  &  76  &  0.060 \\
2010   2012   2014  &    &  1  &  146  &  143  &  94  &  89  &  0.246 \\
2010   2012   2014  &    &  0  &  49  &  49  &  30  &  26  &  0.034 \\
2010   2012   2013  &    &  2  &  55  &  53  &  98  &  87  &  0.030 \\
2010   2012   2013  &    &  1  &  143  &  135  &  93  &  84  &  0.050 \\
2010   2012   2013  &    &  0  &  52  &  49  &  36  &  33  &  0.012 \\
2010   2011   2013  &    &  2  &  56  &  54  &  96  &  86  &  0.020 \\
2010   2011   2013  &    &  1  &  145  &  138  &  95  &  86  &  0.022 \\
2010   2011   2013  &    &  0  &  52  &  52  &  35  &  29  &  0.000 \\
2010   2011   2012  &    &  2  &  72  &  70  &  115  &  105  &  0.060 \\
2010   2011   2012  &    &  1  &  156  &  151  &  100  &  94  &  0.030 \\
2010   2011   2012  &    &  0  &  55  &  55  &  39  &  35  &  0.100 \\
2005   2011   2012  &    &  2  &  48  &  45  &  99  &  85  &  0.018 \\
2005   2011   2012  &    &  1  &  114  &  110  &  72  &  65  &  0.124 \\
2005   2011   2012  &    &  0  &  88  &  78  &  37  &  31  &  0.306 \\
[-7pt]\\

\bottomrule
\end{tabular}}
\label{tab:public.nomip.het}
\vspace{-\baselineskip}
\floatnote{Number of total couples, rational couples according to aggregate CARP, total singles and rational singles according to GARP, for different combinations of periods and demographics.}
\vspace{-\baselineskip}
\end{table}
 }

{\singlespacing \begin{table}[ht!]
\caption{Results with Endogenous Expenditure}
\centering
\singlespacing
\small
\renewcommand{\arraystretch}{0.9}
\resizebox{\textwidth}{!}{
\begin{tabular}{cccccccc}
\multicolumn{8}{c}{\textsc{Longitudinal Internet studies for the Social Sciences (LISS): Private}} \\
\toprule
Years  &  \#  &  Min  &  25\textsuperscript{th} Quantile  &  Median  &  Mean  &  75\textsuperscript{th} Quantile  &  Max
\\
\midrule
2009   2010   2012  &  9  &  0.000  &  0.005  &  0.016  &  0.122  &  0.183  &  0.532\\
[-7pt]\\

\multicolumn{8}{c}{\textsc{Longitudinal Internet studies for the Social Sciences (LISS): Public}} \\
\toprule
Years  &  \#  &  Min  &  25\textsuperscript{th} Quantile  &  Median  &  Mean  &  75\textsuperscript{th} Quantile  &  Max
\\
\midrule
2009   2010   2012  &  7  &  0.002  &  0.010  &  0.556  &  0.436  &  0.649  &  0.922\\
[-7pt]\\

\bottomrule
\end{tabular}}
\label{tab:ec.nomip.nohet}
\vspace{-\baselineskip}
\floatnote{Summary statistics of p-values for different combinations of periods with control function $ v $ evaluated at $ v_0 $. The control function is estimated with income level acting as an instrument for total consumption.}
\vspace{-\baselineskip}
\end{table}
 }

{\singlespacing \begin{table}[ht!]
\caption{Results for Mixed Integer Procedure}
\centering
\singlespacing
\small
\renewcommand{\arraystretch}{0.9}
\resizebox{\textwidth}{!}{
\begin{tabular}{cccccccc}
\multicolumn{8}{c}{\textsc{LISS (private)}} \\
\toprule
Years  &    &    &  $ N_{\text{couples}}^{\text{total}} $  &  $ N_{\text{couples}}^{\text{rational}} $  &  $ N_{\text{singles}}^{\text{total}} $  &  $ N_{\text{singles}}^{\text{rational}} $  &  p-value \\
\midrule
2009   2010   2012  &    &    &  605  &  605  &  463  &  380  &  0.010 \\
[-7pt]\\

\multicolumn{8}{c}{\textsc{LISS (public)}} \\
\toprule
Years  &    &    &  $ N_{\text{couples}}^{\text{total}} $  &  $ N_{\text{couples}}^{\text{rational}} $  &  $ N_{\text{singles}}^{\text{total}} $  &  $ N_{\text{singles}}^{\text{rational}} $  &  p-value \\
\midrule
2009   2010   2012  &    &    &  1259  &  1259  &  390  &  320  &  0.000 \\
[-7pt]\\

\multicolumn{8}{c}{\textsc{LISS (private, conditional on demographics)}} \\
\toprule
Years  &  College  &  Age  &  $ N_{\text{couples}}^{\text{total}} $  &  $ N_{\text{couples}}^{\text{rational}} $  &  $ N_{\text{singles}}^{\text{total}} $  &  $ N_{\text{singles}}^{\text{rational}} $  &  p-value \\
\midrule
2009   2010   2012  &  1  &  2  &  92  &  92  &  79  &  71  &  0.084 \\
2009   2010   2012  &  1  &  1  &  163  &  163  &  93  &  75  &  0.064 \\
2009   2010   2012  &  1  &  0  &  65  &  65  &  64  &  56  &  0.000 \\
2009   2010   2012  &  0  &  2  &  137  &  137  &  121  &  97  &  0.000 \\
2009   2010   2012  &  0  &  1  &  114  &  114  &  85  &  65  &  0.000 \\
2009   2010   2012  &  0  &  0  &  34  &  34  &  21  &  16  &  0.000 \\
[-7pt]\\

\multicolumn{8}{c}{\textsc{LISS (public, conditional on demographics)}} \\
\toprule
Years  &  College  &  Age  &  $ N_{\text{couples}}^{\text{total}} $  &  $ N_{\text{couples}}^{\text{rational}} $  &  $ N_{\text{singles}}^{\text{total}} $  &  $ N_{\text{singles}}^{\text{rational}} $  &  p-value \\
\midrule
2009   2010   2012  &  1  &  2  &  167  &  167  &  72  &  64  &  0.066 \\
2009   2010   2012  &  1  &  1  &  373  &  373  &  83  &  69  &  0.014 \\
2009   2010   2012  &  1  &  0  &  183  &  183  &  50  &  32  &  0.034 \\
2009   2010   2012  &  0  &  2  &  239  &  239  &  105  &  92  &  0.044 \\
2009   2010   2012  &  0  &  1  &  237  &  237  &  62  &  50  &  0.054 \\
2009   2010   2012  &  0  &  0  &  60  &  60  &  18  &  13  &  0.154 \\
[-7pt]\\

\bottomrule
\end{tabular}}
\label{tab:mixedinteger}
\vspace{-\baselineskip}
\floatnote{Mixed Integer \citet{Cherchye2011}}
\vspace{-\baselineskip}
\end{table}
 }

 \clearpage
\section{Descriptive Statistics}\label{sec:descriptives}
{\singlespacing 
  \begin{table}[H]
\caption{Descriptive Statistics (Private Goods)}
\resizebox{0.95\columnwidth}{!}{
{\centering
\begin{tabular}{ccccccccccc}
\multicolumn{11}{c}{\textsc{Longitudinal Internet Studies for the Social Sciences (LISS)}} \\
\toprule
& & \multicolumn{3}{c}{Food out} & \multicolumn{3}{c}{Clothing} & \multicolumn{3}{c}{Leisure} \\
Year & $ N $ & Mean & IQR & P & Mean & IQR & P & Mean & IQR & P \\ \midrule
2009 & 5594 & 43.463 & 45.0 & 108.1 & 69.825 & 60.0 & 107.4 & 20.121 & 20.0 & 98.6 \\
2010 & 5337 & 38.798 & 50.0 & 109.7 & 71.835 & 75.0 & 105.3 & 23.362 & 20.0 & 98.0 \\
2012 & 5463 & 40.611 & 50.0 & 113.4 & 74.347 & 80.0 & 106.6 & 23.743 & 25.0 & 100.4 \\[-10pt]
\\
\multicolumn{11}{c}{\textsc{Russian Longitudinal Monitoring Survey (RLMS)}} \\
\toprule
& & \multicolumn{3}{c}{Dairy} & \multicolumn{3}{c}{Bread} & \multicolumn{3}{c}{Meat} \\
Year & $ N $ & Mean & IQR & P & Mean & IQR & P & Mean & IQR & P \\ \midrule
2000 & 1506 &  81.7 & 120.4 & 121.1 & 113.2 & 102.8 & 116.5 & 322.0 & 439.7 & 128.3
\\
2005 & 1601 & 222.3 & 277.8 & 110.5 & 199.3 & 171.8 & 103.0 & 1003.5 & 1175.8 & 118.6
\\
2010 & 2839 & 475.7 & 492.7 & 116.7 & 303.3 & 270.8 & 107.6 & 1862.1 & 1926.0 & 105.3
\\
2011 & 2983 & 520.3 & 544.8 & 106.3 & 317.3 & 270.9 & 108.9 & 2165.8 & 2154.7 & 109.2
\\
2012 & 3154 & 551.0 & 550.5 & 104.4 & 330.2 & 291.7 & 112.0 & 2284.8 & 2315.0 & 108.3
\\
2013 & 3076 & 617.7 & 622.9 & 113.1 & 352.9 & 287.1 & 108.0 & 2366.2 & 2399.8 &  97.0
\\
2014 & 2516 & 695.3 & 675.8 & 114.4 & 372.9 & 323.0 & 107.5 & 2805.6 & 2847.5 & 102.1
\\[-10pt]
\\
\multicolumn{11}{c}{\textsc{Spanish Continuous Family Expenditure Survey (ECPF)}} \\
\toprule
& & \multicolumn{3}{c}{Clothing} & \multicolumn{3}{c}{Food out} & \multicolumn{3}{c}{Nondurables} \\
Year & $ N $ & Mean & IQR & P & Mean & IQR & P & Mean & IQR & P \\ \midrule
1985 & 65 & 1284.4 & 1406.4 & 165.7 & 940.9 & 899.7 & 174.3 &  35.1 &  43.8 & 150.6
\\
1986 & 95 & 1334.8 & 1545.3 & 191.6 & 927.3 & 1128.1 & 224.7 &  28.8 &  47.6 & 164.5
\\
1987 & 288 & 1743.3 & 1897.1 & 174.2 & 1054.7 & 1244.8 & 191.5 &  37.0 &  53.8 & 157.2
\\
1988 & 195 & 1537.6 & 1831.0 & 158.1 & 1036.2 & 1400.9 & 160.0 &  42.6 &  53.3 & 145.8
\\
1989 & 225 & 2253.9 & 2344.4 & 134.3 & 1446.3 & 1685.6 & 139.6 &  57.2 &  70.5 & 140.7
\\
1990 & 205 & 2289.9 & 2565.5 & 106.6 & 1636.4 & 2152.1 & 112.2 &  41.0 &  53.2 & 101.9
\\
1991 & 210 & 2255.2 & 2398.5 & 183.5 & 1852.4 & 2229.4 & 208.7 &  52.5 &  66.7 & 160.8
\\
1992 & 202 & 2652.5 & 2795.1 & 154.5 & 1852.6 & 1957.9 & 154.5 &  69.1 &  85.4 & 144.0
\\
1993 & 185 & 2823.0 & 2471.3 & 112.4 & 2386.2 & 3022.2 & 121.0 &  75.7 &  80.8 & 114.5
\\
1994 & 210 & 2102.8 & 2471.0 & 106.4 & 2322.7 & 2730.5 & 111.2 &  79.8 &  97.4 & 102.9
\\
1995 & 194 & 2186.9 & 2287.9 & 113.6 & 2068.9 & 2187.8 & 122.2 & 117.8 & 118.3 & 114.4
\\
1996 & 199 & 2397.4 & 2595.2 & 126.5 & 2761.4 & 3230.4 & 126.6 & 107.5 & 129.6 & 129.5
\\
\bottomrule
\end{tabular}
 }}
\label{tab:descriptive.onebig.priv}
\vspace{-\baselineskip}
\floatnote{Descriptive statistics of the LISS, RMLS and ECPF reporting mean, interquantile range (IQR) and price index P. LISS quantities consumed per month are inflated to 2005 prices and denoted in Euro (source: Eurostat \url{http://www.ecb.europa.eu/stats/prices/hicp/html/hicp\_coicop\_inx\_index.en.html}). RMLS quantities are per week and inflated to 2014 prices and denoted in local currency (Russian Ruble). Goods are aggregated to composite good categories as follows. Dairy: Canned/powdered milk, fresh milk, sour milk products and sour cream; Bread: White (wheat) bread and black (rye) bread; Meat: Canned meat, beef/veal, lamb/goat, pork, giblets, poultry, lard, sausage and semi-prepared meat products. ECPF consumption is per week with quarterly collection frequency. We only report descriptive statistics of the first quarter of a given year. ECPF quantities are normalized to arbitrary units using the price indices P.}
\end{table}}

{\singlespacing
\begin{table}[htbp]
\caption{Descriptive Statistics (Public Goods)}
\resizebox{0.95\columnwidth}{!}{
{\centering
\begin{tabular}{ccccccccccc}
\multicolumn{11}{c}{\textsc{Longitudinal Internet Studies for the Social Sciences (LISS)}} \\
\toprule
& & \multicolumn{3}{c}{Housing} & \multicolumn{3}{c}{Transport} & \multicolumn{3}{c}{Energy} \\
Year & $ N $ & Mean & IQR & P & Mean & IQR & P & Mean & IQR & P \\ \midrule
2009 & 5594 & 590.6 & 460.0 & 108.1 & 141.1 & 150.0 & 107.4 & 282.0 & 123.0 & 98.6 \\
2010 & 5337 & 600.1 & 449.5 & 109.7 & 135.2 & 150.0 & 105.3 & 210.4 & 125.5 & 98.0 \\
2012 & 5463 & 577.4 & 475.0 & 113.4 & 148.2 & 150.0 & 106.6 & 219.1 & 116.0 & 100.4 \\[-10pt]
\\
\multicolumn{11}{c}{\textsc{Russian Longitudinal Monitoring Survey (RMLS)}} \\
\toprule
& & \multicolumn{3}{c|}{Housing} & \multicolumn{3}{c|}{Transport} & \multicolumn{3}{c}{Energy} \\
Year & $ N $ & Mean & IQR & P & Mean & IQR & P & Mean & IQR & P \\ \midrule
2000 & 1506 & 4388.4 & 6051.9 & 116.3 & 2942.3 & 670.7 &  29.0 & 3454.2 &   0.0 &   5.7 \\
2005 & 1601 & 7961.1 & 9264.3 & 118.0 & 2996.0 & 2282.5 & 117.9 & 3808.0 & 1280.2 &  18.0 \\
2010 & 2839 & 12336.3 & 11699.2 & 102.7 & 3697.7 & 2437.6 & 240.0 & 3976.4 & 2226.7 &  43.8 \\
2011 & 2983 & 12324.6 & 11744.0 & 112.1 & 3697.0 & 2500.1 & 234.0 & 4496.1 & 2500.1 &  48.3 \\
2012 & 3154 & 12349.8 & 11444.4 & 112.1 & 3799.7 & 2699.0 & 250.5 & 4413.2 & 3922.6 &  55.3 \\
2013 & 3076 & 13018.6 & 11098.5 & 103.6 & 4109.9 & 2591.4 & 270.9 & 4958.9 & 4311.5 &  63.6 \\
2014 & 2516 & 13576.1 & 11050.0 & 105.1 & 4169.0 & 2550.0 & 275.3 & 5238.0 & 5100.0 &  63.8 \\[-10pt]
\\
\multicolumn{11}{c}{\textsc{Spanish Continuous Family Expenditure Survey (ECPF)}} \\
\toprule
& & \multicolumn{3}{c|}{Clothing} & \multicolumn{3}{c|}{Transport} & \multicolumn{3}{c}{Petrol} \\
Year & $ N $ & Mean & IQR & P & Mean & IQR & P & Mean & IQR & P \\ \midrule
1985 & 65 & 1284.4 & 1406.4 & 165.7 & 553.0 & 942.6 & 188.9 & 183.4 &  41.8 &  92.7 \\
1986 & 95 & 1334.8 & 1545.3 & 191.6 & 647.9 & 1002.6 & 277.0 & 120.8 &   0.0 & 112.4 \\
1987 & 288 & 1743.3 & 1897.1 & 174.2 & 612.5 & 1000.7 & 213.2 & 190.3 &  48.6 &  98.2 \\
1988 & 195 & 1537.6 & 1831.0 & 158.1 & 658.1 & 1050.2 & 154.2 & 244.0 & 172.8 &  86.4 \\
1989 & 225 & 2253.9 & 2344.4 & 134.3 & 681.2 & 1120.6 & 119.6 & 397.2 & 314.0 &  95.1 \\
1990 & 205 & 2289.9 & 2565.5 & 106.6 & 617.2 & 1017.1 & 115.1 & 451.4 & 254.2 & 114.6 \\
1991 & 210 & 2255.2 & 2398.5 & 183.5 & 999.5 & 1252.6 & 244.6 & 493.3 & 398.2 & 105.0 \\
1992 & 202 & 2652.5 & 2795.1 & 154.5 & 838.8 & 1307.8 & 147.2 & 459.3 & 437.0 &  90.6 \\
1993 & 185 & 2823.0 & 2471.3 & 112.4 & 1099.0 & 1631.4 & 119.9 & 481.2 & 318.0 & 122.4 \\
1994 & 210 & 2102.8 & 2471.0 & 106.4 & 1101.0 & 1559.4 & 114.4 & 752.0 & 959.9 & 113.4 \\
1995 & 194 & 2186.9 & 2287.9 & 113.6 & 1005.7 & 1470.0 & 119.4 & 497.9 & 561.0 & 125.1 \\
1996 & 199 & 2397.4 & 2595.2 & 126.5 & 1302.5 & 1954.3 & 112.1 & 667.5 & 938.7 & 106.3 \\
\bottomrule
\end{tabular}
 }}
\label{tab:descriptive.onebig.pub}
\vspace{-\baselineskip}
\floatnote{Descriptive statistics of the LISS, RMLS and ECPF reporting mean, interquantile range (IQR) and price index P. LISS quantities consumed per month are inflated to 2005 prices (CPI and HPI) and denoted in Euro (source: Eurostat \url{http://www.ecb.europa.eu/stats/prices/hicp/html/hicp\_coicop\_inx\_index.en.html}). RMLS quantities are per week and inflated to 2014 prices and denoted in local currency (Russian Ruble). Goods are aggregated to composite good categories as follows. Transport: Transportation services, running costs for cars (excluding fuel) and Energy: Fuel, Gas, Coal and Firewood. ECPF consumption is per week with quarterly collection frequency. We only report descriptive statistics of the first quarter of a given year. ECPF quantities are normalized to arbitrary units using the price indices P. We chose a combination of private and public goods due to the limited availability of the latter.}
\end{table}}

 \end{document}